\documentclass{article}
\pdfoutput=1

\usepackage[T1]{fontenc}
\usepackage{amsmath,amssymb,amsthm,amscd}
\usepackage{graphicx}
\usepackage{dsfont}
\usepackage{lmodern}
\usepackage{mathrsfs}
\usepackage{wasysym}
\usepackage[indent]{parskip}
\usepackage{tabularray}
\usepackage[strict]{changepage}

\UseTblrLibrary{diagbox}
\UseTblrLibrary{amsmath}

\usepackage[backref=page]{hyperref}
\renewcommand*{\backref}[1]{}
\renewcommand*{\backrefalt}[4]{%
 \ifcase #1 %
  \relax%
 \or
  ~{\small [\textsc{p.~\footnotesize{\!#2}}]}%
 \else
  ~{\small [\textsc{p.~\footnotesize{\!#2}}]}%
 \fi
}

\usepackage{footnotebackref}
\usepackage[capitalise]{cleveref}

\usepackage{authblk}

\numberwithin{equation}{section}

\def\be{\begin{equation}}
\def\ee{\end{equation}}
\def\ba{\begin{aligned}}
\def\ea{\end{aligned}}

\newcommand{\mathd}{\mathrm{d}}
\newcommand{\mathe}{\mathrm{e}}

\newcommand{\GBE}{G$\beta$E}
\newcommand{\WLBE}{WL$\beta$E}
\newcommand{\Tr}{\mathrm{Tr}}
\newcommand{\Aut}{\mathrm{Aut}}
\newcommand{\genus}{z}

\newcommand{\bu}{\mathbf{u}}
\newcommand{\bv}{\mathbf{v}}

\newcommand{\bx}{\mathbf{x}}
\newcommand{\by}{\mathbf{y}}
\newcommand{\bz}{\mathbf{z}}

\theoremstyle{plain}
\newtheorem{lemma}{Lemma}[section]
\newtheorem{ansatz}{Ansatz}[section]
\newtheorem{proposition}{Proposition}[section]
\newtheorem{theorem}{Theorem}[section]

\theoremstyle{definition}
\newtheorem{definition}{Definition}[section]

\theoremstyle{remark}
\newtheorem{remark}{Remark}[section]

\begin{document}

\NewTblrTheme{fancy}{
\SetTblrStyle{head}{font=\normalsize}
\SetTblrStyle{foot}{font=\normalsize}
\SetTblrStyle{caption-tag}{font=\scshape}
\SetTblrStyle{conthead}{\itshape}
\SetTblrStyle{contfoot}{\itshape}
}

\title{$\beta$-Ensembles and higher genera Catalan numbers}

\author[a]{Luca Cassia\footnote{\texttt{luca.cassia@unimelb.edu.au}}}
\author[b]{Vera Posch\footnote{\texttt{poschv@tcd.ie}}}
\author[c]{Maxim Zabzine\footnote{\texttt{maxim.zabzine@physics.uu.se}}}
\affil[a]{School of Mathematics and Statistics, The University of Melbourne Parkville, Melbourne, VIC 3010, Australia}
\affil[b]{School of Mathematics and Hamilton Mathematics Institute, Trinity College Dublin, Dublin, Ireland}
\affil[c]{Department of Physics and Astronomy, Uppsala University, Box 516, 75120 Uppsala, Sweden}

\date{ }

\maketitle

\begin{abstract}
We propose formulas for the large $N$ expansion of the generating function of connected correlators of the $\beta$-deformed Gaussian and Wishart--Laguerre matrix models.
We show that our proposal satisfies the known transformation properties under the exchange of $\beta$ with $1/\beta$ and, using Virasoro constraints, we derive a recursion formula for the coefficients of the expansion.
In the undeformed limit $\beta=1$, these coefficients are integers and they have the combinatorial interpretation of generalized Catalan numbers.
For generic $\beta$, we define the higher genus Catalan polynomials $C_{g,\nu}(\beta)$ whose coefficients are integer numbers.
\end{abstract}

\tableofcontents

\parskip=4pt

\section{Introduction}

It is by now a well-known fact that the generating functions of connected correlators of matrix models often admit a topological expansion \cite{Borot:2013asy} which contains information about some enumerative geometric problems, such as map enumeration, Hurwitz theory, intersection theory on moduli spaces and Gromov--Witten theory \cite{Ekedahl:2001hur,Okounkov:2000gx,Okounkov:2002cja,Eynard:2016yaa}.

On the other hand, matrix models satisfy Virasoro (or W-algebra) constraints associated to reparametrization invariance of the integrals under certain infinitesimal deformations \cite{MIRONOV199047}. Such constraints arise as Ward identities for the correlations functions and can be recast in the form of linear differential equations for the formal generating function of all correlation functions.
The generating function of connected correlators of the matrix model, being the logarithm of the generating function of all correlators, satisfies a related set of nonlinear equations which follow directly from Virasoro constraints.

In this article, we will be studying the interplay between these two aspects of the theory of random matrix models, namely, we will be interested in exploiting Virasoro constraints to derive the topological expansion of the generating function of connected correlators.
By a slight abuse of notation, we will refer to this generating function as the \emph{time-dependent free energy} which is a formal power series in higher times.

More specifically, we consider $\beta$-deformed ensembles of random Hermitian matrices with polynomial potentials,
\be
\label{eq:generating-function-intro}
 Z(N,\beta,\lambda,\bu) =
 \frac1{N!} \int \prod_{i=1}^N \mathd x_i \prod_{i<j} |x_i-x_j|^{2\beta}
 \mathe^{-\sum_{i=1}^NV(x_i)+\frac1{N}\sum_{k=1}^\infty u_k\sum_{i=1}^N x_i^k}
\ee
where $V(x)=\frac{N}{m\lambda}x^m$ and $\beta$ is an arbitrary complex deformation parameter which allows to interpolate between the various types of standard ensembles such as orthogonal ($\beta=1/2$), unitary ($\beta=1$) and symplectic ($\beta=2$) \cite{Mehta:1990,Forrester:2010,Akemann:10.1093}.
This generating function provides an highest weight representation of the Virasoro algebra whose generators are represented by differential operators in the \emph{higher times} variables $\bu=\{u_1,u_2,\dots\}$, such that \eqref{eq:generating-function-intro} is annihilated by a parabolic subalgebra,
\be
 \left(
 \frac{N^2}{\lambda} \frac{\partial}{\partial u_{n+m}}
 - L_n\right) Z(N,\beta,\lambda,\bu) = 0\, , \qquad n \geq 1-m
\ee
for certain operators $L_n$ defined as in \eqref{eq:virasoro_generators}.
In the case of $m=1,2$, the constraints admit a unique solution which can be recast either in the form of exponentials of W-operators \cite{Morozov:2009xk,Cassia:2020uxy} or in the form of \emph{superintegrability} formulas for Jack polynomials \cite{Cassia:2020uxy,Mironov:2022fsr,Mishnyakov:2022bkg,Bawane:2022cpd}.

We restrict ourselves to considering these two cases, which correspond to the Gaussian $\beta$-ensemble (\GBE{}) for $m=2$ and the Wishart--Laguerre $\beta$-ensemble (\WLBE{}) for $m=1$. In order to derive the topological expansion of the time-dependent free energy,
we make use of the following key ingredients:
\begin{itemize}
\item nonlinear Virasoro constraints for the generating function
\begin{equation*}
 F(N,\beta,\lambda,\bu):=\log \frac{Z(N,\beta,\lambda,\bu)}{Z(N,\beta,\lambda,0)}\,;
\end{equation*}
\item superintegrability formulas for the average of characters \cite{Mironov:2022fsr}, i.e., the property of some matrix models that
\begin{equation*}
 \langle\mathrm{Jack}_\mu\rangle\sim\mathrm{Jack}_\mu \,;
\end{equation*}
\item symmetries of the constraint equations under the involution that sends $\beta$ to $1/\beta$.
\end{itemize}
Combining these properties of the matrix model, we are led to the following \emph{ansatz} for the time-dependent free energy,
\begin{multline}
 F(N,\beta,\lambda,\bu)
 = \sum_{\ell=1}^\infty \sum_{g\in\frac12\mathbb{N}}
 \frac{N^{2-2g-2\ell} \beta^{1-\ell-2g}}{\ell!}
 \sum_{k_1,\dots,k_\ell=1}^\infty (\beta\lambda)^{\sum_{j=1}^\ell \frac{k_j}{m}} \\
 \times\sum_{i_1+i_2=2g}(-\beta)^{i_1} C_{g,[k_1,\dots,k_\ell]}^{(i_1,i_2)}
 \prod_{j=1}^\ell u_{k_j}
\end{multline}
where the sum over half-integers $g$ is interpreted as a \emph{genus expansion}.
Correspondingly, the coefficients $C_{g,[k_1,\dots,k_\ell]}^{(i_1,i_2)}$ take the role of enumerative invariants associated to the $\beta$-deformed model. By analogy with the undeformed case, we define the sum
\be
 C_{g,[k_1,\dots,k_\ell]}(\beta) :=
 \sum_{i_1+i_2=2g}(-\beta)^{i_1} C_{g,[k_1,\dots,k_\ell]}^{(i_1,i_2)}
\ee
to be the \emph{Catalan polynomial} of genus $g$ associated to the tuple $[k_1,\dots,k_\ell]$. Plugging the ansatz into the Virasoro constraint equations gives a set of (\emph{cut-and-join}) recursion relations for the polynomials $C_{g,[k_1,\dots,k_\ell]}(\beta)$ which can be solved uniquely. Remarkably, the coefficients $C_{g,[k_1,\dots,k_\ell]}^{(i_1,i_2)}$ are all \emph{integer} numbers and they provide a refinement of the ordinary higher genus Catalan numbers described in \cite{WALSH1972192,Carlet:2020enf}.

The organization of the article is as follows.
\begin{itemize}
\item In \cref{sec:GBE,sec:virasoro,sec:symmetry,sec:genus,sec:recursion,sec:limit} we give the definition of the generating function of (connected) correlators of the \GBE{} as a function of higher times, the rank $N$, the coupling $\lambda$ and the deformation parameter $\beta$.
We derive the Virasoro constraints satisfied by the time-dependent free energy, and we observe the symmetry of these objects under the involution that exchanges $\beta$ and $1/\beta$.
Making use of these properties, we then derive an ansatz for the $1/N$ expansion of the time-dependent free energy, and we identify the coefficients $C_{g,\nu}(\beta)$ as $\beta$-deformations of the higher genus Catalan numbers that appear in the topological expansion of the ordinary Gaussian matrix model. We show that the $\beta$-dependence of the Catalan polynomials $C_{g,\nu}(\beta)$ can be reabsorbed into Schur polynomials of two variables $s_\lambda(1,-\beta)$ and this leads to the definition of a secondary set of integer invariants, the $n_{\nu,\lambda}$.
A cut-and-join recursion formula for the Catalan polynomials is then obtained from Virasoro constraints and the undeformed limit $\beta\to1$ is discussed.
\item In \cref{sec:WLBE} we repeat the analysis for the case of linear potential, i.e., the \WLBE{} matrix model, and we obtain similar formulas for the genus expansion and the recursion relations.
\item In \cref{sec:outlook} we comment on our results and identify some of the open questions that deserve further investigation.
\end{itemize}
Finally, in \cref{sec:Schur} we collect some useful facts about Schur polynomials in two variables, in \cref{sec:hypermaps} we provide a formula to relate the Catalan polynomials of the \GBE{} to the marginal $b$-polynomials that appear in the $b$-conjectures of Goulden and Jackson, and in \cref{sec:tables} we tabulate some of the polynomials $C_{g,\nu}(\beta)$ and integer invariants $n_{\nu,\lambda}$ obtained by solving the recursion up to finite order in the genus and in higher times.

\subsubsection*{Acknowledgements}
We are grateful to Peter J.\ Forrester, Anas A.\ Rahman and Marvin Anas Hahn for useful comments on an earlier draft of this manuscript.
L.C.\ was supported by the ARC Discovery Grant DP210103081.
V.P.\ has received funding from the European Research Council (ERC) under
the European Union’s Horizon 2020 research and innovation programme
(Grant Agreement 101088193).
M.Z.\ was supported by the VR excellence center grant ``Geometry and Physics'' 2022-06593.

\section{The \texorpdfstring{$\beta$}{beta}-deformed Gaussian ensemble}
\label{sec:GBE}

Recall that the classical Gaussian unitary ensemble (GUE) is defined by the matrix integral
\be
 \int_{N\times N} \mathd M ~ \mathe^{-\frac{N}{2\lambda}\Tr M^2}
\ee
where the integration is over the space of Hermitian $N\times N$ matrices and the parameter $N$ is the rank.
Here we consider a quadratic potential with coupling constant $\lambda$ such that $\Re(\lambda)>0$.

The measure and the potential are invariant under the adjoint action of the group $U(N)$, therefore, it is possible to rewrite this matrix integral as an integral over eigenvalues $x_i$ as
\be
\label{eq:integral-representation}
 \frac1{N!}\int_{\mathbb{R}^N} \prod_{i=1}^N \mathd x_i ~ \prod_{i<j}|x_i-x_j|^2
 ~ \mathe^{ - \frac{N}{2\lambda} \sum_i x_i^2}
\ee
where $\prod_{i<j}(x_i-x_j)$ is the Vandermonde determinant.
The $\beta$-deformation of the GUE is defined as a 1-parameter deformation of the eigenvalue integral by substituting the Vandermonde determinant with its $\beta$ power. This is a very natural and well-studied deformation which is known to have a matrix integral representation via tridiagonal matrices as shown in \cite{Dumitriu:2002mat}. We refer to this matrix model as the Gaussian $\beta$-deformed ensemble.

The generating function of all polynomial expectation values is defined by the formal power series in higher times $\bu=\{u_1,u_2,\dots\}$ as
\be
\label{eq:generating-function}
 Z(N,\beta,\lambda,\bu) := \frac1{N!}\int_{\mathbb{R}^N} \prod_{i=1}^N \mathd x_i ~ \prod_{i<j}|x_i-x_j|^{2\beta} ~
 \mathe^{ - \frac{N}{2\lambda} \sum_i x_i^2
 + \frac{1}{N} \sum_{k=1}^\infty u_k \sum_i x_i^k}
\ee
This is the main object that we will study in the following. Its logarithm is the generating function of all connected correlation functions and is known as the time-dependent free energy of the matrix model. We denote this function as
\be
 F(N,\beta,\lambda,\bu) := \log \frac{Z(N,\beta,\lambda,\bu)}{Z(N,\beta,\lambda,0)}
\ee
where we normalized the time-dependent free energy so that $F(N,\beta,\lambda,0)=0$.

It is a well-known fact that the time-dependent free energy of a matrix model admits a simultaneous expansion in the parameter $N$ and higher times $u_k$ which can be interpreted as a genus expansion, i.e., a sum over contributions coming from surfaces of arbitrary genus $g$ and punctures.
In \cref{sec:genus} we show that a similar expansion also exists after the $\beta$-deformation.

We conclude this section by observing that the dependence on the coupling parameter $\lambda$ is quite simple, and it can be easily reabsorbed via a rescaling of times, as we show in the following lemma.
\begin{lemma}
\label{lm:homogeneity}
The function $F(N,\beta,\lambda,\bu)$ satisfies the following homogeneity equation
\be
 F(N,\beta,\lambda,\bu) = \lambda^{\frac{D_{\bu}}{2}} F(N,\beta,1,\bu)
\ee
with the dilation operator defined as
\be
\label{eq:dilatation}
 D_{\bu} = \sum_{k\geq1} k u_k \frac{\partial}{\partial u_k}~.
\ee
\end{lemma}
\begin{proof}
This follows from the change of variables $x_i=\lambda^{\frac12}y_i$ in the integral \eqref{eq:integral-representation}.
\end{proof}

\section{Virasoro constraints}
\label{sec:virasoro}

The generating function of a matrix model satisfies an infinite set of differential equations known as Virasoro constraints \cite{MIRONOV199047} which encode all the linear relations among the correlation functions. These constraints are a consequence of invariance of the integral under infinitesimal reparametrizations generated by the vector fields $\sum_i x_i^{n+1} \frac{\partial}{\partial x_i}$ for $n\in\mathbb{Z}$ which provide a representation of the Virasoro algebra.

In the case of the $\beta$-ensemble in \eqref{eq:generating-function}, the constraints can be written explicitly as
\be
\label{eq:virasoro_constraints}
 \left(
 \frac{N^2}{\lambda} \frac{\partial}{\partial u_{n+2}}
 - L_n\right) Z(N,\beta,\lambda,\bu) = 0\, , \qquad n \geq - 1
\ee
with the Virasoro generators $L_n$ defined as
\be
\label{eq:virasoro_generators}
\begin{aligned}
 L_{n>0} &= 2\beta N^2 \frac{\partial}{\partial u_n}
 + \beta N^2\sum_{i+j=n} \frac{\partial^2}{\partial u_i \partial u_j}
 + (1-\beta)N (n+1) \frac{\partial}{\partial u_n}
 + \sum_{k>0} k u_k \frac{\partial}{\partial u_{k+n}} \\
 L_0 &= \beta N^2 + (1-\beta) N + \sum_{k>0} k u_k \frac{\partial}{\partial u_k} \\
 L_{-1} &= u_1 +\sum_{k>0} k u_k \frac{\partial}{\partial u_{k-1}}~,
\end{aligned}
\ee
While these constraints are linear for the function $Z(N,\beta,\lambda,\bu)$, they are nonlinear and non-homogeneous for $F(N,\beta,\lambda,\bu)$; in fact, we have
\begin{multline}
\label{eq:virasoro_constraints_F}
 \Big(
   \frac{N^2}{\lambda} \frac{\partial}{\partial u_{n+2}}
 - \left(2\beta N^2+(1-\beta)N(n+1)\right) \frac{\partial}{\partial u_{n}}
 - \beta N^2\sum_{i+j=n}\frac{\partial^2}{\partial u_i\partial u_j}
 - \sum_{k>0} k u_k \frac{\partial}{\partial u_{k+n}}
 \Big) F(N,\beta,\lambda,\bu) \\
 - \beta N^2 \sum_{i+j=n} \frac{\partial F(N,\beta,\lambda,\bu)}{\partial u_i}
 \frac{\partial F(N,\beta,\lambda,\bu)}{\partial u_j}
 = \left(\beta N^2+(1-\beta)N\right)\delta_{n,0} + u_1 \delta_{n,-1}
\end{multline}
Later, we will use these equations to derive a recursion relation for the coefficients of the series expansion of $F(N,\beta,\lambda,\bu)$.

\begin{remark}
The rank $N$ of a matrix model, even $\beta$-deformed, is by definition the number of eigenvalues (integration variables) and therefore it is a positive integer number.
We remark, however, that the Virasoro constraint \cref{eq:virasoro_constraints_F} make sense not just for positive integer rank but can be analytically continued to arbitrary complex values. The corresponding solutions will then interpolate between actual matrix models and more general functions that can be interpreted as analytic continuations away from integer rank $N$. From now on, we will therefore assume that $N$ is just another complex variable on which the coefficients of the generating function $F(N,\beta,\lambda,\bu)$ depend analytically.
\end{remark}

\section{Symmetries of the \texorpdfstring{\GBE{}}{GBE}}
\label{sec:symmetry}

Before discussing the dependence of $F(N,\beta,\lambda,\bu)$ on the parameter $N$ as a power series, we pause to observe that the $\beta$-ensemble has a non-trivial symmetry under the exchange of $\beta$ and $\beta^{-1}$, \cite{Dumitriu_2006,Dumitriu:2012glo,Cassia:2021dpd}. This is not a symmetry which is manifest at the level of the integral representation of the model, and in fact, it is better understood as a non-perturbative duality known as Langlands duality \cite{Cassia:2021dpd}
or, alternatively, as high-low temperature duality for classical $\beta$-ensembles as suggested in \cite{Forrester:2021hig}.
Namely, one notices that the Virasoro constraints \eqref{eq:virasoro_constraints} are invariant under the duality transformation
\be
 \beta \mapsto 1/\beta,
\ee
together with the rescalings
\be
 N \mapsto -\beta N,
 \quad\quad\quad
 u_k \mapsto u_k,
 \quad\quad\quad
 \lambda \mapsto \beta^2\lambda
\ee
Since the solution of the constraints is unique (up to normalization), it must follow that the (normalized) generating function is invariant under this symmetry.
This produces the identity
\be
\label{eq:symmetry-F}
 F(-\beta N,\beta^{-1},\beta^{2}\lambda,\bu) = F(N,\beta,\lambda,\bu)
\ee
for the time-dependent free energy. We will use this identity to constrain the form of $F(N,\beta,\lambda,\bu)$ as a power series in $N$.

In addition to the inversion symmetry of $\beta$, the \GBE{} just like the GUE is symmetric under the transformation $x_i\mapsto-x_i$.
Then it is well-known that there is a corresponding Ward identity that sets all odd correlation functions to zero, namely,
\be
 \frac{\partial}{\partial u_{k_1}} \dots \frac{\partial}{\partial u_{k_n}} Z(N,\beta,\lambda,\bu) = 0
\ee
when $k_1+\dots+k_n$ is odd. Correspondingly, $F(N,\beta,\lambda,\bu)$ must be a formal power series in times such that each monomial is even w.r.t.\ the degree induced by the dilation operator $D_{\bu}$ in \eqref{eq:dilatation}.

\section{Large \texorpdfstring{$N$}{N} behavior and genus expansion}
\label{sec:genus}

The Virasoro constraints of the \GBE{} induce a recursion for the correlation functions, and this recursion is known to have a unique solution \cite{Morozov:2009xk} up to normalization.
Using the superintegrability formula for averages of Jack polynomials \cite{Morozov:2019gbt,Cassia:2020uxy,Cassia:2021dpd,Mishnyakov:2022bkg,Bawane:2022cpd},
\be
 \frac{\langle\mathrm{JackP}_\mu(x_1,\dots,x_N)\rangle_{\text{\GBE}}}{\langle 1 \rangle_{\text{\GBE}}}
 = \frac{ \mathrm{JackP}_\mu(p_k=N)
 \mathrm{JackP}_\mu(p_k=(\tfrac{\beta\lambda}{N})^{\frac{k}{2}}\delta_{k,2})}
 {\mathrm{JackP}_\mu(p_k=\delta_{k,1})}
\ee
together with Cauchy's identity for Jack polynomials
\be
 \exp\Big(\beta\sum_{k\geq1}\frac{p_kt_k}{k}\Big)
 = \sum_\mu \mathrm{JackP}_\mu(p_k)\mathrm{JackQ}_\mu(t_k)
\ee
we can rewrite the generating function as follows
\be
\label{eq:characterexpansionGBE}
 \frac{Z(N,\beta,\lambda,\bu)}{Z(N,\beta,\lambda,0)}
 = \sum_\mu \frac{ \mathrm{JackP}_\mu(p_k=N)
 \mathrm{JackP}_\mu(p_k=\delta_{k,2})}
 {\mathrm{JackP}_\mu(p_k=\delta_{k,1})}
 \mathrm{JackQ}_\mu\Big(p_k=(\tfrac{\beta\lambda}{N})^{\frac{k}{2}}
 \frac{ku_k}{\beta N}\Big)
\ee
This implies that $Z(N,\beta,\lambda,\bu)$ admits a Taylor series expansion in $N$ around $\infty$ and,
after taking the logarithm, the same is also true of $F(N,\beta,\lambda,\bu)$,
so that we can write the series expansion
\be
 F(N,\beta,\lambda,\bu) = \sum_{s=0}^\infty N^{-s} F_{s}(\beta,\lambda,\bu)~.
\ee
The Virasoro constraints for the coefficients $F_s(\beta,\lambda,\bu)$ read
\begin{multline}
\label{eq:newVirasoro}
 \Big(
 \lambda^{-1} \frac{\partial}{\partial u_{n+2}}
 - 2\beta \frac{\partial}{\partial u_n}
 - \beta\sum_{n_1+n_2=n}\frac{\partial^2}{\partial u_{n_1}\partial u_{n_2}}
 \Big) F_s(\beta,\lambda,\bu) \\
 = (1-\beta)(n+1) \frac{\partial}{\partial u_n} F_{s-1}(\beta,\lambda,\bu)
 + \sum_{k>0} k u_k \frac{\partial}{\partial u_{k+n}} F_{s-2}(\beta,\lambda,\bu) \\
 + \beta\sum_{\substack{n_1+n_2=n\\s_1+s_2=s}}
 \frac{\partial F_{s_1}(\beta,\lambda,\bu)}{\partial u_{n_1}}
 \frac{\partial F_{s_2}(\beta,\lambda,\bu)}{\partial u_{n_2}}
 + \beta\delta_{n,0}\delta_{s,0}
 + (1-\beta)\delta_{n,0}\delta_{s,1}
 + u_1 \delta_{n,-1}\delta_{s,2}
\end{multline}

Let us assume we know the functions $F_0(\beta,\lambda,\bu),\dots,F_{s-1}(\beta,\lambda,\bu)$ and we want to solve the constraints with respect to $F_s(\beta,\lambda,\bu)$. Then \eqref{eq:newVirasoro} gives the following
\begin{multline}
\label{eq:non-homogeneous-virasoro}
 \Big(
 \lambda^{-1} \frac{\partial}{\partial u_{n+2}}
 - 2\beta \frac{\partial}{\partial u_n}
 - \beta\sum_{n_1+n_2=n}\frac{\partial^2}{\partial u_{n_1}\partial u_{n_2}}
 - 2\beta\sum_{n_1+n_2=n} \frac{\partial F_0(\beta,\lambda,\bu)}{\partial u_{n_1}}
 \frac{\partial}{\partial u_{n_2}}
 \Big) F_s(\beta,\lambda,\bu) \\
 = B_{s,n}(\beta,\lambda,\bu)
\end{multline}
where $B_{s,n}(\beta,\lambda,\bu)$ is given by
\be
\label{eq:frhs}
\begin{tblr}{rowspec={lcr},cells={mode=dmath}}
 B_{s,n}(\beta,\lambda,\bu) :=
 (1-\beta)(n+1) \frac{\partial F_{s-1}(\beta,\lambda,\bu)}{\partial u_n} \\
 + \sum_{k>0} k u_k \frac{\partial F_{s-2}(\beta,\lambda,\bu)}{\partial u_{k+n}}
 + \beta \sum_{a=1}^{n-1} \sum_{j=1}^{s-1} \frac{\partial F_j(\beta,\lambda,\bu)}{\partial u_a}
 \frac{\partial F_{s-j}(\beta,\lambda,\bu)}{\partial u_{n-a}} \\
 + \beta\delta_{n,0}\delta_{s,0}
 + (1-\beta)\delta_{n,0}\delta_{s,1}
 + u_1 \delta_{n,-1}\delta_{s,2}
\end{tblr}
\ee
Observe that \eqref{eq:non-homogeneous-virasoro} as an equation for $F_s(\beta,\lambda,\bu)$
is now non-homogeneous but linear (for $s>0$). This fact now allows us to solve the Virasoro constraints in a fashion similar to that of \cite{Cassia:2021dpd} in the case of matrix models with boundaries. A formal solution can be derived as follows.
We multiply \eqref{eq:non-homogeneous-virasoro} by $(n+2)u_{n+2}$ and sum over $n=-1,0,\dots$, to get
\be
 \Big( D_{\bu} - (\beta\lambda) W \Big) F_s(\beta,\lambda,\bu)
 = \lambda \sum_{n=1}^\infty n u_n B_{s,n-2}(\beta,\lambda,\bu)
\ee
with
\be
 W := 2\sum_{n=1}^\infty (n+2) u_{n+2} \frac{\partial}{\partial u_n}
 + \sum_{n_1,n_2=1}^\infty (n_1+n_2+2) u_{n_1+n_2+2}
 \left( \frac{\partial^2}{\partial u_{n_1}\partial u_{n_2}}
 + 2 \frac{\partial F_0(\beta,\lambda,\bu)}{\partial u_{n_1}}
 \frac{\partial}{\partial u_{n_2}} \right)
\ee
and $D_{\bu}$ is the dilation operator in \eqref{eq:dilatation}.
Since both $F_s(\beta,\lambda,\bu)$ and the function in the r.h.s.\ have no constant terms in $\bu$, we can invert the operator $D_{\bu}$ to get
\be
\label{eq:solutionW}
\begin{aligned}
 F_s(\beta,\lambda,\bu) = \sum_{k=0}^\infty (\beta\lambda)^k (D_{\bu}^{-1}W)^k
 \lambda D_{\bu}^{-1} \sum_{n=1}^\infty n u_n B_{s,n-2}(\beta,\lambda,\bu)
\end{aligned}
\ee
With enough computational power, this formula allows to derive recursively all the functions $F_s(\beta,\lambda,\bu)$ by repeatedly applying the operators $W$ and $D_{\bu}^{-1}$.
Unfortunately, we do not know how to use \eqref{eq:solutionW} to give a closed formula for the solution of the constraints. In the next sections, however, we will make use of this formal expression to argue some properties about the polynomial dependence on the times $\bu$. In fact, equation \eqref{eq:solutionW} is instrumental in the proof of \cref{{prop:boundedness}}.

\subsection{Order zero}

The order-zero term $F_0(\beta,\lambda,\bu)$ can be easily computed by solving Virasoro constraints in that limit.
The constraints \eqref{eq:newVirasoro} for $s=0$ yield the relations
\begin{multline}
 \Big(
 (\beta\lambda)^{-1} \frac{\partial}{\partial u_{n+2}}
 - 2 \frac{\partial}{\partial u_n}
 - \sum_{n_1+n_2=n}\frac{\partial^2}{\partial u_{n_1}\partial u_{n_2}}
 \Big) F_0(\beta,\lambda,\bu) \\
 = \sum_{n_1+n_2=n} \frac{\partial F_0(\beta,\lambda,\bu)}{\partial u_{n_1}}
 \frac{\partial F_0(\beta,\lambda,\bu)}{\partial u_{n_2}}
 + \delta_{n,0}
\end{multline}
In order to solve this equation, we first come up with an ansatz for the function $F_0(\beta,\lambda,\bu)$ and we plug it in the constraint.
If we can fix the parameters of the ansatz so that the constraints are satisfied, then it must follow that the ansatz is correct, since we know that the solution is unique. Let us consider the ansatz
\be
\label{eq:F0}
 F_0(\beta,\lambda,\bu) = \sum_{k=1}^\infty (\beta\lambda)^\frac{k}{2}
 F_{0,[k]}(\beta) \, u_{k}
\ee
with $F_{0,[k]}(\beta)$ some coefficients to be determined.
From the Virasoro constraints, we get the recursion relation
\be
\label{eq:recursion-catalan0}
 F_{0,[n+2]}(\beta) = 2F_{0,[n]}(\beta)
 + \sum_{n_1+n_2=n} F_{0,[n_1]}(\beta) F_{0,[n_2]}(\beta) + \delta_{n,0}
\ee
which admits the unique solution
\be
 F_{0,[k]}(\beta) = \frac{(1+(-1)^k)}{2} \frac{1}{k/2+1} \binom{k}{k/2}
\ee
for $k\geq 1$. We then have that, at order zero, the coefficients $F_{0,[k]}(\beta)$
are constant that do not depend on the deformation parameter $\beta$.

Observe that naively it would appear that $F_0(\beta,\lambda,\bu)$ is not
polynomial in $\beta$ because of the fractional power in the term
$(\beta\lambda)^{\frac{k}{2}}$.
However, by explicitly solving the recursion \eqref{eq:recursion-catalan0},
we obtain that $F_{0,[k]}(\beta)=0$ if $k$ is odd, so that $F_0(\bu)$ is indeed polynomial in $\beta$ (at each order in the expansion in times $\bu$).
It then follows that the numbers $F_{0,[2k]}(\beta)$ coincide with the ordinary Catalan numbers.

\subsection{Order one}

The order-one term $F_1(\beta,\lambda,\bu)$ satisfies the constraints
\begin{multline}
 \Big(
 \lambda^{-1} \frac{\partial}{\partial u_{n+2}}
 - 2\beta \frac{\partial}{\partial u_n}
 - \beta\sum_{n_1+n_2=n}\frac{\partial^2}{\partial u_{n_1}\partial u_{n_2}}
 - 2\beta\sum_{n_1+n_2=n} \frac{\partial F_0(\beta,\lambda,\bu)}{\partial u_{n_1}}
 \frac{\partial}{\partial u_{n_2}}
 \Big) F_1(\beta,\lambda,\bu) = \\
 = (1-\beta)(n+1) \frac{\partial}{\partial u_n} F_{0}(\beta,\lambda,\bu)
 + (1-\beta)\delta_{n,0}
\end{multline}
which also involve the order-zero function $F_0(\beta,\lambda,\bu)$.
In this case we consider the following ansatz
\be
\label{eq:F1}
 F_1(\beta,\lambda,\bu) = \beta^{-1} \sum_{k=1}^\infty
 (\beta\lambda)^\frac{k}{2} F_{1,[k]}(\beta)\, u_{k}
\ee
with $F_{1,[k]}(\beta)$ some polynomial function of $\beta$.
Then from the Virasoro constraints, we get the recursion relation
\be
\label{eq:recursion-catalan1}
 F_{1,[n+2]}(\beta) = 2F_{1,[n]}(\beta) + (1-\beta)(n+1) F_{0,[n]}(\beta)
 + 2\sum_{n_1+n_2=n} F_{0,[n_1]}(\beta) F_{1,[n_2]}(\beta) + (1-\beta)\delta_{n,0}
\ee
which admits the unique solution
\be
 F_{1,[k]}(\beta) = (1-\beta) \frac{(1+(-1)^k)}{2} \left(2^{k-1}-\binom{k-1}{k/2}\right)
\ee
for $k\geq 1$. As in the previous case, the ansatz allows to solve the recursion
uniquely, and therefore, it is the correct formula for the function $F_1(\beta,\lambda,\bu)$.

It would appear at this point that $F_s(\beta,\lambda,\bu)$ is always a polynomial of degree 1 in the times (w.r.t.\ the grading operator $\sum_{k\geq1} u_k\frac{\partial}{\partial u_k}$); however, this is not true for $s\geq 2$, as we will show in the next section.

\subsection{Higher orders}

Next, we want to compute all higher-order terms in the $1/N$ expansion.
While the constraint equations for the leading order are self-contained and can be solved independently of other orders,
the constraints for higher-order functions $F_s(\beta,\lambda,\bu)$ do depend explicitly on lower orders as well.
Nevertheless, we would like to use the same strategy to solve the constraints at all orders.
Namely, we come up with an ansatz for the full time-dependent free energy, and we use the constraints to fix the coefficients.
If a solution exists, then it must follow that the ansatz gives the correct answer for $F(N,\beta,\lambda,\bu)$.
Moreover, the Virasoro constraints will give a recursive definition of the coefficients.
By analogy with the order-zero case, we name these coefficients \textit{$\beta$-deformed generalized Catalan numbers}.
This is in fact compatible with the definition of (undeformed) generalized Catalan numbers of \cite{WALSH1972192,Carlet:2020enf} when $\beta=1$.

We first want to fix the polynomial dependence on times of the functions $F_s(\beta,\lambda,\bu)$. Making use of the formal solution \eqref{eq:solutionW}, we obtain the following.
\begin{proposition}
\label{prop:boundedness}
The function $F_s(\beta,\lambda,\bu)$ is polynomial in the time variables $\bu$ of degree at most $\lfloor s/2\rfloor+1$ w.r.t.\ the grading operator $\sum_{k>0} u_k \frac{\partial}{\partial u_k}$.
\end{proposition}
\begin{proof}
The proposition can be proven by induction on $s$.
The functions $F_0(\beta,\lambda,\bu)$ and $F_1(\beta,\lambda,\bu)$ are both of degree 1 in times as shown in the previous sections; hence, they satisfy the statement of the proposition.
For $s\geq 2$, we assume that $F_r(\beta,\lambda,\bu)$ is of degree at most $\lfloor r/2\rfloor+1$ for $0\leq r < s$; then, from \eqref{eq:frhs} one observes that the function
\be
 \sum_{n=1}^\infty nu_{n}B_{s,n-2}(\beta,\lambda,\bu)
\ee
is of degree at most $\lfloor s/2\rfloor+1$.%
\footnote{To see this, we need to use the inequality $\lfloor i/2 \rfloor+\lfloor j/2 \rfloor\leq\lfloor s/2 \rfloor$ for all $0<i,j<s$ such that $i+j=s$.}
Since the operator $D_{\bu}$ is degree 0 and $W$ is the sum of a degree 0 and a degree -1 operator, it follows from \eqref{eq:solutionW} that $F_s(\beta,\lambda,\bu)$ is a polynomial of degree at most $\lfloor s/2\rfloor+1$.
\end{proof}

Next, we want to fix the $\lambda$ and $\beta$ dependence. The former is dictated by the homogeneity property in \cref{lm:homogeneity}, while the latter follows from the fact that
$F_s(\beta,\lambda,\bu)$ is polynomial in $\beta$ at every order in times.
We can then write an explicit formula for $F_s(\beta,\lambda,\bu)$,
\be
\label{eq:FsAnsatz}
 F_s(\beta,\lambda,\bu) = \sum_{\ell=1}^{\lfloor s/2\rfloor+1}
 \frac{\beta^{m(s,\ell)}}{\ell!}
 \sum_{k_1,\dots,k_\ell=1}^\infty (\beta\lambda)^{\frac{k_1+\dots+k_\ell}{2}}
 F_{s,[k_1,\dots,k_\ell]}(\beta)
 \prod_{j=1}^{\ell} u_{k_j}
\ee
where the coefficients $F_{s,[k_1,\dots,k_\ell]}(\beta)$ and the exponent $m(s,\ell)$ are yet to be determined (via the constraints). 
Our analysis now indicates that the coefficients $F_{s,[k_1,\dots,k_\ell]}(\beta)$ are polynomial functions of $\beta$.
Moreover, the exponent $m(s,\ell)$ is an integer-valued function of $s,\ell$
such that $m(s,\ell)\geq0$ if $F_{s,[k_1,\dots,k_\ell]}(\beta)\neq0$\footnote{This positivity constraint on $m(s,\ell)$
follows from the requirement that $F_s(\beta,\lambda,\bu)$ must be polynomial in $\beta$.}.

At this point of our derivation, we would like to find a concrete formula for $m(s,\ell)$.
As this is not fully specified by just polynomiality or symmetry arguments,
we propose the following ansatz.
\begin{ansatz}
The integer function $m(s,\ell)$ is given by
\be
\label{eq:ansatz-m}
 m(s,\ell) = \ell-s-1
\ee
\end{ansatz}
A proof that this ansatz is indeed the correct choice for the function $m(s,\ell)$ will follow from the recursion relations that the Virasoro constraint impose on the functions $F_{s,[k_1,\dots,k_\ell]}(\beta)$.

Let us make sure that our formulas satisfy the symmetry discussed in \cref{sec:symmetry}.
In the new variables, the symmetry acts trivially on the times $\bu$, while it transforms the coupling as 
\be
 \lambda\mapsto\beta^2\lambda,
\ee
so that from \eqref{eq:symmetry-F} it must follow that
\be
 F_s(\beta^{-1},\beta^2\lambda,\bu) = (-\beta)^s F_s(\beta,\lambda,\bu)
\ee
This can be checked explicitly for $s=0,1$ from \eqref{eq:F0} and \eqref{eq:F1}.
More generally, we get the non-trivial relation
\be
\label{eq:symmetry-constraint}
 F_{s,[k_1,\dots,k_\ell]}(\beta^{-1})
 = (-\beta)^{2\ell-s-2} F_{s,[k_1,\dots,k_\ell]}(\beta)
\ee
Let $\deg F_{s,[k_1,\dots,k_\ell]}(\beta)$ denote the polynomial degree in $\beta$,
then we can write
\be
 F_{s,[k_1,\dots,k_\ell]}(\beta) = \sum_{i=0}^{\deg F_{s,[k_1,\dots,k_\ell]}(\beta)} (-\beta)^i F_{s,[k_1,\dots,k_\ell],i}
\ee
for some constant coefficients $F_{s,[k_1,\dots,k_\ell],i}$.
Then the symmetry constraint \eqref{eq:symmetry-constraint} implies the following:
\be
 \deg F_{s,[k_1,\dots,k_\ell]}(\beta)
 = 2-2\ell+s
\ee
\be
\label{eq:palindrome}
 F_{s,[k_1,\dots,k_\ell],2-2\ell+s-i} = F_{s,[k_1,\dots,k_\ell],i}
\ee
where the second equation can be equally stated as saying that $F_{s,[k_1,\dots,k_\ell]}(\beta)$ is a palindromic polynomial in $-\beta$.

\begin{remark}\label{rmk:palyndromes}
Observe that any polynomial $P(x,y) = \sum_{i=0}^r P_i x^i y^{r-i}$ such that the coefficients are palindromic (i.e., $P_{r-i}=P_i$), is symmetric in the exchange of $x$ and $y$, and therefore can be written as a linear combination of homogeneous degree-$r$ symmetric polynomials in two variables. The space of such symmetric functions has a special linear basis consisting of Schur polynomials $s_\lambda(x,y)$ associated to partitions $\lambda$ with $r$ boxes (see \cref{sec:Schur}). If the coefficients $P_i$ are integers, then the coefficients in the expansion of $P(x,y)$ over the Schur basis are also integer numbers. The polynomial functions $F_{s,[k_1,\dots,k_\ell]}(\beta)$ can be regarded as special cases of polynomials $P(x,y)$ where $x=-\beta$, $y=1$ and $r=2-2\ell+s$.
\end{remark}

\begin{lemma}\label{lem:palindrome-property}
 If $s$ is odd, the polynomial $F_{s,[k_1,\dots,k_\ell]}(\beta)$ has a zero at $\beta=1$.
\end{lemma}
\begin{proof}
Let us use the symbol $\nu$ to indicate the tuple $[k_1,\dots,k_\ell]$.
If $s$ is odd, we can assume that there is an integer $h$ such that $2-2\ell+s=2h+1$.
Then we can use the identity \eqref{eq:palindrome} to write
\be
\begin{aligned}
 F_{s,\nu}(\beta)
 &= \sum_{i=0}^{2h+1} (-\beta)^i F_{s,\nu,i} \\
 &= \sum_{i=0}^{h} (-\beta)^i F_{s,\nu,i} + \sum_{i=h+1}^{2h+1} (-\beta)^i F_{s,\nu,i} \\
 &= \sum_{i=0}^{h} (-\beta)^i F_{s,\nu,i} + \sum_{i=0}^{h} (-\beta)^{2h+1-i} F_{s,\nu,2h+1-i} \\
 &= \sum_{i=0}^{h} (-\beta)^i F_{s,\nu,i} \left(1-\beta^{1+2(h-i)}\right) \\
 &= (1-\beta) \sum_{i=0}^{h} (-\beta)^i F_{s,\nu,i} \frac{1-\beta^{1+2(h-i)}}{1-\beta} \\
\end{aligned}
\ee
where $\frac{1-\beta^{1+2(h-i)}}{1-\beta}=1+\beta+\dots+\beta^{2(h-i)}$ is a polynomial for all $i$ s.t.\ $0\leq i\leq h$.
\end{proof}
Putting everything together, we get the expansion
\be
 F(N,\beta,\lambda,\bu) = \sum_{s=0}^\infty N^{-s} \sum_{\ell=1}^{\lfloor s/2\rfloor+1}
 \frac{\beta^{\ell-s-1}}{\ell!} \sum_{k_1,\dots,k_\ell=1}^\infty
 (\beta\lambda)^{\sum_{j=1}^\ell\frac{k_j}{2}}
 \sum_{i=0}^{2-2\ell+s} (-\beta)^i F_{s,[k_1,\dots,k_\ell],i}
 \prod_{j=1}^\ell u_{k_j}
\ee

\subsection{Genus expansion}

Recall that for any function $f(s,\ell)$ we have the identity
\be
 \sum_{s=0}^\infty N^{-s} \sum_{\ell=1}^{\lfloor s/2\rfloor+1} f(s,\ell)
 = \sum_{\ell=1}^\infty \sum_{g\in\frac12\mathbb{N}} N^{2-2g-2\ell} f(2g-2+2\ell,\ell)
\ee
which is just a reorganization of the sum on the left.
Hence, we can rewrite the time-dependent free energy as
\begin{multline}
\label{eq:genus-expansion-F(u)}
 F(N,\beta,\lambda,\bu)
 = \sum_{\ell=1}^\infty \sum_{g\in\frac12\mathbb{N}}
 \frac{N^{2-2g-2\ell} \beta^{1-\ell-2g}}{\ell!}
 \sum_{k_1,\dots,k_\ell=1}^\infty (\beta\lambda)^{\sum_{j=1}^\ell \frac{k_j}{2}} \\
 \times \sum_{i=0}^{2g}(-\beta)^i F_{2g-2+2\ell,[k_1,\dots,k_\ell],i}
 \prod_{j=1}^\ell u_{k_j}
\end{multline}
where, by analogy with the undeformed case, we regard the sum over the half-integer $g$ as the genus expansion of the time-dependent free energy.
\begin{definition}
Let $\nu$ be an integer partition. We define the genus-$g$ Catalan polynomial associated to $\nu$ as
\be
 C_{g,\nu}(\beta) := F_{2g-2+2\ell(\nu),\nu}(\beta)
\ee
where $\ell(\nu)$ is the length of the partition and the polynomial in the r.h.s.\ is defined as in \cref{eq:FsAnsatz}.
Moreover, from \eqref{eq:symmetry-constraint} it follows that
\be
\label{eq:symmetry-for-Catalan-poly}
 C_{g,\nu}(\beta^{-1}) = (-\beta)^{-2g} C_{g,\nu}(\beta)
\ee
and we can use this symmetry to write
\be
 C_{g,\nu}(\beta) = \sum_{i_1+i_2=2g} C_{g,\nu}^{(i_1,i_2)} (1)^{i_1} (-\beta)^{i_2}
\ee
with coefficients $C_{g,\nu}^{(i_1,i_2)}$ symmetric in $i_1,i_2$.
\end{definition}
In the next section, we will argue that Catalan polynomials satisfy the two crucial properties:
\begin{itemize}
\item the coefficients of $C_{g,\nu}(\beta)$ are integers;
\item the polynomials $C_{g,\nu}(\beta)$ evaluate to the generalized Catalan numbers of \cite{WALSH1972192} at $\beta=1$.
\end{itemize}

Observe that the summation label $g$ in \eqref{eq:genus-expansion-F(u)} can now be interpreted as the genus of some surface.
Remarkably, the sum over $g$ ranges over both integer and half-integer (positive) numbers,
which suggests that the combinatorial interpretation of the time-dependent free energy is not as simple as in the undeformed case.
Namely, the fact that the genus is allowed to take half-integer values is an indication that what the \GBE{} is counting
are not just orientable maps but more generally all locally orientable\footnote{
Observe that the genus of an orientable surface with Euler characteristic $\chi$ is defined as $g=\frac12(2-\chi)$ where $\chi$ is always an even integer. For a non-orientable surface, however, the Euler number $\chi$ can be odd, therefore $\frac12(2-\chi)$ can be half-integer.} maps as already recognized in works of Goulden and Jackson (see \cite{Goulden:1996con,LaCroix:2009}).
We provide a more in-depth discussion of the combinatorial interpretation of our results in \cref{sec:hypermaps}.

In order to solve for the coefficients $C_{g,\nu}(\beta)$, we will now simplify the genus expansion of the time-dependent free energy as much as possible. To this end, we consider the change of time variables\footnote{Under the inversion symmetry $\beta\mapsto1/\beta$, the new times $\bv$ transform trivially, as the combinations of parameters $\beta N^2$ and $\beta\lambda$ are both invariant.}
\be
 u_k = \beta N^2(\beta\lambda)^{-\frac{k}{2}} v_k
\ee
which leads to the following expression
\begin{multline}
 F(N,\beta,\lambda,\{u_k = \beta N^2(\beta\lambda)^{-\frac{k}{2}} v_k\}) \\
 = \sum_{\ell=1}^\infty \sum_{g\in\frac12\mathbb{N}}
 \frac{ (\beta N^2)^{1-g} \beta^{-g} }{\ell!}
 \sum_{k_1,\dots,k_\ell=1}^\infty
 C_{g,[k_1,\dots,k_\ell]}(\beta) \prod_{j=1}^\ell v_{k_j} \\
 = \beta N^2 \sum_{g\in\frac12\mathbb{N}} (\beta N)^{-2g}
 \sum_{\nu\neq\emptyset}
 \frac{1}{|{\Aut}(\nu)|} C_{g,\nu}(\beta) \, p_\nu(\bv) \,,
\end{multline}
where $p_\nu(\bv)=\prod_{k\in\nu}v_k$ are power-sum polynomials in times $\bv$,
and we used the combinatorial identity
\be
 \sum_{\ell=1}^\infty \frac{1}{\ell!}
 \sum_{k_1,\dots,k_\ell=1}^\infty
 f_{[k_1,\dots,k_\ell]}
 = \sum_{\nu\neq\emptyset} \frac{1}{|{\Aut}(\nu)|} f_\nu\,,
\ee
with $|\Aut(\nu)|=p_\nu(\{\frac{\partial}{\partial v_k}\}) p_\nu(\bv)$ being the order of the automorphism group of the partition~$\nu$.

Let us define the function
\be
\label{eq:Gdef}
 G(\genus,\beta,\bv) := \beta\genus^2 F((\beta\genus)^{-1},\beta,\lambda,
 \{u_k = (\beta\genus^2)^{-1} (\beta\lambda)^{-\frac{k}{2}} v_k\})
\ee
then we have that
\be
\label{eq:seriesG}
 G(\genus,\beta,\bv) = \sum_{r=0}^\infty \genus^r G_r(\beta,\bv)
 = \sum_{r=0}^\infty \genus^r
 \sum_{\nu\neq\emptyset} \frac{1}{|{\Aut}(\nu)|} C_{r/2,\nu}(\beta)
 \, p_\nu(\bv)
\ee
is the generating function of all Catalan polynomials, with $\genus$ being the genus counting variable.
The Virasoro constraints for the function $G(\genus,\beta,\bv)$ become
\begin{multline}
\label{eq:virasoro_constraints_G}
 \Big(
   \frac{\partial}{\partial v_{n+2}}
 - \left(2+\genus(1-\beta)(n+1)\right) \frac{\partial}{\partial v_{n}}
 - \genus^2 \beta \sum_{n_1+n_2=n}\frac{\partial^2}{\partial v_{n_1}\partial v_{n_2}}
 - \sum_{k>0} k v_k \frac{\partial}{\partial v_{k+n}}
 \Big) G(\genus,\beta,\bv) \\
 - \sum_{n_1+n_2=n} \frac{\partial G(\genus,\beta,\bv)}{\partial v_{n_1}}
 \frac{\partial G(\genus,\beta,\bv)}{\partial v_{n_2}}
 = \left(1+\genus(1-\beta)\right)\delta_{n,0} + v_1 \delta_{n,-1}
\end{multline}
and expanding in powers of $\genus$, we get
\begin{multline}
\label{eq:virasoroGr}
 \Big( \frac{\partial}{\partial v_{n+2}}
 - 2 \frac{\partial}{\partial v_n}
 - \sum_{k>0} k v_k \frac{\partial}{\partial v_{k+n}} \Big) G_{r}(\beta,\bv) =
  (1-\beta)(n+1) \frac{\partial G_{r-1}(\beta,\bv)}{\partial v_n} \\
 + \beta \sum_{n_1+n_2=n} \frac{\partial^2 G_{r-2}(\beta,\bv)}
 {\partial v_{n_1}\partial v_{n_2}}
 + \sum_{n_1+n_2=n} \sum_{r_1+r_2=r} \frac{\partial G_{r_1}(\beta,\bv)}{\partial v_{n_1}}
 \frac{\partial G_{r_2}(\beta,\bv)}{\partial v_{n_2}} \\
 + (\delta_{n,0} + v_1 \delta_{n,-1}) \delta_{r,0}
 + (1-\beta)\delta_{n,0} \delta_{r,1}\,.
\end{multline}
\begin{remark}
Let us consider the case when $\beta=-\frac{\epsilon_2}{\epsilon_1}$. Then, we can write
\be
\label{eq:epsilon-poly}
 \epsilon_1^{2g}C_{g,\nu}(-\epsilon_2/\epsilon_1)
 = \sum_{i_1+i_2=2g} \epsilon_1^{i_1} \epsilon_2^{i_2} C_{g,\nu}^{(i_1,i_2)}
\ee
which is a homogeneous degree-$2g$ symmetric polynomial in $\epsilon_1,\epsilon_2$, due to the property \eqref{eq:symmetry-for-Catalan-poly}.
Since any such symmetric polynomial can be expressed as an integer linear combination of Schur polynomials $s_\lambda(\epsilon_1,\epsilon_2)$ for $|\lambda|=2g$ (see \cref{sec:Schur}), we can rewrite \eqref{eq:epsilon-poly} as
\be
 \epsilon_1^{2g}C_{g,\nu}(-\epsilon_2/\epsilon_1)
 = \sum_{\lambda\vdash 2g} n_{\nu,\lambda} s_\lambda(\epsilon_1,\epsilon_2)~,
\ee
where $n_{\nu,\lambda}\in\mathbb{Z}$.
Then we can define the generating function of all $n_{\nu,\lambda}$ as
\be
\label{eq:GEnergySchur}
 n(\epsilon_1,\epsilon_2,\bv):=G(\epsilon_1,-\epsilon_2/\epsilon_1,\bv)
\ee
or, equivalently, we can write $G(\genus,\beta,\bv)=n(\genus,-\beta\genus,\bv)$.
The constraints for the function $n(\epsilon_1,\epsilon_2,\bv)$ are
\begin{multline}
 \Big(
   \frac{\partial}{\partial v_{n+2}}
 - \left(2+(\epsilon_1+\epsilon_2)(n+1)\right) \frac{\partial}{\partial v_{n}}
 + (\epsilon_1\epsilon_2) \sum_{n_1+n_2=n}
 \frac{\partial^2}{\partial v_{n_1}\partial v_{n_2}}
 - \sum_{k>0} k v_k \frac{\partial}{\partial v_{k+n}}
 \Big) n(\epsilon_1,\epsilon_2,\bv) \\
 - \sum_{n_1+n_2=n} \frac{\partial n(\epsilon_1,\epsilon_2,\bv)}{\partial v_{n_1}}
 \frac{\partial n(\epsilon_1,\epsilon_2,\bv)}{\partial v_{n_2}}
 = \left(1+(\epsilon_1+\epsilon_2)\right)\delta_{n,0} + v_1 \delta_{n,-1}
\end{multline}

Observe that Schur functions of two variables vanish identically for $\ell(\lambda)>2$;
therefore, the sum over $\lambda$ collapses to a sum over partitions of the form
$\lambda=[h+d,d]$ for $h,d\geq0$. We then have
\be
\begin{aligned}
 n(\epsilon_1,\epsilon_2,\bv)
 &= \sum_{\nu\neq\emptyset} \frac{p_\nu(\bv)}{|\Aut(\nu)|}
 \sum_{r=0}^\infty \sum_{d=0}^{\lfloor r/2 \rfloor}
 n_{\nu,[r-d,d]} s_{[r-d,d]}(\epsilon_1,\epsilon_2) \\
 &= \sum_{\nu\neq\emptyset} \frac{p_\nu(\bv)}{|\Aut(\nu)|}
 \sum_{h=0}^\infty \sum_{d=0}^\infty
 n_{\nu,[h+d,d]} s_{[h+d,d]}(\epsilon_1,\epsilon_2) \\
 &= \sum_{\nu\neq\emptyset} \frac{p_\nu(\bv)}{|\Aut(\nu)|}
 \sum_{i,j,d\geq0}
 n_{\nu,[i+j+d,d]} \epsilon_1^{d+i}\epsilon_2^{d+j} \\
 &= \sum_{\nu\neq\emptyset} \frac{p_\nu(\bv)}{|\Aut(\nu)|}
 \sum_{a,b\geq0}
 C_{\frac{a+b}{2},\nu}^{(a,b)} \epsilon_1^{a}\epsilon_2^{b} \\
\end{aligned}
\ee
where we used the combinatorial identity
\be
 \sum_{r=0}^\infty \sum_{d=0}^{\lfloor r/2\rfloor} f(r,d)
 = \sum_{h=0}^\infty \sum_{d=0}^\infty f(h+2d,d)
\ee
together with \eqref{eq:coproductSchur2}.
The genus in the formula for $n(\epsilon_1,\epsilon_2,\bv)$ is given by half of the monomial degree in the $\epsilon_i$'s, i.e., the eigenvalue of the operator
$\frac12(\epsilon_1\frac{\partial}{\partial\epsilon_1}+\epsilon_2\frac{\partial}{\partial\epsilon_2})$.
We have the relation of coefficients
\be
\label{eq:catalan-to-n}
 C_{g,\nu}^{(i,2g-i)} = \sum_{d=0}^{\min(i,2g-i)} n_{\nu,[2g-d,d]}
\ee
which suggests that the integers $n_{\nu,\lambda}$ might have a more fundamental role than the coefficients of the Catalan polynomials $C_{g,\nu}(\beta)$.
Moreover, explicit computations of some of the $n_{\nu,\lambda}$ indicate that these coefficients are always positive integers, see \cref{tblr:SchurEpsilon}.
While this is strong evidence that these numbers are the solution to some counting problem, we are not aware of possible combinatorial or enumerative geometric interpretations of the integers $n_{\nu,\lambda}$. It would be interesting to investigate this further.

\end{remark}

\section{Recursion formula}
\label{sec:recursion}

Having fixed the general form of the series expansion of the time-dependent free energy, and the closely related function $G(\genus,\beta,\bv)$, we now need to fix the coefficients of $C_{g,[k_1,\dots,k_\ell]}(\beta)$.
To this end, we make use again of the Virasoro constraints as written in \eqref{eq:virasoroGr}.
By plugging the formula for the generating function $G(\genus,\beta,\bv)$ into the constraint equations, we get nonlinear relations between the coefficients.
Provided we fix some appropriate initial conditions, these relations determine a recursion on the Catalan polynomials, which can be solved uniquely. This recursion formula is completely equivalent to the Virasoro constraints themselves, and it is in fact also interpreted as a topological recursion formula for the genus expansion of the time-dependent free energy.
Similar topological recursion formulas for the \GBE{} were already considered in \cite{Brini:2010fc} and explicit formulas for the $1/N$ expansion of the resolvent (i.e., the generating function of Catalan polynomials associated to partitions with a single part) were also computed in \cite{Witte:2013cea} through the loop equation formalism.

\subsection{Initial conditions}

We consider first the constraint equations following from the Virasoro constraints for $n=-1$ and $n=0$ as they allow to define the initial conditions for the recursion.

For $n=-1$ we have the differential equation
\be
 \boxed{n=-1} \quad\quad
 \Big( \frac{\partial}{\partial v_1}
 - \sum_{k>0} k v_k \frac{\partial}{\partial v_{k-1}} \Big) G_r(\beta,\bv) =
 v_1 \delta_{r,0}
\ee
which translates to the following recursion relation for the Catalan polynomials
\be
 C_{g,[k_1,\dots,k_\ell,1]}(\beta) =
 \sum_{j=1}^\ell k_j C_{g,[k_1,\dots,k_{j-1},k_j-1,k_{j+1},\dots,k_\ell]}(\beta)
 + \delta_{[k_1,\dots,k_\ell],[1]} \delta_{g,0}
\ee
Similarly, for $n=0$ we have the differential equation
\be
 \boxed{n=0} \quad\quad
 \Big( \frac{\partial}{\partial v_2}
 - \sum_{k>0} k v_k \frac{\partial}{\partial v_k} \Big) G_r(\beta,\bv) =
 \delta_{r,0} + (1-\beta) \delta_{r,1}
\ee
which translates to
\be
 C_{g,[k_1,\dots,k_\ell,2]}(\beta) =
 \sum_{j=1}^\ell k_j C_{g,[k_1,\dots,k_\ell]}(\beta) \\
 + \left( \delta_{g,0}
 + (1-\beta) \delta_{g,\frac12}
 \right) \delta_{[k_1,\dots,k_\ell],[]} \,.
\ee
Observe that both equations only contain contributions $G_r(\beta,\bv)$ for a fixed $r$,
which means that these relations are independent of all other genera (differently from what happens for $n\geq1$).

\subsection{Recursion for \texorpdfstring{$n\geq1$}{n>=1}}

In general, we can write the recursion relation for an arbitrary polynomial $C_{g,[k_1,\dots,k_\ell]}(\beta)$ by expanding \eqref{eq:virasoroGr} in all possible monomials in times $\bv$.
For $n\geq1$, the equation corresponding to the monomial $p_{[k_1,\dots,k_\ell]}(\bv)=v_{k_1}\cdots v_{k_\ell}$ gives the recursion
\begin{multline}
\label{eq:catalanpolyrecursion}
 C_{g,[k_1,\dots,k_\ell,n+2]}(\beta) =
 2 C_{g,[k_1,\dots,k_\ell,n]}(\beta)
 + \sum_{j=1}^\ell k_j C_{g,[k_1,\dots,k_{j-1},k_j+n,k_{j+1},\dots,k_\ell]}(\beta) \\
 + (1-\beta)(n+1) C_{g-\frac12,[k_1,\dots,k_\ell,n]}(\beta)
 + \beta \sum_{\substack{n_1+n_2=n\\n_1,n_2\geq1}} C_{g-1,[k_1,\dots,k_\ell,n_1,n_2]}(\beta) \\
 + \sum_{\substack{n_1+n_2=n\\n_1,n_2\geq1}} \sum_{\substack{g_1+g_2=g\\g_1,g_2\geq0}}
 \sum_{\substack{I_1\cup I_2=[k_1,\dots,k_\ell]\\I_1\cap I_2=\emptyset}}
 C_{g_1,I_1\cup[n_1]}(\beta) \, C_{g_2,I_2\cup[n_2]}(\beta)
\end{multline}
where it is understood that the polynomials $C_{g,[k_1,\dots,k_\ell]}(\beta)$ are
identically zero whenever one or more of the labels $k_j$ are $\leq0$.
This recursion relation is sometimes also known as \emph{cut-and-join recursion}, and it corresponds to the $\beta$-deformed version of the recursion in \cite[(6)]{WALSH1972192}.
Observe that symmetry of the equation \eqref{eq:catalanpolyrecursion} under $\beta\mapsto1/\beta$ is a straightforward consequence of \eqref{eq:symmetry-for-Catalan-poly}.

Expanding further in powers of $\beta$, we get
\begin{multline}
 C_{g,[k_1,\dots,k_\ell,n+2]}^{(i,2g-i)} =
 2 C_{g,[k_1,\dots,k_\ell,n]}^{(i,2g-i)}
 + \sum_{j=1}^\ell k_j C_{g,[k_1,\dots,k_{j-1},k_j+n,k_{j+1},\dots,k_\ell]}^{(i,2g-i)} \\
 + (n+1)\left( C_{g-\frac12,[k_1,\dots,k_\ell,n]}^{(i-1,2g-i)}
 + C_{g-\frac12,[k_1,\dots,k_\ell,n]}^{(i,2g-i-1)} \right)
 - \sum_{\substack{n_1+n_2=n\\n_1,n_2\geq1}}
 C_{g-1,[k_1,\dots,k_\ell,n_1,n_2]}^{(i-1,2g-i-1)} \\
 + \sum_{\substack{i_1+i_2=i\\i_1,i_2\geq0}} \sum_{\substack{n_1+n_2=n\\n_1,n_2\geq1}} \sum_{\substack{g_1+g_2=g\\g_1,g_2\geq0}}
 \sum_{\substack{I_1\cup I_2=[k_1,\dots,k_\ell]\\I_1\cap I_2=\emptyset}}
 C_{g_1,I_1\cup[n_1]}^{(i_1,2g_1-i_1)} \,
 C_{g_2,I_2\cup[n_2]}^{(i_2,2g_2-i_2)} \\
 + \left( \delta_{g,0}
 + (\delta_{i,1}+\delta_{2g-i,1}) \delta_{g,\frac12}
 \right) \delta_{[k_1,\dots,k_\ell],[]} \delta_{n,0}
 + \delta_{[k_1,\dots,k_\ell],[1]} \delta_{g,0} \delta_{n,-1}
\end{multline}
One can check explicitly that this recursion admits a unique solution and, moreover, that the coefficients $C_{g,\nu}^{(i,2g-i)}$ that solve the recursion are integer numbers.
Solving the equations for the first few Catalan polynomials gives \cref{tblr:betapoly}.

\section{The undeformed limit}
\label{sec:limit}

In the limit $\beta\to1$, the polynomials $C_{g,[k_1,\dots,k_\ell]}(\beta)$ reduce to alternating sums of coefficients $C_{g,[k_1,\dots,k_\ell]}^{(a,b)}$ as follows
\be
\label{eq:undeformed-Catalan}
 C_{g,[k_1,\dots,k_\ell]}(1)
 = \sum_{i=0}^{2g} (-1)^i C_{g,[k_1,\dots,k_\ell]}^{(i,2g-i)}
 =: C^{\rm CLPS}_{g,[k_1,\dots,k_\ell]}
\ee
where $C^{\rm CLPS}_{g,[k_1,\dots,k_\ell]}$ are the coefficients of \cite{Carlet:2020enf}.

Observe also that, in the limit $\beta\to1$ (i.e., $\epsilon_1+\epsilon_2=0$), the Schur polynomials $s_\lambda(\epsilon_1,\epsilon_2)$ evaluate to either $0$ or $\pm1$ (see \eqref{eq:coproductSchur2}):
\be
 \lim_{\beta\to1} s_{[h+d,d]}(\genus,-\beta\genus) = \genus^{h+2d} (-1)^d \frac{1+(-1)^h}{2}
\ee
This implies that
\be
 \lim_{\beta\to1} G(\genus,\beta,\bv)
 = \sum_{\nu\neq\emptyset} \frac{p_\nu(\bv)}{|\Aut(\nu)|}
 \sum_{h=0}^\infty \sum_{d=0}^\infty \genus^{2(h+d)}
 (-1)^{d} n_{\nu,[2h+d,d]}
\ee
which means that the genus expansion of $G(\genus,1,\bv)$ only contains integer genera contributions.

Similarly, the time-dependent free energy simplifies to
\be
 \lim_{\beta\to 1} F(N,\beta,\lambda,\bu) =
 \sum_{\ell=1}^\infty \sum_{g\in\mathbb{N}}
 \frac1{\ell!} N^{2-2g-2\ell}
 \sum_{k_1,\dots,k_\ell=1}^\infty \lambda^{\sum_{j=1}^\ell k_j /2}
 C^{\rm CLPS}_{g,[k_1,\dots,k_\ell]} \prod_{j=1}^\ell u_{k_j}
\ee
so that
\be
 C^{\rm CLPS}_{g,\nu} = \sum_{d=0}^{g} (-1)^d n_{\nu,[2g-d,d]}\,.
\ee
Observe that from \cref{lem:palindrome-property} we have that the sum in
\eqref{eq:undeformed-Catalan} vanishes identically when $2g$ is odd.
The sum over the genus $g$ then collapses to a sum over only positive integers,
and we recover the known genus expansion of the time-dependent free energy of the GUE as described in \cite{WALSH1972192}.

\section{Wishart--Laguerre \texorpdfstring{$\beta$}{beta}-ensemble}
\label{sec:WLBE}

There is another example of $\beta$-deformed ensemble that has the property that
Virasoro constraints admit a unique solution. This is the Wishart--Laguerre
$\beta$-ensemble which corresponds to the deformation of the matrix model with linear potential $V(x)=\frac{N}{\lambda}x$.
Differently from the case of the Gaussian matrix model, one can additionally introduce
a logarithmic term in the potential without spoiling the solvability of Virasoro constraints \cite{Cassia:2020uxy}. We will denote as $\alpha$ the coupling corresponding to such logarithmic interaction. After exponentiation of the potential, this gives rise to a determinant insertion of the form $\prod_{i=1}^N x_i^\alpha$.
The generating function of all polynomial expectation values can then be defined as
\be
\label{eq:generating-functionWL}
 \frac1{N!}\int_{\mathbb{R}_+^N}
 \prod_{i=1}^N \mathd x_i ~ \prod_{i<j}|x_i-x_j|^{2\beta} ~
 \prod_{i=1}^N x_i^{\alpha} ~ \mathe^{ - \frac{N}{\lambda} \sum_i x_i
 + \frac{1}{N} \sum_{k=1}^\infty u_k \sum_i x_i^k}
\ee
Convergence of the integral imposes that the contour be restricted to the positive orthant in real $N$-dimensional space, together with the conditions $\Re(\alpha)>-1$ and $\Re(\lambda^{-1})>0$.

We remark that explicit formulas for the $1/N$ expansion of the resolvent of the \WLBE{} (and corresponding topological recursion relations) were previously obtained in \cite{Forrester:2017lar}. In the following, we will derive an analogous formula for the $1/N$ expansion of the full time-dependent free energy following a similar strategy to that of the case of the \GBE{}.

\subsection{Symmetries}
The time-dependent free energy of the \WLBE{} satisfies the following symmetry properties.
First, both the normalized generating function and its logarithm are invariant under the involution
\be
\label{eq:WLBELanglands}
\begin{tblr}{cccc}
 \beta\mapsto\frac1{\beta}~, &
 N\mapsto-\beta N~, &
 \lambda\mapsto\beta^2\lambda~, &
 \alpha\mapsto-\frac1{\beta}\alpha
\end{tblr}
\ee
For later convenience, we introduce a new parameter $\phi$ defined by the equation
\be
 \alpha = (1-\beta)(\phi-1)
\ee
which implies $\phi=1+\frac{\alpha}{1-\beta}$. Observe that under the symmetry \eqref{eq:WLBELanglands}, the new parameter $\phi$ is trivially mapped to itself.
We then have
\be
\label{eq:WLBELanglandsF}
 F(-\beta N,\beta^{-1},\beta^2\lambda,\phi,\bu) = F(N,\beta,\lambda,\phi,\bu)
\ee

Second, we observe the homogeneity equation
\be
\label{eq:WLBEhomo}
 F(N,\beta,\lambda,\phi,\bu) = \lambda^{D_{\bu}} F(N,\beta,1,\phi,\bu)
\ee
which follows from the change of variables $x_i=\lambda y_i$ in the integral \eqref{eq:generating-functionWL}. This is analogous to the equation \eqref{lm:homogeneity} in the Gaussian case, with the caveat that the homogeneity degree of the time-dependent free energy is different in the two cases.

\subsection{Superintegrability}
Similarly to the Gaussian case, the \WLBE{} is a matrix model that exhibits superintegrability for the averages of characters; namely, it can be shown \cite{Mezzadri:2017mom,Cassia:2020uxy,Bawane:2022cpd} that the ensemble average of Jack polynomials satisfies
\be
 \frac{\langle\mathrm{JackP}_\mu(x_1,\dots,x_N)\rangle_{\text{\WLBE}}}
 {\langle 1 \rangle_{\text{\WLBE}}}
 = \frac{\mathrm{JackP}_\mu(p_k=N)
 \mathrm{JackP}_\mu(p_k=(\tfrac{\beta\lambda}{N})^{k}(N+\beta^{-1}\phi(1-\beta)))}
 {\mathrm{JackP}_\mu(p_k=\delta_{k,1})}
\ee
This formula then can be used to define the generating function of the matrix model
\be
\label{eq:character-expansion-WL}
 \frac{Z(N,\beta,\lambda,\phi,\bu)}{Z(N,\beta,\lambda,\phi,0)} =
 \sum_{\lambda}
 \frac{\mathrm{JackP}_\mu(p_k=N)
 \mathrm{JackP}_\mu(p_k=N+\beta^{-1}\phi(1-\beta))}
 {\mathrm{JackP}_\mu(p_k=\delta_{k,1})}
 \mathrm{JackQ}_\mu\Big(p_k=(\tfrac{\beta\lambda}{N})^{k} \frac{k u_k}{\beta N}\Big)
\ee
and can be used to argue for the polynomiality of the time-dependent free energy in the parameters $\beta$ and $\phi$.

\subsection{Virasoro constraints}

By analogy with the \GBE{}, we are led to consider the change of time variables
\be
 u_k = \beta N^2(\beta\lambda)^{-k} v_k
\ee
where the power of $(\beta\lambda)$ has no factors of $\frac12$ due to the different dependence on the coupling $\lambda$ as observed in \eqref{eq:WLBEhomo}.

A similar analysis to that of \cref{sec:genus} then suggests that one should consider the function
\be
 G(\genus,\beta,\phi,\bv) :=
 \beta\genus^2 F\Big((\beta\genus)^{-1},\beta,\lambda,\phi,
 \{u_k = (\beta\genus^2)^{-1} (\beta\lambda)^{-k} v_k\}\Big)
\ee
for which the Virasoro constraints%
\footnote{See \cite{Cassia:2020uxy} for a derivation of the constraints for the generating function $Z(N,\beta,\lambda,\phi,\bu)$ of the \WLBE{}.}
become
\begin{multline}
\label{eq:WLBEVirasoroG}
 \Big(
 \frac{\partial}{\partial v_{n+1}}
 - (2+\genus(1-\beta)(n+\phi)) \frac{\partial}{\partial v_n}
 - \genus^2 \beta \sum_{n_1+n_2=n} \frac{\partial^2}{\partial v_{n_1} \partial v_{n_2}}
 - \sum_{k=1}^\infty k v_k \frac{\partial}{\partial v_{k+n}}
 \Big) G(\genus,\beta,\phi,\bv) \\
 - \sum_{n_1+n_2=n} \frac{\partial G(\genus,\beta,\phi,\bv)}{\partial v_{n_1}}
 \frac{\partial G(\genus,\beta,\phi,\bv)}{\partial v_{n_2}}
 = (1+\genus(1-\beta)\phi) \delta_{n,0}
\end{multline}
with $n\geq0$.

The function $G(\genus,\beta,\phi,\bv)$ then admits the following genus expansion in $\genus$,
\be
\begin{aligned}
 G(\genus,\beta,\phi,\bv) &= \sum_{r=0}^\infty \genus^r G_r(\beta,\phi,\bv) \\
 &= \sum_{r=0}^\infty \genus^r
 \sum_{\nu\neq\emptyset} \frac{p_\nu(\bv)}{|{\Aut}(\nu)|} C_{r/2,\nu}(\beta,\phi) \\
 &= \sum_{\nu\neq\emptyset}
 \frac{p_\nu(\bv)}{|{\Aut}(\nu)|}
 \sum_{\lambda} \genus^{|\lambda|}
 n_{\nu,\lambda}(\phi) s_\lambda(1,-\beta)
\end{aligned}
\ee
where $C_{g,\nu}(\beta,\phi)$ are polynomial functions in $\beta$ and $\phi$,
\be
 C_{g,\nu}(\beta,\phi) \equiv \sum_{i,j=0}^{2g} (-\beta)^i \phi^j C_{g,\nu,j}^{(i,2g-i)}
\ee
with coefficients $C_{g,\nu,j}^{(i,2g-i)}$ to be fixed by the constraints.
We shall call $C_{g,\nu}(\beta,\phi)$ the \emph{Catalan polynomials} in genus $g$.
The symmetry in \eqref{eq:WLBELanglands} imposes that $C_{g,\nu}(\beta,\phi)$ be
palindromic polynomials in $(-\beta)$; however, there does not seem to be such a symmetry w.r.t.\ the variable $\phi$.

The equations in \eqref{eq:WLBEVirasoroG} then can be expanded in
powers of $\genus$ and monomials in times $\bv$, to obtain a recursion on the
Catalan polynomials $C_{g,\nu}(\beta,\phi)$. Namely, we have
\begin{multline}
 C_{g,[k_1,\dots,k_\ell,n+1]}(\beta,\phi) =
 2 C_{g,[k_1,\dots,k_\ell,n]}(\beta,\phi)
 + \sum_{j=1}^\ell k_j C_{g,[k_1,\dots,k_{j-1},k_j+n,k_{j+1},\dots,k_\ell]}(\beta,\phi) \\
 + (1-\beta)(n+\phi) C_{g-\frac12,[k_1,\dots,k_\ell,n]}(\beta,\phi)
 + \beta \sum_{\substack{n_1+n_2=n\\n_1,n_2\geq1}} C_{g-1,[k_1,\dots,k_\ell,n_1,n_2]}(\beta,\phi) \\
 + \sum_{\substack{n_1+n_2=n\\n_1,n_2\geq1}} \sum_{\substack{g_1+g_2=g\\g_1,g_2\geq0}}
 \sum_{\substack{I_1\cup I_2=[k_1,\dots,k_\ell]\\I_1\cap I_2=\emptyset}}
 C_{g_1,I_1\cup[n_1]}(\beta,\phi) \, C_{g_2,I_2\cup[n_2]}(\beta,\phi) \\
 + \left( \delta_{g,0} + (1-\beta) \phi\, \delta_{g,\frac12} \right)
 \delta_{[k_1,\dots,k_\ell],[]} \delta_{n,0}
\end{multline}
for $n\geq0$. We recall that the $n=-1$ constraint is not well-defined in the \WLBE{} case;
however, the recursion for non-negative $n$ can still be solved uniquely.

These recursion relations imply that the coefficients of $C_{g,[k_1,\dots,k_\ell]}(\beta,\phi)$ are integer numbers. Equivalently, we have that
$n_{\nu,\lambda}(\phi)$ are polynomials in $\phi$ of degree $2g=|\lambda|$ and coefficients in $\mathbb{Z}$.
Explicit computations of these integer numbers by means of the recursion suggest that they are always positive; however, we do not have a proof of this claim.
See \cref{tblr:SchurEpsilonWLphi} and \cref{tblr:SchurEpsilonWL} for lists of some of these coefficients for arbitrary $\phi$ and for $\phi=1$, respectively.

\section{Outlook}
\label{sec:outlook}
In this article we have analyzed the dependence of the free energy (logarithm of the generating function of correlators) of the \GBE{} and \WLBE{} on the parameters of
the matrix model, most importantly the rank $N$ and the deformation parameter $\beta$.

We exploited Virasoro constraints for the (time-dependent) free energy to derive an ansatz
for the genus expansion \eqref{eq:genus-expansion-F(u)} with coefficients defined as
$\beta$-deformations of generalized higher genus Catalan numbers.
The ansatz is motivated both by symmetry arguments and by showing that it leads to
a (topological) recursion relation on the coefficients, whose solution is unique.

Upon suitable change of time variables, we have been able to also define the function $G(\genus,\beta,\bv)$ which naturally admits a genus expansion in the variable $\genus$,
with coefficients given by the higher genus Catalan polynomials $C_{g,\nu}(\beta)$.
This led us to reformulate the genus expansion in terms of more fundamental integer invariants $n_{\nu,\lambda}$ and their generating function $n(\epsilon_1,\epsilon_2,\bv)$,
which is manifestly symmetric in the exchange of $\epsilon_1$ and $\epsilon_2$, equivalent to the involution sending $\beta$ to $1/\beta$. Quite remarkably, all the invariants $n_{\nu,\lambda}$ that we have been able to compute turn out to be positive integers, strongly suggesting that they should have a specific combinatorial interpretation.

As a last comment, we observe that our results hold for arbitrary values of the deformation parameter $\beta$, and that, by specializing to the values $\beta=\tfrac12,1$ and $2$ we immediately obtain the case of the orthogonal, unitary and symplectic ensembles, respectively. This implies that the $C_{g,\nu}(\beta)$ provide an analytic continuation of
higher genus Catalan numbers for all those classical matrix ensembles. 

A number of question arise naturally from our discussion.
\begin{itemize}
\item While integrality of the $n_{\nu,\lambda}$ follows from the cut-and-join recursion relations, their positivity is not immediately obvious; however, \cref{tblr:SchurEpsilon,tblr:SchurEpsilonWL} give strong evidence for this claim.
A possible explanation of this fact could follow from some combinatorial interpretation of the positive integers $n_{\nu,\lambda}$ as dimensions of certain vector spaces.
Comparison with the results of \cite{Goulden:1996con,Goulden:1996b,LaCroix:2009}
would suggest that they should be related to counts of locally orientable maps as discussed in \cref{sec:hypermaps};
however, the precise details of the identification remain unclear.
Another possible connection is with the theory of $b$-weighted Hurwitz numbers, as in recent works of \cite{CHAPUY2022108645,10.1093/imrn/rnac177}.
\item Can one use similar techniques to derive a genus expansion for the time-dependent free energy of $q,t$-deformed matrix models as those considered in \cite{Morozov:2018eiq,Cassia:2020uxy,Cassia:2021uly}? What kind of topological recursion relation can one obtain from $q$-Virasoro constraints? What is the combinatorial meaning of the coefficients of the corresponding genus expansion?
\item It would appear from our analysis that one could come up with an ansatz for the time-dependent free energy of any $\beta$-deformed matrix model with polynomial potential
$V(x)=\frac{N}{m\lambda}x^m$. The natural generalization of the ansatz for the genus expansion would be
\begin{multline}
\label{eq:ansatz-higher-power-potential}
 F(N,\beta,\lambda,\bu)
 = \sum_{\ell=1}^\infty \sum_{g\in\frac12\mathbb{N}}
 \frac{N^{2-2g-2\ell} \beta^{1-\ell-2g}}{\ell!}
 \sum_{k_1,\dots,k_\ell=1}^\infty (\beta\lambda)^{\sum_{j=1}^\ell \frac{k_j}{m}} \\
 \times\sum_{i_1+i_2=2g}(-\beta)^{i_1} C_{g,[k_1,\dots,k_\ell]}^{(i_1,i_2)}
 \prod_{j=1}^\ell u_{k_j}
\end{multline}
for some appropriate coefficients $C_{g,[k_1,\dots,k_\ell]}^{(i_1,i_2)}$.
However, it is known from \cite{Morozov:2009xk,Cassia:2020uxy,Cassia:2021dpd} for example, that Virasoro constraints for $m\geq3$ do not admit a unique solution. For this reason, the polynomials $C_{g,[k_1,\dots,k_\ell]}(\beta)$ are not well-defined in this case, and we cannot guarantee that the ansatz \eqref{eq:ansatz-higher-power-potential} gives the correct genus expansion of the time-dependent free energy. It would be interesting to study this case further, as it seems to require different techniques from those used in this article.
\item In the case of the \GBE{}, we are able to give an explicit identification between Catalan polynomials and marginal $b$-polynomials as discussed in \cref{sec:hypermaps}. It is not known to us whether marginal $b$-polynomials can be defined also in the case of the \WLBE{} or not. Nevertheless, an analogue of the rooted map series \eqref{eq:seriesM} can be defined as
\be
\begin{aligned}
 M(N,b,\phi,\by) :=& H(\{x_k=N\},\by,\{z_k=N+b\phi\};b) \\
 =& \sum_{\mu,\nu,\lambda} c_{\mu,\nu,\lambda}(b)
 N^{\ell(\mu)} (N+b\phi)^{\ell(\lambda)} p_\nu(\by)
\end{aligned}
\ee
by using superintegrability \eqref{eq:character-expansion-WL}.
The obvious question now is: what is the function $M(N,b,\phi,\by)$ counting?
Moreover, what is the combinatorial meaning of the parameter $\phi$?
The answer might be related to a $\beta$-deformed version of the analysis in \cite{Cunden:2021int,Gisonni:2020lag} which relates moments of the \WLBE{} at $\beta=1$ and double monotone Hurwitz numbers.
\item Superintegrability of averages of Jack polynomials gives a closed formula for the character expansion of the generating function $Z(N,\beta,\lambda,\bu)$. Is there an equivalent reformulation of superintegrability for the function $F(N,\beta,\lambda,\bu)$? The existence of such a formula would give a closed form expression for all Catalan polynomials $C_{g,\nu}(\beta)$ and integer invariants $n_{\nu,\lambda}$.
A possible way to relate superintegrability to the topological expansion could be
via the W-representation of the matrix model as discussed in \cite{Mironov:2022xjc}.
\end{itemize}
We leave the investigation of these and other questions to future research.

\appendix

\section{Schur polynomials in two variables}
\label{sec:Schur}

We recall some properties of Schur symmetric polynomials in two variables.
See \cite{Macdonald:book} for a thorough discussion on symmetric functions.

Schur polynomials constitute a linear basis of the space of symmetric functions in $n$ variables.
They are labeled by integer partitions $\lambda$, which correspond to irreducible representations of $gl_n$. Whenever the length of the partition $\lambda$ is greater than the number of variables $n$, the corresponding Schur polynomial $s_\lambda$ vanishes identically. In the case $n=2$, the only nonzero Schur polynomials are those associated to
partitions of length at most two. Any such partition can be represented by an ordered set of two nonnegative integers, $[\lambda_1,\lambda_2]$ with $\lambda_1\geq\lambda_2\geq0$. We parametrize such partitions as $\lambda=[h+d,d]$ for $h,d\in\mathbb{N}$.

Let $\epsilon_1,\epsilon_2$ be the two arguments of the Schur polynomial.
Then one of the Pieri rule can be stated as
\be
 s_{[h+d,d]}(\epsilon_1,\epsilon_2)
 = s_{[1,1]}(\epsilon_1,\epsilon_2)^d s_{[h]}(\epsilon_1,\epsilon_2)
\ee
where
\be
 s_{[1,1]}(\epsilon_1,\epsilon_2)=\epsilon_1\epsilon_2
\ee
This implies that knowing $s_{[h]}(\epsilon_1,\epsilon_2)$ for any $h$ is enough to determine uniquely all other Schur polynomials.
Using the coproduct rule
\be
 s_{\lambda}(\epsilon_1,\epsilon_2)
 = \sum_{\nu\subset\lambda} s_{\lambda/\nu}(\epsilon_1)s_{\nu}(\epsilon_2),
\ee
we obtain the closed formula
\be
\label{eq:coproductSchur2}
 s_{[h+d,d]}(\epsilon_1,\epsilon_2)
 = \sum_{i+j=h} \epsilon_1^{d+i}\epsilon_2^{d+j}
 = (\epsilon_1\epsilon_2)^d\,
 \frac{\epsilon_1^{h+1}-\epsilon_2^{h+1}}{\epsilon_1-\epsilon_2}
\ee
As a corollary, we observe that
\be
 \frac{s_{[2h+1+d,d]}(\epsilon_1,\epsilon_2)}{s_{[1]}(\epsilon_1,\epsilon_2)}
 = (\epsilon_1\epsilon_2)^d\,
 \frac{\epsilon_1^{2(h+1)}-\epsilon_2^{2(h+1)}}{\epsilon_1^2-\epsilon_2^2}
\ee
where the r.h.s.\ is a symmetric polynomial in $\epsilon_1,\epsilon_2$ for all values of $d,h\in\mathbb{N}$. In other words, $s_{[1]}(\epsilon_1,\epsilon_2)$ always divides $s_{[2h+1+d,d]}(\epsilon_1,\epsilon_2)$.
This implies that if $|\lambda|$ is odd and $\epsilon_1+\epsilon_2=0$, then $s_\lambda(\epsilon_1,\epsilon_2)=0$.

\section{Hypermap counts and the \texorpdfstring{$b$}{b}-conjecture}
\label{sec:hypermaps}

In this section, we review some known results regarding counting hypermaps and the so-called
$b$-conjecture \cite{Goulden:1996con,Goulden:1996b}.
Our reference on the subject will be \cite{LaCroix:2009}.
After introducing the \emph{rooted hypermap series} $H$, \emph{rooted map series} $M$ and the \emph{marginal $b$-polynomials} $d_{\ell,\nu}(b)$,
we will give a formula to relate the polynomials $d_{\ell,\nu}(b)$ to the Catalan polynomials $C_{g,\nu}(\beta)$ that we have defined in \cref{sec:genus}.

Throughout this section, we adopt the notations of \cite{LaCroix:2009}, with the understanding that the parameters $b$ and $\beta$ are related by the equation
\be
\label{eq:bbeta}
 1+b=\beta^{-1}
\ee
Let us introduce the series
\be
 \Phi(\bx,\by,\bz;\beta) := \sum_{\lambda}
 \frac{\mathrm{JackP}_\lambda(\bx)
 \mathrm{JackQ}_\lambda(\by)
 \mathrm{JackP}_\lambda(\bz)}{\mathrm{JackP}_\lambda(p_k=\delta_{k,1})}
\ee
for three independent sets of times $\bx=\{x_1,x_2,\dots\}$, $\by=\{y_1,y_2,\dots\}$ and $\bz=\{z_1,z_2,\dots\}$.
\begin{definition}
The rooted Hypermap series is defined as
\be
 H(\bx,\by,\bz;b) := (1+b) D_{\by} \log\Phi(\bx,\by,\bz;\tfrac1{1+b})
\ee
with $D_{\by} = \sum_{k>0} k y_k\frac{\partial}{\partial y_k}$ being the degree operator in the $\by$ times.
\end{definition}
Then we can express $H(\bx,\by,\bz;b)$ in the power-sum basis
\be
 H(\bx,\by,\bz;b) = \sum_{\substack{\mu,\nu,\lambda\\|\mu|=|\nu|=|\lambda|}}
 c_{\mu,\nu,\lambda}(b) \, p_\mu(\bx) \, p_\nu(\by) \, p_\lambda(\bz)
\ee
with coefficients $c_{\mu,\nu,\lambda}(b)$ being the $b$-polynomials
implicitly defined by this expansion. Here, $p_\mu(\bx)=\prod_{k\in\mu}x_k$
and similarly for $p_\nu(\by)$, $p_\lambda(\bz)$.

\begin{definition}
The rooted map series $M(N,b,\by)$ is defined as a specialization of $H(\bx,\by,\bz;b)$,
\be
\label{eq:seriesM}
\begin{aligned}
 M(N,b,\by) :=& H(\{x_k=N\},\by,\{z_k=\delta_{k,2}\};b) \\
 =& \sum_{n=0}^\infty\sum_{\mu,\nu\vdash 2n} c_{\mu,\nu,[2^n]}(b) N^{\ell(\mu)}p_\nu(\by) \\
 =& \sum_{n=0}^\infty \sum_{\ell=0}^\infty N^\ell\sum_{\nu\vdash 2n} d_{\ell,\nu}(b) p_\nu(\by)
\end{aligned}
\ee
where the coefficients
\be
 d_{\ell,\nu}(b) = \sum_{\substack{\mu\vdash 2n\\\ell(\mu)=\ell}} c_{\mu,\nu,[2^n]}(b)
\ee
are the \emph{marginal $b$-polynomials}.
\end{definition}
The marginal $b$-conjecture (proven in \cite[Corollary 4.17]{LaCroix:2009}) states the following.
\begin{theorem}[La Croix]
The $d_{\ell,\nu}(b)$ are polynomials with non-negative integer coefficients.
Moreover, the coefficient of $b^k$ in $d_{\ell,\nu}(b)$ is the number of rooted maps with $\ell$ vertices, face-degree partition $\nu$, and a total of $k$ twisted edges.
\end{theorem}

We now claim that marginal polynomials $d_{\ell,\nu}(b)$ and Catalan polynomials $C_{g,\nu}(\beta)$ are related (up to overall combinatorial factors) by an appropriate identification of labels and variables as in \eqref{eq:bbeta}.
In order to show this, we first identify the corresponding generating functions.
From \eqref{eq:characterexpansionGBE}, we can identify the time-dependent free energy of the \GBE{} with a specialization of the logarithm of $\Phi(\bx,\by,\bz;\beta)$,
\be
 F(N,\beta,\lambda,\bu)
 = \log\Phi\Big(\{x_k=N\},\Big\{y_k=(\tfrac{\beta\lambda}{N})^{\frac{k}{2}}\frac{k u_k}{\beta N}\Big\},\{z_k=\delta_{k,2}\};\beta\Big)
\ee
which, together with \eqref{eq:Gdef}, implies
\be
 D_{\bv} G(\genus,\beta,\bv)
 = (\beta\genus)^{2} M\Big((\beta\genus)^{-1},\beta^{-1}-1,\{y_k=(\beta\genus)^{\frac{k}{2}-1} k v_k\}\Big)
\ee
Expanding both sides as power series as in \eqref{eq:seriesG} and \eqref{eq:seriesM}, we can identify the coefficients to finally obtain the relation
\be
 C_{g,\nu}(\beta)
 = |\Aut(\nu)| ~ \frac{\prod_{k\in\nu}k}{|\nu|}
 \beta^{2g}
 ~ d_{2-2g-\ell(\nu)+\frac{|\nu|}{2},\nu}(\beta^{-1}-1)\,.
\ee
This identification of polynomials is evidence of an underlying relation between the invariants $n_{\nu,\lambda}$ with the counting of locally orientable maps; however, the precise nature of this relation is unclear to us.

\section{Tables of Catalan polynomials}
\label{sec:tables}


\subsection{The \texorpdfstring{\GBE{}}{GBE} case}

See \cref{tblr:betapoly,tblr:SchurEpsilon}.

{\tiny
\begin{adjustwidth}{-1.4in}{-1.4in}
\begin{longtblr}[
theme = fancy,
caption = {Polynomials $C_{g,\nu}(\beta)$ for the \GBE},
label = {tblr:betapoly},
]{
colspec = {X[-1,c]X[-1,c]X[-1,c]X[-1,c]X[-1,c]X[-1,c]},
cells = {mode=dmath,font=\tiny},
rowsep = {1pt},
colsep = {2pt},
rowhead = 1,
hspan=minimal,
}
\hline
 \nu & g=0 & g=1/2 & g=1 & g=3/2 & g=2 \\
\hline
 \lbrack 1,1\rbrack & 1 & 0 & 0 & 0 & 0 \\
 \lbrack 2\rbrack & 1 & 1-\beta  & 0 & 0 & 0 \\
 \lbrack 1,1,1,1\rbrack & 0 & 0 & 0 & 0 & 0 \\
 \lbrack 2,1,1\rbrack & 2 & 0 & 0 & 0 & 0 \\
 \lbrack 2,2\rbrack & 2 & 2-2 \beta  & 0 & 0 & 0 \\
 \lbrack 3,1\rbrack & 3 & 3-3 \beta  & 0 & 0 & 0 \\
 \lbrack 4\rbrack & 2 & 5-5 \beta  & 3 \beta ^2-5 \beta +3 & 0 & 0 \\
 \lbrack 1,1,1,1,1,1\rbrack & 0 & 0 & 0 & 0 & 0 \\
 \lbrack 2,1,1,1,1\rbrack & 0 & 0 & 0 & 0 & 0 \\
 \lbrack 2,2,1,1\rbrack & 8 & 0 & 0 & 0 & 0 \\
 \lbrack 2,2,2\rbrack & 8 & 8-8 \beta  & 0 & 0 & 0 \\
 \lbrack 3,1,1,1\rbrack & 6 & 0 & 0 & 0 & 0 \\
 \lbrack 3,2,1\rbrack & 12 & 12-12 \beta  & 0 & 0 & 0 \\
 \lbrack 3,3\rbrack & 12 & 27-27 \beta  & 15 \beta ^2-27 \beta +15 & 0 & 0 \\
 \lbrack 4,1,1\rbrack & 12 & 12-12 \beta  & 0 & 0 & 0 \\
 \lbrack 4,2\rbrack & 8 & 20-20 \beta  & 12 \beta ^2-20 \beta +12 & 0 & 0 \\
 \lbrack 5,1\rbrack & 10 & 25-25 \beta  & 15 \beta ^2-25 \beta +15 & 0 & 0 \\
 \lbrack 6\rbrack & 5 & 22-22 \beta  & 32 \beta ^2-54 \beta +32 & -15 \beta ^3+32 \beta ^2-32 \beta +15 & 0 \\
 \lbrack 1,1,1,1,1,1,1,1\rbrack & 0 & 0 & 0 & 0 & 0 \\
 \lbrack 2,1,1,1,1,1,1\rbrack & 0 & 0 & 0 & 0 & 0 \\
 \lbrack 2,2,1,1,1,1\rbrack & 0 & 0 & 0 & 0 & 0 \\
 \lbrack 2,2,2,1,1\rbrack & 48 & 0 & 0 & 0 & 0 \\
 \lbrack 2,2,2,2\rbrack & 48 & 48-48 \beta  & 0 & 0 & 0 \\
 \lbrack 3,1,1,1,1,1\rbrack & 0 & 0 & 0 & 0 & 0 \\
 \lbrack 3,2,1,1,1\rbrack & 36 & 0 & 0 & 0 & 0 \\
 \lbrack 3,2,2,1\rbrack & 72 & 72-72 \beta  & 0 & 0 & 0 \\
 \lbrack 3,3,1,1\rbrack & 72 & 72-72 \beta  & 0 & 0 & 0 \\
 \lbrack 3,3,2\rbrack & 72 & 162-162 \beta  & 90 \beta ^2-162 \beta +90 & 0 & 0 \\
 \lbrack 4,1,1,1,1\rbrack & 24 & 0 & 0 & 0 & 0 \\
 \lbrack 4,2,1,1\rbrack & 72 & 72-72 \beta  & 0 & 0 & 0 \\
 \lbrack 4,2,2\rbrack & 48 & 120-120 \beta  & 72 \beta ^2-120 \beta +72 & 0 & 0 \\
 \lbrack 4,3,1\rbrack & 72 & 168-168 \beta  & 96 \beta ^2-168 \beta +96 & 0 & 0 \\
 \lbrack 4,4\rbrack & 36 & 152-152 \beta  & 212 \beta ^2-364 \beta +212 & -96 \beta ^3+212 \beta ^2-212 \beta +96 & 0 \\
 \lbrack 5,1,1,1\rbrack & 60 & 60-60 \beta  & 0 & 0 & 0 \\
 \lbrack 5,2,1\rbrack & 60 & 150-150 \beta  & 90 \beta ^2-150 \beta +90 & 0 & 0 \\
 \lbrack 5,3\rbrack & 45 & 180-180 \beta  & 240 \beta ^2-420 \beta +240 & -105 \beta ^3+240 \beta ^2-240 \beta +105 & 0 \\
 \lbrack 6,1,1\rbrack & 60 & 150-150 \beta  & 90 \beta ^2-150 \beta +90 & 0 & 0 \\
 \lbrack 6,2\rbrack & 30 & 132-132 \beta  & 192 \beta ^2-324 \beta +192 & -90 \beta ^3+192 \beta ^2-192 \beta +90 & 0 \\
 \lbrack 7,1\rbrack & 35 & 154-154 \beta  & 224 \beta ^2-378 \beta +224 & -105 \beta ^3+224 \beta ^2-224 \beta +105 & 0 \\
 \lbrack 8\rbrack & 14 & 93-93 \beta  & 234 \beta ^2-398 \beta +234 & -260 \beta ^3+565 \beta ^2-565 \beta +260 & 105 \beta ^4-260 \beta ^3+331 \beta ^2-260 \beta +105 \\
 \lbrack 1,1,1,1,1,1,1,1,1,1\rbrack & 0 & 0 & 0 & 0 & 0 \\
 \lbrack 2,1,1,1,1,1,1,1,1\rbrack & 0 & 0 & 0 & 0 & 0 \\
 \lbrack 2,2,1,1,1,1,1,1\rbrack & 0 & 0 & 0 & 0 & 0 \\
 \lbrack 2,2,2,1,1,1,1\rbrack & 0 & 0 & 0 & 0 & 0 \\
 \lbrack 2,2,2,2,1,1\rbrack & 384 & 0 & 0 & 0 & 0 \\
 \lbrack 2,2,2,2,2\rbrack & 384 & 384-384 \beta  & 0 & 0 & 0 \\
 \lbrack 3,1,1,1,1,1,1,1\rbrack & 0 & 0 & 0 & 0 & 0 \\
 \lbrack 3,2,1,1,1,1,1\rbrack & 0 & 0 & 0 & 0 & 0 \\
 \lbrack 3,2,2,1,1,1\rbrack & 288 & 0 & 0 & 0 & 0 \\
 \lbrack 3,2,2,2,1\rbrack & 576 & 576-576 \beta  & 0 & 0 & 0 \\
 \lbrack 3,3,1,1,1,1\rbrack & 216 & 0 & 0 & 0 & 0 \\
 \lbrack 3,3,2,1,1\rbrack & 576 & 576-576 \beta  & 0 & 0 & 0 \\
 \lbrack 3,3,2,2\rbrack & 576 & 1296-1296 \beta  & 720 \beta ^2-1296 \beta +720 & 0 & 0 \\
 \lbrack 3,3,3,1\rbrack & 648 & 1458-1458 \beta  & 810 \beta ^2-1458 \beta +810 & 0 & 0 \\
 \lbrack 4,1,1,1,1,1,1\rbrack & 0 & 0 & 0 & 0 & 0 \\
 \lbrack 4,2,1,1,1,1\rbrack & 192 & 0 & 0 & 0 & 0 \\
 \lbrack 4,2,2,1,1\rbrack & 576 & 576-576 \beta  & 0 & 0 & 0 \\
 \lbrack 4,2,2,2\rbrack & 384 & 960-960 \beta  & 576 \beta ^2-960 \beta +576 & 0 & 0 \\
 \lbrack 4,3,1,1,1\rbrack & 504 & 504-504 \beta  & 0 & 0 & 0 \\
 \lbrack 4,3,2,1\rbrack & 576 & 1344-1344 \beta  & 768 \beta ^2-1344 \beta +768 & 0 & 0 \\
 \lbrack 4,3,3\rbrack & 432 & 1656-1656 \beta  & 2124 \beta ^2-3780 \beta +2124 & -900 \beta ^3+2124 \beta ^2-2124 \beta +900 & 0 \\
 \lbrack 4,4,1,1\rbrack & 576 & 1344-1344 \beta  & 768 \beta ^2-1344 \beta +768 & 0 & 0 \\
 \lbrack 4,4,2\rbrack & 288 & 1216-1216 \beta  & 1696 \beta ^2-2912 \beta +1696 & -768 \beta ^3+1696 \beta ^2-1696 \beta +768 & 0 \\
 \lbrack 5,1,1,1,1,1\rbrack & 120 & 0 & 0 & 0 & 0 \\
 \lbrack 5,2,1,1,1\rbrack & 480 & 480-480 \beta  & 0 & 0 & 0 \\
 \lbrack 5,2,2,1\rbrack & 480 & 1200-1200 \beta  & 720 \beta ^2-1200 \beta +720 & 0 & 0 \\
 \lbrack 5,3,1,1\rbrack & 540 & 1290-1290 \beta  & 750 \beta ^2-1290 \beta +750 & 0 & 0 \\
 \lbrack 5,3,2\rbrack & 360 & 1440-1440 \beta  & 1920 \beta ^2-3360 \beta +1920 & -840 \beta ^3+1920 \beta ^2-1920 \beta +840 & 0 \\
 \lbrack 5,4,1\rbrack & 360 & 1480-1480 \beta  & 2020 \beta ^2-3500 \beta +2020 & -900 \beta ^3+2020 \beta ^2-2020 \beta +900 & 0 \\
 \lbrack 5,5\rbrack & 180 & 1075-1075 \beta  & 2450 \beta ^2-4300 \beta +2450 & -2500 \beta ^3+5725 \beta ^2-5725 \beta +2500 & 945 \beta ^4-2500 \beta ^3+3275 \beta ^2-2500 \beta +945 \\
 \lbrack 6,1,1,1,1\rbrack & 360 & 360-360 \beta  & 0 & 0 & 0 \\
 \lbrack 6,2,1,1\rbrack & 480 & 1200-1200 \beta  & 720 \beta ^2-1200 \beta +720 & 0 & 0 \\
 \lbrack 6,2,2\rbrack & 240 & 1056-1056 \beta  & 1536 \beta ^2-2592 \beta +1536 & -720 \beta ^3+1536 \beta ^2-1536 \beta +720 & 0 \\
 \lbrack 6,3,1\rbrack & 360 & 1476-1476 \beta  & 2016 \beta ^2-3492 \beta +2016 & -900 \beta ^3+2016 \beta ^2-2016 \beta +900 & 0 \\
 \lbrack 6,4\rbrack & 144 & 912-912 \beta  & 2184 \beta ^2-3768 \beta +2184 & -2316 \beta ^3+5172 \beta ^2-5172 \beta +2316 & 900 \beta ^4-2316 \beta ^3+2988 \beta ^2-2316 \beta +900 \\
 \lbrack 7,1,1,1\rbrack & 420 & 1050-1050 \beta  & 630 \beta ^2-1050 \beta +630 & 0 & 0 \\
 \lbrack 7,2,1\rbrack & 280 & 1232-1232 \beta  & 1792 \beta ^2-3024 \beta +1792 & -840 \beta ^3+1792 \beta ^2-1792 \beta +840 & 0 \\
 \lbrack 7,3\rbrack & 168 & 1029-1029 \beta  & 2394 \beta ^2-4158 \beta +2394 & -2478 \beta ^3+5607 \beta ^2-5607 \beta +2478 & 945 \beta ^4-2478 \beta ^3+3213 \beta ^2-2478 \beta +945 \\
 \lbrack 8,1,1\rbrack & 280 & 1232-1232 \beta  & 1792 \beta ^2-3024 \beta +1792 & -840 \beta ^3+1792 \beta ^2-1792 \beta +840 & 0 \\
 \lbrack 8,2\rbrack & 112 & 744-744 \beta  & 1872 \beta ^2-3184 \beta +1872 & -2080 \beta ^3+4520 \beta ^2-4520 \beta +2080 & 840 \beta ^4-2080 \beta ^3+2648 \beta ^2-2080 \beta +840 \\
 \lbrack 9,1\rbrack & 126 & 837-837 \beta  & 2106 \beta ^2-3582 \beta +2106 & -2340 \beta ^3+5085 \beta ^2-5085 \beta +2340 & 945 \beta ^4-2340 \beta ^3+2979 \beta ^2-2340 \beta +945 \\
 \lbrack 10\rbrack & 42 & 386-386 \beta  & 1450 \beta ^2-2480 \beta +1450 & -2750 \beta ^3+6050 \beta ^2-6050 \beta +2750 & 2589 \beta ^4-6545 \beta ^3+8395 \beta ^2-6545 \beta +2589 \\
 \lbrack 1,1,1,1,1,1,1,1,1,1,1,1\rbrack & 0 & 0 & 0 & 0 & 0 \\
 \lbrack 2,1,1,1,1,1,1,1,1,1,1\rbrack & 0 & 0 & 0 & 0 & 0 \\
 \lbrack 2,2,1,1,1,1,1,1,1,1\rbrack & 0 & 0 & 0 & 0 & 0 \\
 \lbrack 2,2,2,1,1,1,1,1,1\rbrack & 0 & 0 & 0 & 0 & 0 \\
 \lbrack 2,2,2,2,1,1,1,1\rbrack & 0 & 0 & 0 & 0 & 0 \\
 \lbrack 2,2,2,2,2,1,1\rbrack & 3840 & 0 & 0 & 0 & 0 \\
 \lbrack 2,2,2,2,2,2\rbrack & 3840 & 3840-3840 \beta  & 0 & 0 & 0 \\
 \lbrack 3,1,1,1,1,1,1,1,1,1\rbrack & 0 & 0 & 0 & 0 & 0 \\
 \lbrack 3,2,1,1,1,1,1,1,1\rbrack & 0 & 0 & 0 & 0 & 0 \\
 \lbrack 3,2,2,1,1,1,1,1\rbrack & 0 & 0 & 0 & 0 & 0 \\
 \lbrack 3,2,2,2,1,1,1\rbrack & 2880 & 0 & 0 & 0 & 0 \\
 \lbrack 3,2,2,2,2,1\rbrack & 5760 & 5760-5760 \beta  & 0 & 0 & 0 \\
 \lbrack 3,3,1,1,1,1,1,1\rbrack & 0 & 0 & 0 & 0 & 0 \\
 \lbrack 3,3,2,1,1,1,1\rbrack & 2160 & 0 & 0 & 0 & 0 \\
 \lbrack 3,3,2,2,1,1\rbrack & 5760 & 5760-5760 \beta  & 0 & 0 & 0 \\
 \lbrack 3,3,2,2,2\rbrack & 5760 & 12960-12960 \beta  & 7200 \beta ^2-12960 \beta +7200 & 0 & 0 \\
 \lbrack 3,3,3,1,1,1\rbrack & 5184 & 5184-5184 \beta  & 0 & 0 & 0 \\
 \lbrack 3,3,3,2,1\rbrack & 6480 & 14580-14580 \beta  & 8100 \beta ^2-14580 \beta +8100 & 0 & 0 \\
 \lbrack 3,3,3,3\rbrack & 5184 & 19116-19116 \beta  & 23652 \beta ^2-42768 \beta +23652 & -9720 \beta ^3+23652 \beta ^2-23652 \beta +9720 & 0 \\
 \lbrack 4,1,1,1,1,1,1,1,1\rbrack & 0 & 0 & 0 & 0 & 0 \\
 \lbrack 4,2,1,1,1,1,1,1\rbrack & 0 & 0 & 0 & 0 & 0 \\
 \lbrack 4,2,2,1,1,1,1\rbrack & 1920 & 0 & 0 & 0 & 0 \\
 \lbrack 4,2,2,2,1,1\rbrack & 5760 & 5760-5760 \beta  & 0 & 0 & 0 \\
 \lbrack 4,2,2,2,2\rbrack & 3840 & 9600-9600 \beta  & 5760 \beta ^2-9600 \beta +5760 & 0 & 0 \\
 \lbrack 4,3,1,1,1,1,1\rbrack & 1440 & 0 & 0 & 0 & 0 \\
 \lbrack 4,3,2,1,1,1\rbrack & 5040 & 5040-5040 \beta  & 0 & 0 & 0 \\
 \lbrack 4,3,2,2,1\rbrack & 5760 & 13440-13440 \beta  & 7680 \beta ^2-13440 \beta +7680 & 0 & 0 \\
 \lbrack 4,3,3,1,1\rbrack & 6048 & 13896-13896 \beta  & 7848 \beta ^2-13896 \beta +7848 & 0 & 0 \\
 \lbrack 4,3,3,2\rbrack & 4320 & 16560-16560 \beta  & 21240 \beta ^2-37800 \beta +21240 & -9000 \beta ^3+21240 \beta ^2-21240 \beta +9000 & 0 \\
 \lbrack 4,4,1,1,1,1\rbrack & 4032 & 4032-4032 \beta  & 0 & 0 & 0 \\
 \lbrack 4,4,2,1,1\rbrack & 5760 & 13440-13440 \beta  & 7680 \beta ^2-13440 \beta +7680 & 0 & 0 \\
 \lbrack 4,4,2,2\rbrack & 2880 & 12160-12160 \beta  & 16960 \beta ^2-29120 \beta +16960 & -7680 \beta ^3+16960 \beta ^2-16960 \beta +7680 & 0 \\
 \lbrack 4,4,3,1\rbrack & 4320 & 16896-16896 \beta  & 22080 \beta ^2-38976 \beta +22080 & -9504 \beta ^3+22080 \beta ^2-22080 \beta +9504 & 0 \\
 \lbrack 4,4,4\rbrack & 1728 & 10592-10592 \beta  & 24512 \beta ^2-42688 \beta +24512 & -25152 \beta ^3+57248 \beta ^2-57248 \beta +25152 & 9504 \beta ^4-25152 \beta ^3+32736 \beta ^2-25152 \beta +9504 \\
 \lbrack 5,1,1,1,1,1,1,1\rbrack & 0 & 0 & 0 & 0 & 0 \\
 \lbrack 5,2,1,1,1,1,1\rbrack & 1200 & 0 & 0 & 0 & 0 \\
 \lbrack 5,2,2,1,1,1\rbrack & 4800 & 4800-4800 \beta  & 0 & 0 & 0 \\
 \lbrack 5,2,2,2,1\rbrack & 4800 & 12000-12000 \beta  & 7200 \beta ^2-12000 \beta +7200 & 0 & 0 \\
 \lbrack 5,3,1,1,1,1\rbrack & 3960 & 3960-3960 \beta  & 0 & 0 & 0 \\
 \lbrack 5,3,2,1,1\rbrack & 5400 & 12900-12900 \beta  & 7500 \beta ^2-12900 \beta +7500 & 0 & 0 \\
 \lbrack 5,3,2,2\rbrack & 3600 & 14400-14400 \beta  & 19200 \beta ^2-33600 \beta +19200 & -8400 \beta ^3+19200 \beta ^2-19200 \beta +8400 & 0 \\
 \lbrack 5,3,3,1\rbrack & 4320 & 16920-16920 \beta  & 22140 \beta ^2-39060 \beta +22140 & -9540 \beta ^3+22140 \beta ^2-22140 \beta +9540 & 0 \\
 \lbrack 5,4,1,1,1\rbrack & 5040 & 11880-11880 \beta  & 6840 \beta ^2-11880 \beta +6840 & 0 & 0 \\
 \lbrack 5,4,2,1\rbrack & 3600 & 14800-14800 \beta  & 20200 \beta ^2-35000 \beta +20200 & -9000 \beta ^3+20200 \beta ^2-20200 \beta +9000 & 0 \\
 \lbrack 5,4,3\rbrack & 2160 & 12540-12540 \beta  & 27720 \beta ^2-48960 \beta +27720 & -27420 \beta ^3+63840 \beta ^2-63840 \beta +27420 & 10080 \beta ^4-27420 \beta ^3+36120 \beta ^2-27420 \beta +10080 \\
 \lbrack 5,5,1,1\rbrack & 3600 & 14800-14800 \beta  & 20200 \beta ^2-35000 \beta +20200 & -9000 \beta ^3+20200 \beta ^2-20200 \beta +9000 & 0 \\
 \lbrack 5,5,2\rbrack & 1800 & 10750-10750 \beta  & 24500 \beta ^2-43000 \beta +24500 & -25000 \beta ^3+57250 \beta ^2-57250 \beta +25000 & 9450 \beta ^4-25000 \beta ^3+32750 \beta ^2-25000 \beta +9450 \\
 \lbrack 6,1,1,1,1,1,1\rbrack & 720 & 0 & 0 & 0 & 0 \\
 \lbrack 6,2,1,1,1,1\rbrack & 3600 & 3600-3600 \beta  & 0 & 0 & 0 \\
 \lbrack 6,2,2,1,1\rbrack & 4800 & 12000-12000 \beta  & 7200 \beta ^2-12000 \beta +7200 & 0 & 0 \\
 \lbrack 6,2,2,2\rbrack & 2400 & 10560-10560 \beta  & 15360 \beta ^2-25920 \beta +15360 & -7200 \beta ^3+15360 \beta ^2-15360 \beta +7200 & 0 \\
 \lbrack 6,3,1,1,1\rbrack & 4680 & 11340-11340 \beta  & 6660 \beta ^2-11340 \beta +6660 & 0 & 0 \\
 \lbrack 6,3,2,1\rbrack & 3600 & 14760-14760 \beta  & 20160 \beta ^2-34920 \beta +20160 & -9000 \beta ^3+20160 \beta ^2-20160 \beta +9000 & 0 \\
 \lbrack 6,3,3\rbrack & 2160 & 12582-12582 \beta  & 27900 \beta ^2-49140 \beta +27900 & -27648 \beta ^3+64206 \beta ^2-64206 \beta +27648 & 10170 \beta ^4-27648 \beta ^3+36306 \beta ^2-27648 \beta +10170 \\
 \lbrack 6,4,1,1\rbrack & 3600 & 14784-14784 \beta  & 20184 \beta ^2-34968 \beta +20184 & -9000 \beta ^3+20184 \beta ^2-20184 \beta +9000 & 0 \\
 \lbrack 6,4,2\rbrack & 1440 & 9120-9120 \beta  & 21840 \beta ^2-37680 \beta +21840 & -23160 \beta ^3+51720 \beta ^2-51720 \beta +23160 & 9000 \beta ^4-23160 \beta ^3+29880 \beta ^2-23160 \beta +9000 \\
 \lbrack 6,5,1\rbrack & 1800 & 11010-11010 \beta  & 25620 \beta ^2-44640 \beta +25620 & -26580 \beta ^3+60210 \beta ^2-60210 \beta +26580 & 10170 \beta ^4-26580 \beta ^3+34590 \beta ^2-26580 \beta +10170 \\
 \lbrack 6,6\rbrack & 600 & 5184-5184 \beta  & 18276 \beta ^2-31752 \beta +18276 & -32616 \beta ^3+73812 \beta ^2-73812 \beta +32616 & 29094 \beta ^4-76338 \beta ^3+99258 \beta ^2-76338 \beta +29094 \\
 \lbrack 7,1,1,1,1,1\rbrack & 2520 & 2520-2520 \beta  & 0 & 0 & 0 \\
 \lbrack 7,2,1,1,1\rbrack & 4200 & 10500-10500 \beta  & 6300 \beta ^2-10500 \beta +6300 & 0 & 0 \\
 \lbrack 7,2,2,1\rbrack & 2800 & 12320-12320 \beta  & 17920 \beta ^2-30240 \beta +17920 & -8400 \beta ^3+17920 \beta ^2-17920 \beta +8400 & 0 \\
 \lbrack 7,3,1,1\rbrack & 3360 & 14028-14028 \beta  & 19488 \beta ^2-33516 \beta +19488 & -8820 \beta ^3+19488 \beta ^2-19488 \beta +8820 & 0 \\
 \lbrack 7,3,2\rbrack & 1680 & 10290-10290 \beta  & 23940 \beta ^2-41580 \beta +23940 & -24780 \beta ^3+56070 \beta ^2-56070 \beta +24780 & 9450 \beta ^4-24780 \beta ^3+32130 \beta ^2-24780 \beta +9450 \\
 \lbrack 7,4,1\rbrack & 1680 & 10500-10500 \beta  & 24864 \beta ^2-43008 \beta +24864 & -26124 \beta ^3+58632 \beta ^2-58632 \beta +26124 & 10080 \beta ^4-26124 \beta ^3+33768 \beta ^2-26124 \beta +10080 \\
 \lbrack 7,5\rbrack & 700 & 5810-5810 \beta  & 19810 \beta ^2-34720 \beta +19810 & -34440 \beta ^3+78960 \beta ^2-78960 \beta +34440 & 30135 \beta ^4-80185 \beta ^3+104895 \beta ^2-80185 \beta +30135 \\
 \lbrack 8,1,1,1,1\rbrack & 3360 & 8400-8400 \beta  & 5040 \beta ^2-8400 \beta +5040 & 0 & 0 \\
 \lbrack 8,2,1,1\rbrack & 2800 & 12320-12320 \beta  & 17920 \beta ^2-30240 \beta +17920 & -8400 \beta ^3+17920 \beta ^2-17920 \beta +8400 & 0 \\
 \lbrack 8,2,2\rbrack & 1120 & 7440-7440 \beta  & 18720 \beta ^2-31840 \beta +18720 & -20800 \beta ^3+45200 \beta ^2-45200 \beta +20800 & 8400 \beta ^4-20800 \beta ^3+26480 \beta ^2-20800 \beta +8400 \\
 \lbrack 8,3,1\rbrack & 1680 & 10464-10464 \beta  & 24768 \beta ^2-42816 \beta +24768 & -26064 \beta ^3+58416 \beta ^2-58416 \beta +26064 & 10080 \beta ^4-26064 \beta ^3+33648 \beta ^2-26064 \beta +10080 \\
 \lbrack 8,4\rbrack & 560 & 4912-4912 \beta  & 17576 \beta ^2-30392 \beta +17576 & -31776 \beta ^3+71376 \beta ^2-71376 \beta +31776 & 28632 \beta ^4-74528 \beta ^3+96552 \beta ^2-74528 \beta +28632 \\
 \lbrack 9,1,1,1\rbrack & 2520 & 11088-11088 \beta  & 16128 \beta ^2-27216 \beta +16128 & -7560 \beta ^3+16128 \beta ^2-16128 \beta +7560 & 0 \\
 \lbrack 9,2,1\rbrack & 1260 & 8370-8370 \beta  & 21060 \beta ^2-35820 \beta +21060 & -23400 \beta ^3+50850 \beta ^2-50850 \beta +23400 & 9450 \beta ^4-23400 \beta ^3+29790 \beta ^2-23400 \beta +9450 \\
 \lbrack 9,3\rbrack & 630 & 5400-5400 \beta  & 18936 \beta ^2-32832 \beta +18936 & -33642 \beta ^3+75996 \beta ^2-75996 \beta +33642 & 29871 \beta ^4-78435 \beta ^3+101853 \beta ^2-78435 \beta +29871 \\
 \lbrack 10,1,1\rbrack & 1260 & 8370-8370 \beta  & 21060 \beta ^2-35820 \beta +21060 & -23400 \beta ^3+50850 \beta ^2-50850 \beta +23400 & 9450 \beta ^4-23400 \beta ^3+29790 \beta ^2-23400 \beta +9450 \\
 \lbrack 10,2\rbrack & 420 & 3860-3860 \beta  & 14500 \beta ^2-24800 \beta +14500 & -27500 \beta ^3+60500 \beta ^2-60500 \beta +27500 & 25890 \beta ^4-65450 \beta ^3+83950 \beta ^2-65450 \beta +25890 \\
 \lbrack 11,1\rbrack & 462 & 4246-4246 \beta  & 15950 \beta ^2-27280 \beta +15950 & -30250 \beta ^3+66550 \beta ^2-66550 \beta +30250 & 28479 \beta ^4-71995 \beta ^3+92345 \beta ^2-71995 \beta +28479 \\
 \lbrack 12\rbrack & 132 & 1586-1586 \beta  & 8178 \beta ^2-14046 \beta +8178 & -22950 \beta ^3+50945 \beta ^2-50945 \beta +22950 & 36500 \beta ^4-93612 \beta ^3+120692 \beta ^2-93612 \beta +36500 \\
\hline
\end{longtblr}
\end{adjustwidth}
}
{\tiny
\begin{longtblr}[
theme = fancy,
caption = {Integers $n_{\nu,\lambda}$ for the \GBE},
label = {tblr:SchurEpsilon},
]{
colspec = {X[-1,c]X[-1,c]X[-1,c]X[-1,c]X[-1,c]X[-1,c]X[-1,c]X[-1,c]X[-1,c]X[-1,c]},
cells = {mode=dmath,font=\tiny},
rowsep = {1pt},
colsep = {2pt},
rowhead = 1,
hspan=minimal,
}
\hline
 \diagbox{\nu}{\lambda} & {[]} & {[1]} & {[1,1]} & {[2]} & {[2,1]} & {[3]} & {[2,2]} & {[3,1]} & {[4]} \\
\hline
 \lbrack 1,1 \rbrack  & 1 & 0 & 0 & 0 & 0 & 0 & 0 & 0 & 0 \\
 \lbrack 2 \rbrack  & 1 & 1 & 0 & 0 & 0 & 0 & 0 & 0 & 0 \\
 \lbrack 1,1,1,1 \rbrack  & 0 & 0 & 0 & 0 & 0 & 0 & 0 & 0 & 0 \\
 \lbrack 2,1,1 \rbrack  & 2 & 0 & 0 & 0 & 0 & 0 & 0 & 0 & 0 \\
 \lbrack 2,2 \rbrack  & 2 & 2 & 0 & 0 & 0 & 0 & 0 & 0 & 0 \\
 \lbrack 3,1 \rbrack  & 3 & 3 & 0 & 0 & 0 & 0 & 0 & 0 & 0 \\
 \lbrack 4 \rbrack  & 2 & 5 & 2 & 3 & 0 & 0 & 0 & 0 & 0 \\
 \lbrack 1,1,1,1,1,1 \rbrack  & 0 & 0 & 0 & 0 & 0 & 0 & 0 & 0 & 0 \\
 \lbrack 2,1,1,1,1 \rbrack  & 0 & 0 & 0 & 0 & 0 & 0 & 0 & 0 & 0 \\
 \lbrack 2,2,1,1 \rbrack  & 8 & 0 & 0 & 0 & 0 & 0 & 0 & 0 & 0 \\
 \lbrack 2,2,2 \rbrack  & 8 & 8 & 0 & 0 & 0 & 0 & 0 & 0 & 0 \\
 \lbrack 3,1,1,1 \rbrack  & 6 & 0 & 0 & 0 & 0 & 0 & 0 & 0 & 0 \\
 \lbrack 3,2,1 \rbrack  & 12 & 12 & 0 & 0 & 0 & 0 & 0 & 0 & 0 \\
 \lbrack 3,3 \rbrack  & 12 & 27 & 12 & 15 & 0 & 0 & 0 & 0 & 0 \\
 \lbrack 4,1,1 \rbrack  & 12 & 12 & 0 & 0 & 0 & 0 & 0 & 0 & 0 \\
 \lbrack 4,2 \rbrack  & 8 & 20 & 8 & 12 & 0 & 0 & 0 & 0 & 0 \\
 \lbrack 5,1 \rbrack  & 10 & 25 & 10 & 15 & 0 & 0 & 0 & 0 & 0 \\
 \lbrack 6 \rbrack  & 5 & 22 & 22 & 32 & 17 & 15 & 0 & 0 & 0 \\
 \lbrack 1,1,1,1,1,1,1,1 \rbrack  & 0 & 0 & 0 & 0 & 0 & 0 & 0 & 0 & 0 \\
 \lbrack 2,1,1,1,1,1,1 \rbrack  & 0 & 0 & 0 & 0 & 0 & 0 & 0 & 0 & 0 \\
 \lbrack 2,2,1,1,1,1 \rbrack  & 0 & 0 & 0 & 0 & 0 & 0 & 0 & 0 & 0 \\
 \lbrack 2,2,2,1,1 \rbrack  & 48 & 0 & 0 & 0 & 0 & 0 & 0 & 0 & 0 \\
 \lbrack 2,2,2,2 \rbrack  & 48 & 48 & 0 & 0 & 0 & 0 & 0 & 0 & 0 \\
 \lbrack 3,1,1,1,1,1 \rbrack  & 0 & 0 & 0 & 0 & 0 & 0 & 0 & 0 & 0 \\
 \lbrack 3,2,1,1,1 \rbrack  & 36 & 0 & 0 & 0 & 0 & 0 & 0 & 0 & 0 \\
 \lbrack 3,2,2,1 \rbrack  & 72 & 72 & 0 & 0 & 0 & 0 & 0 & 0 & 0 \\
 \lbrack 3,3,1,1 \rbrack  & 72 & 72 & 0 & 0 & 0 & 0 & 0 & 0 & 0 \\
 \lbrack 3,3,2 \rbrack  & 72 & 162 & 72 & 90 & 0 & 0 & 0 & 0 & 0 \\
 \lbrack 4,1,1,1,1 \rbrack  & 24 & 0 & 0 & 0 & 0 & 0 & 0 & 0 & 0 \\
 \lbrack 4,2,1,1 \rbrack  & 72 & 72 & 0 & 0 & 0 & 0 & 0 & 0 & 0 \\
 \lbrack 4,2,2 \rbrack  & 48 & 120 & 48 & 72 & 0 & 0 & 0 & 0 & 0 \\
 \lbrack 4,3,1 \rbrack  & 72 & 168 & 72 & 96 & 0 & 0 & 0 & 0 & 0 \\
 \lbrack 4,4 \rbrack  & 36 & 152 & 152 & 212 & 116 & 96 & 0 & 0 & 0 \\
 \lbrack 5,1,1,1 \rbrack  & 60 & 60 & 0 & 0 & 0 & 0 & 0 & 0 & 0 \\
 \lbrack 5,2,1 \rbrack  & 60 & 150 & 60 & 90 & 0 & 0 & 0 & 0 & 0 \\
 \lbrack 5,3 \rbrack  & 45 & 180 & 180 & 240 & 135 & 105 & 0 & 0 & 0 \\
 \lbrack 6,1,1 \rbrack  & 60 & 150 & 60 & 90 & 0 & 0 & 0 & 0 & 0 \\
 \lbrack 6,2 \rbrack  & 30 & 132 & 132 & 192 & 102 & 90 & 0 & 0 & 0 \\
 \lbrack 7,1 \rbrack  & 35 & 154 & 154 & 224 & 119 & 105 & 0 & 0 & 0 \\
 \lbrack 8 \rbrack  & 14 & 93 & 164 & 234 & 305 & 260 & 71 & 155 & 105 \\
 \lbrack 1,1,1,1,1,1,1,1,1,1 \rbrack  & 0 & 0 & 0 & 0 & 0 & 0 & 0 & 0 & 0 \\
 \lbrack 2,1,1,1,1,1,1,1,1 \rbrack  & 0 & 0 & 0 & 0 & 0 & 0 & 0 & 0 & 0 \\
 \lbrack 2,2,1,1,1,1,1,1 \rbrack  & 0 & 0 & 0 & 0 & 0 & 0 & 0 & 0 & 0 \\
 \lbrack 2,2,2,1,1,1,1 \rbrack  & 0 & 0 & 0 & 0 & 0 & 0 & 0 & 0 & 0 \\
 \lbrack 2,2,2,2,1,1 \rbrack  & 384 & 0 & 0 & 0 & 0 & 0 & 0 & 0 & 0 \\
 \lbrack 2,2,2,2,2 \rbrack  & 384 & 384 & 0 & 0 & 0 & 0 & 0 & 0 & 0 \\
 \lbrack 3,1,1,1,1,1,1,1 \rbrack  & 0 & 0 & 0 & 0 & 0 & 0 & 0 & 0 & 0 \\
 \lbrack 3,2,1,1,1,1,1 \rbrack  & 0 & 0 & 0 & 0 & 0 & 0 & 0 & 0 & 0 \\
 \lbrack 3,2,2,1,1,1 \rbrack  & 288 & 0 & 0 & 0 & 0 & 0 & 0 & 0 & 0 \\
 \lbrack 3,2,2,2,1 \rbrack  & 576 & 576 & 0 & 0 & 0 & 0 & 0 & 0 & 0 \\
 \lbrack 3,3,1,1,1,1 \rbrack  & 216 & 0 & 0 & 0 & 0 & 0 & 0 & 0 & 0 \\
 \lbrack 3,3,2,1,1 \rbrack  & 576 & 576 & 0 & 0 & 0 & 0 & 0 & 0 & 0 \\
 \lbrack 3,3,2,2 \rbrack  & 576 & 1296 & 576 & 720 & 0 & 0 & 0 & 0 & 0 \\
 \lbrack 3,3,3,1 \rbrack  & 648 & 1458 & 648 & 810 & 0 & 0 & 0 & 0 & 0 \\
 \lbrack 4,1,1,1,1,1,1 \rbrack  & 0 & 0 & 0 & 0 & 0 & 0 & 0 & 0 & 0 \\
 \lbrack 4,2,1,1,1,1 \rbrack  & 192 & 0 & 0 & 0 & 0 & 0 & 0 & 0 & 0 \\
 \lbrack 4,2,2,1,1 \rbrack  & 576 & 576 & 0 & 0 & 0 & 0 & 0 & 0 & 0 \\
 \lbrack 4,2,2,2 \rbrack  & 384 & 960 & 384 & 576 & 0 & 0 & 0 & 0 & 0 \\
 \lbrack 4,3,1,1,1 \rbrack  & 504 & 504 & 0 & 0 & 0 & 0 & 0 & 0 & 0 \\
 \lbrack 4,3,2,1 \rbrack  & 576 & 1344 & 576 & 768 & 0 & 0 & 0 & 0 & 0 \\
 \lbrack 4,3,3 \rbrack  & 432 & 1656 & 1656 & 2124 & 1224 & 900 & 0 & 0 & 0 \\
 \lbrack 4,4,1,1 \rbrack  & 576 & 1344 & 576 & 768 & 0 & 0 & 0 & 0 & 0 \\
 \lbrack 4,4,2 \rbrack  & 288 & 1216 & 1216 & 1696 & 928 & 768 & 0 & 0 & 0 \\
 \lbrack 5,1,1,1,1,1 \rbrack  & 120 & 0 & 0 & 0 & 0 & 0 & 0 & 0 & 0 \\
 \lbrack 5,2,1,1,1 \rbrack  & 480 & 480 & 0 & 0 & 0 & 0 & 0 & 0 & 0 \\
 \lbrack 5,2,2,1 \rbrack  & 480 & 1200 & 480 & 720 & 0 & 0 & 0 & 0 & 0 \\
 \lbrack 5,3,1,1 \rbrack  & 540 & 1290 & 540 & 750 & 0 & 0 & 0 & 0 & 0 \\
 \lbrack 5,3,2 \rbrack  & 360 & 1440 & 1440 & 1920 & 1080 & 840 & 0 & 0 & 0 \\
 \lbrack 5,4,1 \rbrack  & 360 & 1480 & 1480 & 2020 & 1120 & 900 & 0 & 0 & 0 \\
 \lbrack 5,5 \rbrack  & 180 & 1075 & 1850 & 2450 & 3225 & 2500 & 775 & 1555 & 945 \\
 \lbrack 6,1,1,1,1 \rbrack  & 360 & 360 & 0 & 0 & 0 & 0 & 0 & 0 & 0 \\
 \lbrack 6,2,1,1 \rbrack  & 480 & 1200 & 480 & 720 & 0 & 0 & 0 & 0 & 0 \\
 \lbrack 6,2,2 \rbrack  & 240 & 1056 & 1056 & 1536 & 816 & 720 & 0 & 0 & 0 \\
 \lbrack 6,3,1 \rbrack  & 360 & 1476 & 1476 & 2016 & 1116 & 900 & 0 & 0 & 0 \\
 \lbrack 6,4 \rbrack  & 144 & 912 & 1584 & 2184 & 2856 & 2316 & 672 & 1416 & 900 \\
 \lbrack 7,1,1,1 \rbrack  & 420 & 1050 & 420 & 630 & 0 & 0 & 0 & 0 & 0 \\
 \lbrack 7,2,1 \rbrack  & 280 & 1232 & 1232 & 1792 & 952 & 840 & 0 & 0 & 0 \\
 \lbrack 7,3 \rbrack  & 168 & 1029 & 1764 & 2394 & 3129 & 2478 & 735 & 1533 & 945 \\
 \lbrack 8,1,1 \rbrack  & 280 & 1232 & 1232 & 1792 & 952 & 840 & 0 & 0 & 0 \\
 \lbrack 8,2 \rbrack  & 112 & 744 & 1312 & 1872 & 2440 & 2080 & 568 & 1240 & 840 \\
 \lbrack 9,1 \rbrack  & 126 & 837 & 1476 & 2106 & 2745 & 2340 & 639 & 1395 & 945 \\
 \lbrack 10 \rbrack  & 42 & 386 & 1030 & 1450 & 3300 & 2750 & 1850 & 3956 & 2589 \\
 \lbrack 1,1,1,1,1,1,1,1,1,1,1,1 \rbrack & 0 & 0 & 0 & 0 & 0 & 0 & 0 & 0 & 0 \\
 \lbrack 2,1,1,1,1,1,1,1,1,1,1 \rbrack & 0 & 0 & 0 & 0 & 0 & 0 & 0 & 0 & 0 \\
 \lbrack 2,2,1,1,1,1,1,1,1,1 \rbrack & 0 & 0 & 0 & 0 & 0 & 0 & 0 & 0 & 0 \\
 \lbrack 2,2,2,1,1,1,1,1,1 \rbrack & 0 & 0 & 0 & 0 & 0 & 0 & 0 & 0 & 0 \\
 \lbrack 2,2,2,2,1,1,1,1 \rbrack & 0 & 0 & 0 & 0 & 0 & 0 & 0 & 0 & 0 \\
 \lbrack 2,2,2,2,2,1,1 \rbrack & 3840 & 0 & 0 & 0 & 0 & 0 & 0 & 0 & 0 \\
 \lbrack 2,2,2,2,2,2 \rbrack & 3840 & 3840 & 0 & 0 & 0 & 0 & 0 & 0 & 0 \\
 \lbrack 3,1,1,1,1,1,1,1,1,1 \rbrack & 0 & 0 & 0 & 0 & 0 & 0 & 0 & 0 & 0 \\
 \lbrack 3,2,1,1,1,1,1,1,1 \rbrack & 0 & 0 & 0 & 0 & 0 & 0 & 0 & 0 & 0 \\
 \lbrack 3,2,2,1,1,1,1,1 \rbrack & 0 & 0 & 0 & 0 & 0 & 0 & 0 & 0 & 0 \\
 \lbrack 3,2,2,2,1,1,1 \rbrack & 2880 & 0 & 0 & 0 & 0 & 0 & 0 & 0 & 0 \\
 \lbrack 3,2,2,2,2,1 \rbrack & 5760 & 5760 & 0 & 0 & 0 & 0 & 0 & 0 & 0 \\
 \lbrack 3,3,1,1,1,1,1,1 \rbrack & 0 & 0 & 0 & 0 & 0 & 0 & 0 & 0 & 0 \\
 \lbrack 3,3,2,1,1,1,1 \rbrack & 2160 & 0 & 0 & 0 & 0 & 0 & 0 & 0 & 0 \\
 \lbrack 3,3,2,2,1,1 \rbrack & 5760 & 5760 & 0 & 0 & 0 & 0 & 0 & 0 & 0 \\
 \lbrack 3,3,2,2,2 \rbrack & 5760 & 12960 & 5760 & 7200 & 0 & 0 & 0 & 0 & 0 \\
 \lbrack 3,3,3,1,1,1 \rbrack & 5184 & 5184 & 0 & 0 & 0 & 0 & 0 & 0 & 0 \\
 \lbrack 3,3,3,2,1 \rbrack & 6480 & 14580 & 6480 & 8100 & 0 & 0 & 0 & 0 & 0 \\
 \lbrack 3,3,3,3 \rbrack & 5184 & 19116 & 19116 & 23652 & 13932 & 9720 & 0 & 0 & 0 \\
 \lbrack 4,1,1,1,1,1,1,1,1 \rbrack & 0 & 0 & 0 & 0 & 0 & 0 & 0 & 0 & 0 \\
 \lbrack 4,2,1,1,1,1,1,1 \rbrack & 0 & 0 & 0 & 0 & 0 & 0 & 0 & 0 & 0 \\
 \lbrack 4,2,2,1,1,1,1 \rbrack & 1920 & 0 & 0 & 0 & 0 & 0 & 0 & 0 & 0 \\
 \lbrack 4,2,2,2,1,1 \rbrack & 5760 & 5760 & 0 & 0 & 0 & 0 & 0 & 0 & 0 \\
 \lbrack 4,2,2,2,2 \rbrack & 3840 & 9600 & 3840 & 5760 & 0 & 0 & 0 & 0 & 0 \\
 \lbrack 4,3,1,1,1,1,1 \rbrack & 1440 & 0 & 0 & 0 & 0 & 0 & 0 & 0 & 0 \\
 \lbrack 4,3,2,1,1,1 \rbrack & 5040 & 5040 & 0 & 0 & 0 & 0 & 0 & 0 & 0 \\
 \lbrack 4,3,2,2,1 \rbrack & 5760 & 13440 & 5760 & 7680 & 0 & 0 & 0 & 0 & 0 \\
 \lbrack 4,3,3,1,1 \rbrack & 6048 & 13896 & 6048 & 7848 & 0 & 0 & 0 & 0 & 0 \\
 \lbrack 4,3,3,2 \rbrack & 4320 & 16560 & 16560 & 21240 & 12240 & 9000 & 0 & 0 & 0 \\
 \lbrack 4,4,1,1,1,1 \rbrack & 4032 & 4032 & 0 & 0 & 0 & 0 & 0 & 0 & 0 \\
 \lbrack 4,4,2,1,1 \rbrack & 5760 & 13440 & 5760 & 7680 & 0 & 0 & 0 & 0 & 0 \\
 \lbrack 4,4,2,2 \rbrack & 2880 & 12160 & 12160 & 16960 & 9280 & 7680 & 0 & 0 & 0 \\
 \lbrack 4,4,3,1 \rbrack & 4320 & 16896 & 16896 & 22080 & 12576 & 9504 & 0 & 0 & 0 \\
 \lbrack 4,4,4 \rbrack & 1728 & 10592 & 18176 & 24512 & 32096 & 25152 & 7584 & 15648 & 9504 \\
 \lbrack 5,1,1,1,1,1,1,1 \rbrack & 0 & 0 & 0 & 0 & 0 & 0 & 0 & 0 & 0 \\
 \lbrack 5,2,1,1,1,1,1 \rbrack & 1200 & 0 & 0 & 0 & 0 & 0 & 0 & 0 & 0 \\
 \lbrack 5,2,2,1,1,1 \rbrack & 4800 & 4800 & 0 & 0 & 0 & 0 & 0 & 0 & 0 \\
 \lbrack 5,2,2,2,1 \rbrack & 4800 & 12000 & 4800 & 7200 & 0 & 0 & 0 & 0 & 0 \\
 \lbrack 5,3,1,1,1,1 \rbrack & 3960 & 3960 & 0 & 0 & 0 & 0 & 0 & 0 & 0 \\
 \lbrack 5,3,2,1,1 \rbrack & 5400 & 12900 & 5400 & 7500 & 0 & 0 & 0 & 0 & 0 \\
 \lbrack 5,3,2,2 \rbrack & 3600 & 14400 & 14400 & 19200 & 10800 & 8400 & 0 & 0 & 0 \\
 \lbrack 5,3,3,1 \rbrack & 4320 & 16920 & 16920 & 22140 & 12600 & 9540 & 0 & 0 & 0 \\
 \lbrack 5,4,1,1,1 \rbrack & 5040 & 11880 & 5040 & 6840 & 0 & 0 & 0 & 0 & 0 \\
 \lbrack 5,4,2,1 \rbrack & 3600 & 14800 & 14800 & 20200 & 11200 & 9000 & 0 & 0 & 0 \\
 \lbrack 5,4,3 \rbrack & 2160 & 12540 & 21240 & 27720 & 36420 & 27420 & 8700 & 17340 & 10080 \\
 \lbrack 5,5,1,1 \rbrack & 3600 & 14800 & 14800 & 20200 & 11200 & 9000 & 0 & 0 & 0 \\
 \lbrack 5,5,2 \rbrack & 1800 & 10750 & 18500 & 24500 & 32250 & 25000 & 7750 & 15550 & 9450 \\
 \lbrack 6,1,1,1,1,1,1 \rbrack & 720 & 0 & 0 & 0 & 0 & 0 & 0 & 0 & 0 \\
 \lbrack 6,2,1,1,1,1 \rbrack & 3600 & 3600 & 0 & 0 & 0 & 0 & 0 & 0 & 0 \\
 \lbrack 6,2,2,1,1 \rbrack & 4800 & 12000 & 4800 & 7200 & 0 & 0 & 0 & 0 & 0 \\
 \lbrack 6,2,2,2 \rbrack & 2400 & 10560 & 10560 & 15360 & 8160 & 7200 & 0 & 0 & 0 \\
 \lbrack 6,3,1,1,1 \rbrack & 4680 & 11340 & 4680 & 6660 & 0 & 0 & 0 & 0 & 0 \\
 \lbrack 6,3,2,1 \rbrack & 3600 & 14760 & 14760 & 20160 & 11160 & 9000 & 0 & 0 & 0 \\
 \lbrack 6,3,3 \rbrack & 2160 & 12582 & 21240 & 27900 & 36558 & 27648 & 8658 & 17478 & 10170 \\
 \lbrack 6,4,1,1 \rbrack & 3600 & 14784 & 14784 & 20184 & 11184 & 9000 & 0 & 0 & 0 \\
 \lbrack 6,4,2 \rbrack & 1440 & 9120 & 15840 & 21840 & 28560 & 23160 & 6720 & 14160 & 9000 \\
 \lbrack 6,5,1 \rbrack & 1800 & 11010 & 19020 & 25620 & 33630 & 26580 & 8010 & 16410 & 10170 \\
 \lbrack 6,6 \rbrack & 600 & 5184 & 13476 & 18276 & 41196 & 32616 & 22920 & 47244 & 29094 \\
 \lbrack 7,1,1,1,1,1 \rbrack & 2520 & 2520 & 0 & 0 & 0 & 0 & 0 & 0 & 0 \\
 \lbrack 7,2,1,1,1 \rbrack & 4200 & 10500 & 4200 & 6300 & 0 & 0 & 0 & 0 & 0 \\
 \lbrack 7,2,2,1 \rbrack & 2800 & 12320 & 12320 & 17920 & 9520 & 8400 & 0 & 0 & 0 \\
 \lbrack 7,3,1,1 \rbrack & 3360 & 14028 & 14028 & 19488 & 10668 & 8820 & 0 & 0 & 0 \\
 \lbrack 7,3,2 \rbrack & 1680 & 10290 & 17640 & 23940 & 31290 & 24780 & 7350 & 15330 & 9450 \\
 \lbrack 7,4,1 \rbrack & 1680 & 10500 & 18144 & 24864 & 32508 & 26124 & 7644 & 16044 & 10080 \\
 \lbrack 7,5 \rbrack & 700 & 5810 & 14910 & 19810 & 44520 & 34440 & 24710 & 50050 & 30135 \\
 \lbrack 8,1,1,1,1 \rbrack & 3360 & 8400 & 3360 & 5040 & 0 & 0 & 0 & 0 & 0 \\
 \lbrack 8,2,1,1 \rbrack & 2800 & 12320 & 12320 & 17920 & 9520 & 8400 & 0 & 0 & 0 \\
 \lbrack 8,2,2 \rbrack & 1120 & 7440 & 13120 & 18720 & 24400 & 20800 & 5680 & 12400 & 8400 \\
 \lbrack 8,3,1 \rbrack & 1680 & 10464 & 18048 & 24768 & 32352 & 26064 & 7584 & 15984 & 10080 \\
 \lbrack 8,4 \rbrack & 560 & 4912 & 12816 & 17576 & 39600 & 31776 & 22024 & 45896 & 28632 \\
 \lbrack 9,1,1,1 \rbrack & 2520 & 11088 & 11088 & 16128 & 8568 & 7560 & 0 & 0 & 0 \\
 \lbrack 9,2,1 \rbrack & 1260 & 8370 & 14760 & 21060 & 27450 & 23400 & 6390 & 13950 & 9450 \\
 \lbrack 9,3 \rbrack & 630 & 5400 & 13896 & 18936 & 42354 & 33642 & 23418 & 48564 & 29871 \\
 \lbrack 10,1,1 \rbrack & 1260 & 8370 & 14760 & 21060 & 27450 & 23400 & 6390 & 13950 & 9450 \\
 \lbrack 10,2 \rbrack & 420 & 3860 & 10300 & 14500 & 33000 & 27500 & 18500 & 39560 & 25890 \\
 \lbrack 11,1 \rbrack & 462 & 4246 & 11330 & 15950 & 36300 & 30250 & 20350 & 43516 & 28479 \\
 \lbrack 12 \rbrack & 132 & 1586 & 5868 & 8178 & 27995 & 22950 & 27080 & 57112 & 36500 \\
\hline
\end{longtblr}
}

\subsection{The \texorpdfstring{\WLBE{}}{WLBE} case}

See \cref{tblr:betapolyWL,tblr:SchurEpsilonWLphi,tblr:SchurEpsilonWL}.

{\tiny
\begin{adjustwidth}{-1.4in}{-1.4in}
\begin{longtblr}[
theme = fancy,
caption = {Polynomials $C_{g,\nu}(\beta,\phi)$ for the \WLBE{} at $\phi=1$},
label = {tblr:betapolyWL},
]{
colspec = {X[-1,c]X[-1,c]X[-1,c]X[-1,c]X[-1,c]X[-1,c]},
cells = {mode=dmath,font=\tiny},
rowsep = {1pt},
colsep = {2pt},
rowhead = 1,
hspan=minimal,
}
\hline
 \nu & g=0 & g=1/2 & g=1 & g=3/2 & g=2 \\
\hline
 \lbrack 1 \rbrack & 1 & 1-\beta  & 0 & 0 & 0 \\
 \lbrack 1,1 \rbrack & 1 & 1-\beta  & 0 & 0 & 0 \\
 \lbrack 2 \rbrack & 2 & 4-4 \beta  & 2 \beta ^2-4 \beta +2 & 0 & 0 \\
 \lbrack 1,1,1 \rbrack & 2 & 2-2 \beta  & 0 & 0 & 0 \\
 \lbrack 2,1 \rbrack & 4 & 8-8 \beta  & 4 \beta ^2-8 \beta +4 & 0 & 0 \\
 \lbrack 3 \rbrack & 5 & 16-16 \beta  & 17 \beta ^2-33 \beta +17 & -6 \beta ^3+17 \beta ^2-17 \beta +6 & 0 \\
 \lbrack 1,1,1,1 \rbrack & 6 & 6-6 \beta  & 0 & 0 & 0 \\
 \lbrack 2,1,1 \rbrack & 12 & 24-24 \beta  & 12 \beta ^2-24 \beta +12 & 0 & 0 \\
 \lbrack 2,2 \rbrack & 18 & 56-56 \beta  & 58 \beta ^2-114 \beta +58 & -20 \beta ^3+58 \beta ^2-58 \beta +20 & 0 \\
 \lbrack 3,1 \rbrack & 15 & 48-48 \beta  & 51 \beta ^2-99 \beta +51 & -18 \beta ^3+51 \beta ^2-51 \beta +18 & 0 \\
 \lbrack 4 \rbrack & 14 & 64-64 \beta  & 110 \beta ^2-210 \beta +110 & -84 \beta ^3+230 \beta ^2-230 \beta +84 & 24 \beta ^4-84 \beta ^3+120 \beta ^2-84 \beta +24 \\
 \lbrack 1,1,1,1,1 \rbrack & 24 & 24-24 \beta  & 0 & 0 & 0 \\
 \lbrack 2,1,1,1 \rbrack & 48 & 96-96 \beta  & 48 \beta ^2-96 \beta +48 & 0 & 0 \\
 \lbrack 2,2,1 \rbrack & 72 & 224-224 \beta  & 232 \beta ^2-456 \beta +232 & -80 \beta ^3+232 \beta ^2-232 \beta +80 & 0 \\
 \lbrack 3,1,1 \rbrack & 60 & 192-192 \beta  & 204 \beta ^2-396 \beta +204 & -72 \beta ^3+204 \beta ^2-204 \beta +72 & 0 \\
 \lbrack 3,2 \rbrack & 72 & 318-318 \beta  & 528 \beta ^2-1020 \beta +528 & -390 \beta ^3+1092 \beta ^2-1092 \beta +390 & 108 \beta ^4-390 \beta ^3+564 \beta ^2-390 \beta +108 \\
 \lbrack 4,1 \rbrack & 56 & 256-256 \beta  & 440 \beta ^2-840 \beta +440 & -336 \beta ^3+920 \beta ^2-920 \beta +336 & 96 \beta ^4-336 \beta ^3+480 \beta ^2-336 \beta +96 \\
 \lbrack 5 \rbrack & 42 & 256-256 \beta  & 630 \beta ^2-1190 \beta +630 & -780 \beta ^3+2090 \beta ^2-2090 \beta +780 & 484 \beta ^4-1640 \beta ^3+2320 \beta ^2-1640 \beta +484 \\
 \lbrack 1,1,1,1,1,1 \rbrack & 120 & 120-120 \beta  & 0 & 0 & 0 \\
 \lbrack 2,1,1,1,1 \rbrack & 240 & 480-480 \beta  & 240 \beta ^2-480 \beta +240 & 0 & 0 \\
 \lbrack 2,2,1,1 \rbrack & 360 & 1120-1120 \beta  & 1160 \beta ^2-2280 \beta +1160 & -400 \beta ^3+1160 \beta ^2-1160 \beta +400 & 0 \\
 \lbrack 2,2,2 \rbrack & 432 & 1864-1864 \beta  & 3024 \beta ^2-5888 \beta +3024 & -2184 \beta ^3+6208 \beta ^2-6208 \beta +2184 & 592 \beta ^4-2184 \beta ^3+3184 \beta ^2-2184 \beta +592 \\
 \lbrack 3,1,1,1 \rbrack & 300 & 960-960 \beta  & 1020 \beta ^2-1980 \beta +1020 & -360 \beta ^3+1020 \beta ^2-1020 \beta +360 & 0 \\
 \lbrack 3,2,1 \rbrack & 360 & 1590-1590 \beta  & 2640 \beta ^2-5100 \beta +2640 & -1950 \beta ^3+5460 \beta ^2-5460 \beta +1950 & 540 \beta ^4-1950 \beta ^3+2820 \beta ^2-1950 \beta +540 \\
 \lbrack 3,3 \rbrack & 300 & 1746-1746 \beta  & 4098 \beta ^2-7848 \beta +4098 & -4842 \beta ^3+13308 \beta ^2-13308 \beta +4842 & 2874 \beta ^4-10080 \beta ^3+14448 \beta ^2-10080 \beta +2874 \\
 \lbrack 4,1,1 \rbrack & 280 & 1280-1280 \beta  & 2200 \beta ^2-4200 \beta +2200 & -1680 \beta ^3+4600 \beta ^2-4600 \beta +1680 & 480 \beta ^4-1680 \beta ^3+2400 \beta ^2-1680 \beta +480 \\
 \lbrack 4,2 \rbrack & 280 & 1648-1648 \beta  & 3912 \beta ^2-7464 \beta +3912 & -4672 \beta ^3+12760 \beta ^2-12760 \beta +4672 & 2800 \beta ^4-9744 \beta ^3+13920 \beta ^2-9744 \beta +2800 \\
 \lbrack 5,1 \rbrack & 210 & 1280-1280 \beta  & 3150 \beta ^2-5950 \beta +3150 & -3900 \beta ^3+10450 \beta ^2-10450 \beta +3900 & 2420 \beta ^4-8200 \beta ^3+11600 \beta ^2-8200 \beta +2420 \\
 \lbrack 6 \rbrack & 132 & 1024-1024 \beta  & 3360 \beta ^2-6300 \beta +3360 & -5952 \beta ^3+15712 \beta ^2-15712 \beta +5952 & 5980 \beta ^4-19812 \beta ^3+27832 \beta ^2-19812 \beta +5980 \\
 \lbrack 1,1,1,1,1,1,1 \rbrack & 720 & 720-720 \beta  & 0 & 0 & 0 \\
 \lbrack 2,1,1,1,1,1 \rbrack & 1440 & 2880-2880 \beta  & 1440 \beta ^2-2880 \beta +1440 & 0 & 0 \\
 \lbrack 2,2,1,1,1 \rbrack & 2160 & 6720-6720 \beta  & 6960 \beta ^2-13680 \beta +6960 & -2400 \beta ^3+6960 \beta ^2-6960 \beta +2400 & 0 \\
 \lbrack 2,2,2,1 \rbrack & 2592 & 11184-11184 \beta  & 18144 \beta ^2-35328 \beta +18144 & -13104 \beta ^3+37248 \beta ^2-37248 \beta +13104 & 3552 \beta ^4-13104 \beta ^3+19104 \beta ^2-13104 \beta +3552 \\
 \lbrack 3,1,1,1,1 \rbrack & 1800 & 5760-5760 \beta  & 6120 \beta ^2-11880 \beta +6120 & -2160 \beta ^3+6120 \beta ^2-6120 \beta +2160 & 0 \\
 \lbrack 3,2,1,1 \rbrack & 2160 & 9540-9540 \beta  & 15840 \beta ^2-30600 \beta +15840 & -11700 \beta ^3+32760 \beta ^2-32760 \beta +11700 & 3240 \beta ^4-11700 \beta ^3+16920 \beta ^2-11700 \beta +3240 \\
 \lbrack 3,2,2 \rbrack & 2160 & 12336-12336 \beta  & 28392 \beta ^2-54648 \beta +28392 & -32880 \beta ^3+91296 \beta ^2-91296 \beta +32880 & 19128 \beta ^4-68112 \beta ^3+98136 \beta ^2-68112 \beta +19128 \\
 \lbrack 3,3,1 \rbrack & 1800 & 10476-10476 \beta  & 24588 \beta ^2-47088 \beta +24588 & -29052 \beta ^3+79848 \beta ^2-79848 \beta +29052 & 17244 \beta ^4-60480 \beta ^3+86688 \beta ^2-60480 \beta +17244 \\
 \lbrack 4,1,1,1 \rbrack & 1680 & 7680-7680 \beta  & 13200 \beta ^2-25200 \beta +13200 & -10080 \beta ^3+27600 \beta ^2-27600 \beta +10080 & 2880 \beta ^4-10080 \beta ^3+14400 \beta ^2-10080 \beta +2880 \\
 \lbrack 4,2,1 \rbrack & 1680 & 9888-9888 \beta  & 23472 \beta ^2-44784 \beta +23472 & -28032 \beta ^3+76560 \beta ^2-76560 \beta +28032 & 16800 \beta ^4-58464 \beta ^3+83520 \beta ^2-58464 \beta +16800 \\
 \lbrack 4,3 \rbrack & 1200 & 8856-8856 \beta  & 27600 \beta ^2-52440 \beta +27600 & -46440 \beta ^3+125616 \beta ^2-125616 \beta +46440 & 44400 \beta ^4-152016 \beta ^3+216168 \beta ^2-152016 \beta +44400 \\
 \lbrack 5,1,1 \rbrack & 1260 & 7680-7680 \beta  & 18900 \beta ^2-35700 \beta +18900 & -23400 \beta ^3+62700 \beta ^2-62700 \beta +23400 & 14520 \beta ^4-49200 \beta ^3+69600 \beta ^2-49200 \beta +14520 \\
 \lbrack 5,2 \rbrack & 1080 & 8100-8100 \beta  & 25660 \beta ^2-48520 \beta +25660 & -43860 \beta ^3+117660 \beta ^2-117660 \beta +43860 & 42540 \beta ^4-144160 \beta ^3+204160 \beta ^2-144160 \beta +42540 \\
 \lbrack 6,1 \rbrack & 792 & 6144-6144 \beta  & 20160 \beta ^2-37800 \beta +20160 & -35712 \beta ^3+94272 \beta ^2-94272 \beta +35712 & 35880 \beta ^4-118872 \beta ^3+166992 \beta ^2-118872 \beta +35880 \\
 \lbrack 7 \rbrack & 429 & 4096-4096 \beta  & 17094 \beta ^2-31878 \beta +17094 & -40320 \beta ^3+105280 \beta ^2-105280 \beta +40320 & 57841 \beta ^4-188496 \beta ^3+263431 \beta ^2-188496 \beta +57841 \\
 \lbrack 1,1,1,1,1,1,1,1 \rbrack & 5040 & 5040-5040 \beta  & 0 & 0 & 0 \\
 \lbrack 2,1,1,1,1,1,1 \rbrack & 10080 & 20160-20160 \beta  & 10080 \beta ^2-20160 \beta +10080 & 0 & 0 \\
 \lbrack 2,2,1,1,1,1 \rbrack & 15120 & 47040-47040 \beta  & 48720 \beta ^2-95760 \beta +48720 & -16800 \beta ^3+48720 \beta ^2-48720 \beta +16800 & 0 \\
 \lbrack 2,2,2,1,1 \rbrack & 18144 & 78288-78288 \beta  & 127008 \beta ^2-247296 \beta +127008 & -91728 \beta ^3+260736 \beta ^2-260736 \beta +91728 & 24864 \beta ^4-91728 \beta ^3+133728 \beta ^2-91728 \beta +24864 \\
 \lbrack 2,2,2,2 \rbrack & 18144 & 101568-101568 \beta  & 229008 \beta ^2-443280 \beta +229008 & -259776 \beta ^3+729216 \beta ^2-729216 \beta +259776 & 148080 \beta ^4-535584 \beta ^3+776016 \beta ^2-535584 \beta +148080 \\
 \lbrack 3,1,1,1,1,1 \rbrack & 12600 & 40320-40320 \beta  & 42840 \beta ^2-83160 \beta +42840 & -15120 \beta ^3+42840 \beta ^2-42840 \beta +15120 & 0 \\
 \lbrack 3,2,1,1,1 \rbrack & 15120 & 66780-66780 \beta  & 110880 \beta ^2-214200 \beta +110880 & -81900 \beta ^3+229320 \beta ^2-229320 \beta +81900 & 22680 \beta ^4-81900 \beta ^3+118440 \beta ^2-81900 \beta +22680 \\
 \lbrack 3,2,2,1 \rbrack & 15120 & 86352-86352 \beta  & 198744 \beta ^2-382536 \beta +198744 & -230160 \beta ^3+639072 \beta ^2-639072 \beta +230160 & 133896 \beta ^4-476784 \beta ^3+686952 \beta ^2-476784 \beta +133896 \\
 \lbrack 3,3,1,1 \rbrack & 12600 & 73332-73332 \beta  & 172116 \beta ^2-329616 \beta +172116 & -203364 \beta ^3+558936 \beta ^2-558936 \beta +203364 & 120708 \beta ^4-423360 \beta ^3+606816 \beta ^2-423360 \beta +120708 \\
 \lbrack 3,3,2 \rbrack & 10800 & 77688-77688 \beta  & 235728 \beta ^2-450720 \beta +235728 & -385920 \beta ^3+1056744 \beta ^2-1056744 \beta +385920 & 358992 \beta ^4-1250856 \beta ^3+1789776 \beta ^2-1250856 \beta +358992 \\
 \lbrack 4,1,1,1,1 \rbrack & 11760 & 53760-53760 \beta  & 92400 \beta ^2-176400 \beta +92400 & -70560 \beta ^3+193200 \beta ^2-193200 \beta +70560 & 20160 \beta ^4-70560 \beta ^3+100800 \beta ^2-70560 \beta +20160 \\
 \lbrack 4,2,1,1 \rbrack & 11760 & 69216-69216 \beta  & 164304 \beta ^2-313488 \beta +164304 & -196224 \beta ^3+535920 \beta ^2-535920 \beta +196224 & 117600 \beta ^4-409248 \beta ^3+584640 \beta ^2-409248 \beta +117600 \\
 \lbrack 4,2,2 \rbrack & 10080 & 73248-73248 \beta  & 224560 \beta ^2-428080 \beta +224560 & -371328 \beta ^3+1011504 \beta ^2-1011504 \beta +371328 & 348624 \beta ^4-1206784 \beta ^3+1722256 \beta ^2-1206784 \beta +348624 \\
 \lbrack 4,3,1 \rbrack & 8400 & 61992-61992 \beta  & 193200 \beta ^2-367080 \beta +193200 & -325080 \beta ^3+879312 \beta ^2-879312 \beta +325080 & 310800 \beta ^4-1064112 \beta ^3+1513176 \beta ^2-1064112 \beta +310800 \\
 \lbrack 4,4 \rbrack & 4900 & 44224-44224 \beta  & 174152 \beta ^2-329304 \beta +174152 & -387616 \beta ^3+1038000 \beta ^2-1038000 \beta +387616 & 525700 \beta ^4-1771920 \beta ^3+2507180 \beta ^2-1771920 \beta +525700 \\
 \lbrack 5,1,1,1 \rbrack & 8820 & 53760-53760 \beta  & 132300 \beta ^2-249900 \beta +132300 & -163800 \beta ^3+438900 \beta ^2-438900 \beta +163800 & 101640 \beta ^4-344400 \beta ^3+487200 \beta ^2-344400 \beta +101640 \\
 \lbrack 5,2,1 \rbrack & 7560 & 56700-56700 \beta  & 179620 \beta ^2-339640 \beta +179620 & -307020 \beta ^3+823620 \beta ^2-823620 \beta +307020 & 297780 \beta ^4-1009120 \beta ^3+1429120 \beta ^2-1009120 \beta +297780 \\
 \lbrack 5,3 \rbrack & 4725 & 42900-42900 \beta  & 169950 \beta ^2-320790 \beta +169950 & -380400 \beta ^3+1015680 \beta ^2-1015680 \beta +380400 & 518505 \beta ^4-1741560 \beta ^3+2460735 \beta ^2-1741560 \beta +518505 \\
 \lbrack 6,1,1 \rbrack & 5544 & 43008-43008 \beta  & 141120 \beta ^2-264600 \beta +141120 & -249984 \beta ^3+659904 \beta ^2-659904 \beta +249984 & 251160 \beta ^4-832104 \beta ^3+1168944 \beta ^2-832104 \beta +251160 \\
 \lbrack 6,2 \rbrack & 4158 & 38448-38448 \beta  & 155172 \beta ^2-291444 \beta +155172 & -353664 \beta ^3+936144 \beta ^2-936144 \beta +353664 & 490230 \beta ^4-1628688 \beta ^3+2291658 \beta ^2-1628688 \beta +490230 \\
 \lbrack 7,1 \rbrack & 3003 & 28672-28672 \beta  & 119658 \beta ^2-223146 \beta +119658 & -282240 \beta ^3+736960 \beta ^2-736960 \beta +282240 & 404887 \beta ^4-1319472 \beta ^3+1844017 \beta ^2-1319472 \beta +404887 \\
 \lbrack 8 \rbrack & 1430 & 16384-16384 \beta  & 84084 \beta ^2-156156 \beta +84084 & -251904 \beta ^3+652288 \beta ^2-652288 \beta +251904 & 480238 \beta ^4-1545456 \beta ^3+2151226 \beta ^2-1545456 \beta +480238 \\
 \lbrack 1,1,1,1,1,1,1,1,1 \rbrack & 40320 & 40320-40320 \beta  & 0 & 0 & 0 \\
 \lbrack 2,1,1,1,1,1,1,1 \rbrack & 80640 & 161280-161280 \beta  & 80640 \beta ^2-161280 \beta +80640 & 0 & 0 \\
 \lbrack 2,2,1,1,1,1,1 \rbrack & 120960 & 376320-376320 \beta  & 389760 \beta ^2-766080 \beta +389760 & -134400 \beta ^3+389760 \beta ^2-389760 \beta +134400 & 0 \\
 \lbrack 2,2,2,1,1,1 \rbrack & 145152 & 626304-626304 \beta  & 1016064 \beta ^2-1978368 \beta +1016064 & -733824 \beta ^3+2085888 \beta ^2-2085888 \beta +733824 & 198912 \beta ^4-733824 \beta ^3+1069824 \beta ^2-733824 \beta +198912 \\
 \lbrack 2,2,2,2,1 \rbrack & 145152 & 812544-812544 \beta  & 1832064 \beta ^2-3546240 \beta +1832064 & -2078208 \beta ^3+5833728 \beta ^2-5833728 \beta +2078208 & 1184640 \beta ^4-4284672 \beta ^3+6208128 \beta ^2-4284672 \beta +1184640 \\
 \lbrack 3,1,1,1,1,1,1 \rbrack & 100800 & 322560-322560 \beta  & 342720 \beta ^2-665280 \beta +342720 & -120960 \beta ^3+342720 \beta ^2-342720 \beta +120960 & 0 \\
 \lbrack 3,2,1,1,1,1 \rbrack & 120960 & 534240-534240 \beta  & 887040 \beta ^2-1713600 \beta +887040 & -655200 \beta ^3+1834560 \beta ^2-1834560 \beta +655200 & 181440 \beta ^4-655200 \beta ^3+947520 \beta ^2-655200 \beta +181440 \\
 \lbrack 3,2,2,1,1 \rbrack & 120960 & 690816-690816 \beta  & 1589952 \beta ^2-3060288 \beta +1589952 & -1841280 \beta ^3+5112576 \beta ^2-5112576 \beta +1841280 & 1071168 \beta ^4-3814272 \beta ^3+5495616 \beta ^2-3814272 \beta +1071168 \\
 \lbrack 3,2,2,2 \rbrack & 103680 & 733440-733440 \beta  & 2186784 \beta ^2-4197600 \beta +2186784 & -3515472 \beta ^3+9702192 \beta ^2-9702192 \beta +3515472 & 3209952 \beta ^4-11315808 \beta ^3+16256064 \beta ^2-11315808 \beta +3209952 \\
 \lbrack 3,3,1,1,1 \rbrack & 100800 & 586656-586656 \beta  & 1376928 \beta ^2-2636928 \beta +1376928 & -1626912 \beta ^3+4471488 \beta ^2-4471488 \beta +1626912 & 965664 \beta ^4-3386880 \beta ^3+4854528 \beta ^2-3386880 \beta +965664 \\
 \lbrack 3,3,2,1 \rbrack & 86400 & 621504-621504 \beta  & 1885824 \beta ^2-3605760 \beta +1885824 & -3087360 \beta ^3+8453952 \beta ^2-8453952 \beta +3087360 & 2871936 \beta ^4-10006848 \beta ^3+14318208 \beta ^2-10006848 \beta +2871936 \\
 \lbrack 3,3,3 \rbrack & 54000 & 472608-472608 \beta  & 1802016 \beta ^2-3430944 \beta +1802016 & -3879576 \beta ^3+10528488 \beta ^2-10528488 \beta +3879576 & 5088024 \beta ^4-17479152 \beta ^3+24897168 \beta ^2-17479152 \beta +5088024 \\
 \lbrack 4,1,1,1,1,1 \rbrack & 94080 & 430080-430080 \beta  & 739200 \beta ^2-1411200 \beta +739200 & -564480 \beta ^3+1545600 \beta ^2-1545600 \beta +564480 & 161280 \beta ^4-564480 \beta ^3+806400 \beta ^2-564480 \beta +161280 \\
 \lbrack 4,2,1,1,1 \rbrack & 94080 & 553728-553728 \beta  & 1314432 \beta ^2-2507904 \beta +1314432 & -1569792 \beta ^3+4287360 \beta ^2-4287360 \beta +1569792 & 940800 \beta ^4-3273984 \beta ^3+4677120 \beta ^2-3273984 \beta +940800 \\
 \lbrack 4,2,2,1 \rbrack & 80640 & 585984-585984 \beta  & 1796480 \beta ^2-3424640 \beta +1796480 & -2970624 \beta ^3+8092032 \beta ^2-8092032 \beta +2970624 & 2788992 \beta ^4-9654272 \beta ^3+13778048 \beta ^2-9654272 \beta +2788992 \\
 \lbrack 4,3,1,1 \rbrack & 67200 & 495936-495936 \beta  & 1545600 \beta ^2-2936640 \beta +1545600 & -2600640 \beta ^3+7034496 \beta ^2-7034496 \beta +2600640 & 2486400 \beta ^4-8512896 \beta ^3+12105408 \beta ^2-8512896 \beta +2486400 \\
 \lbrack 4,3,2 \rbrack & 50400 & 445056-445056 \beta  & 1712640 \beta ^2-3253200 \beta +1712640 & -3721008 \beta ^3+10055904 \beta ^2-10055904 \beta +3721008 & 4922880 \beta ^4-16819008 \beta ^3+23908080 \beta ^2-16819008 \beta +4922880 \\
 \lbrack 4,4,1 \rbrack & 39200 & 353792-353792 \beta  & 1393216 \beta ^2-2634432 \beta +1393216 & -3100928 \beta ^3+8304000 \beta ^2-8304000 \beta +3100928 & 4205600 \beta ^4-14175360 \beta ^3+20057440 \beta ^2-14175360 \beta +4205600 \\
 \lbrack 5,1,1,1,1 \rbrack & 70560 & 430080-430080 \beta  & 1058400 \beta ^2-1999200 \beta +1058400 & -1310400 \beta ^3+3511200 \beta ^2-3511200 \beta +1310400 & 813120 \beta ^4-2755200 \beta ^3+3897600 \beta ^2-2755200 \beta +813120 \\
 \lbrack 5,2,1,1 \rbrack & 60480 & 453600-453600 \beta  & 1436960 \beta ^2-2717120 \beta +1436960 & -2456160 \beta ^3+6588960 \beta ^2-6588960 \beta +2456160 & 2382240 \beta ^4-8072960 \beta ^3+11432960 \beta ^2-8072960 \beta +2382240 \\
 \lbrack 5,2,2 \rbrack & 45360 & 406560-406560 \beta  & 1588400 \beta ^2-3004880 \beta +1588400 & -3502400 \beta ^3+9397840 \beta ^2-9397840 \beta +3502400 & 4697760 \beta ^4-15906080 \beta ^3+22532480 \beta ^2-15906080 \beta +4697760 \\
 \lbrack 5,3,1 \rbrack & 37800 & 343200-343200 \beta  & 1359600 \beta ^2-2566320 \beta +1359600 & -3043200 \beta ^3+8125440 \beta ^2-8125440 \beta +3043200 & 4148040 \beta ^4-13932480 \beta ^3+19685880 \beta ^2-13932480 \beta +4148040 \\
 \lbrack 5,4 \rbrack & 19600 & 211780-211780 \beta  & 1023120 \beta ^2-1925840 \beta +1023120 & -2885080 \beta ^3+7656000 \beta ^2-7656000 \beta +2885080 & 5186080 \beta ^4-17243440 \beta ^3+24289120 \beta ^2-17243440 \beta +5186080 \\
 \lbrack 6,1,1,1 \rbrack & 44352 & 344064-344064 \beta  & 1128960 \beta ^2-2116800 \beta +1128960 & -1999872 \beta ^3+5279232 \beta ^2-5279232 \beta +1999872 & 2009280 \beta ^4-6656832 \beta ^3+9351552 \beta ^2-6656832 \beta +2009280 \\
 \lbrack 6,2,1 \rbrack & 33264 & 307584-307584 \beta  & 1241376 \beta ^2-2331552 \beta +1241376 & -2829312 \beta ^3+7489152 \beta ^2-7489152 \beta +2829312 & 3921840 \beta ^4-13029504 \beta ^3+18333264 \beta ^2-13029504 \beta +3921840 \\
 \lbrack 6,3 \rbrack & 18480 & 201546-201546 \beta  & 982800 \beta ^2-1845144 \beta +982800 & -2796012 \beta ^3+7387320 \beta ^2-7387320 \beta +2796012 & 5066064 \beta ^4-16756128 \beta ^3+23554008 \beta ^2-16756128 \beta +5066064 \\
 \lbrack 7,1,1 \rbrack & 24024 & 229376-229376 \beta  & 957264 \beta ^2-1785168 \beta +957264 & -2257920 \beta ^3+5895680 \beta ^2-5895680 \beta +2257920 & 3239096 \beta ^4-10555776 \beta ^3+14752136 \beta ^2-10555776 \beta +3239096 \\
 \lbrack 7,2 \rbrack & 16016 & 178038-178038 \beta  & 885248 \beta ^2-1654072 \beta +885248 & -2567124 \beta ^3+6725544 \beta ^2-6725544 \beta +2567124 & 4735920 \beta ^4-15495536 \beta ^3+21694456 \beta ^2-15495536 \beta +4735920 \\
 \lbrack 8,1 \rbrack & 11440 & 131072-131072 \beta  & 672672 \beta ^2-1249248 \beta +672672 & -2015232 \beta ^3+5218304 \beta ^2-5218304 \beta +2015232 & 3841904 \beta ^4-12363648 \beta ^3+17209808 \beta ^2-12363648 \beta +3841904 \\
 \lbrack 9 \rbrack & 4862 & 65536-65536 \beta  & 403260 \beta ^2-746460 \beta +403260 & -1483776 \beta ^3+3816960 \beta ^2-3816960 \beta +1483776 & 3586726 \beta ^4-11428560 \beta ^3+15857842 \beta ^2-11428560 \beta +3586726 \\
 \lbrack 1,1,1,1,1,1,1,1,1,1 \rbrack & 362880 & 362880-362880 \beta  & 0 & 0 & 0 \\
 \lbrack 2,1,1,1,1,1,1,1,1 \rbrack & 725760 & 1451520-1451520 \beta  & 725760 \beta ^2-1451520 \beta +725760 & 0 & 0 \\
 \lbrack 2,2,1,1,1,1,1,1 \rbrack & 1088640 & 3386880-3386880 \beta  & 3507840 \beta ^2-6894720 \beta +3507840 & -1209600 \beta ^3+3507840 \beta ^2-3507840 \beta +1209600 & 0 \\
 \lbrack 2,2,2,1,1,1,1 \rbrack & 1306368 & 5636736-5636736 \beta  & 9144576 \beta ^2-17805312 \beta +9144576 & -6604416 \beta ^3+18772992 \beta ^2-18772992 \beta +6604416 & 1790208 \beta ^4-6604416 \beta ^3+9628416 \beta ^2-6604416 \beta +1790208 \\
 \lbrack 2,2,2,2,1,1 \rbrack & 1306368 & 7312896-7312896 \beta  & 16488576 \beta ^2-31916160 \beta +16488576 & -18703872 \beta ^3+52503552 \beta ^2-52503552 \beta +18703872 & 10661760 \beta ^4-38562048 \beta ^3+55873152 \beta ^2-38562048 \beta +10661760 \\
 \lbrack 2,2,2,2,2 \rbrack & 1119744 & 7782912-7782912 \beta  & 22783488 \beta ^2-43923456 \beta +22783488 & -35944320 \beta ^3+100041600 \beta ^2-100041600 \beta +35944320 & 32205312 \beta ^4-114919680 \beta ^3+165799680 \beta ^2-114919680 \beta +32205312 \\
 \lbrack 3,1,1,1,1,1,1,1 \rbrack & 907200 & 2903040-2903040 \beta  & 3084480 \beta ^2-5987520 \beta +3084480 & -1088640 \beta ^3+3084480 \beta ^2-3084480 \beta +1088640 & 0 \\
 \lbrack 3,2,1,1,1,1,1 \rbrack & 1088640 & 4808160-4808160 \beta  & 7983360 \beta ^2-15422400 \beta +7983360 & -5896800 \beta ^3+16511040 \beta ^2-16511040 \beta +5896800 & 1632960 \beta ^4-5896800 \beta ^3+8527680 \beta ^2-5896800 \beta +1632960 \\
 \lbrack 3,2,2,1,1,1 \rbrack & 1088640 & 6217344-6217344 \beta  & 14309568 \beta ^2-27542592 \beta +14309568 & -16571520 \beta ^3+46013184 \beta ^2-46013184 \beta +16571520 & 9640512 \beta ^4-34328448 \beta ^3+49460544 \beta ^2-34328448 \beta +9640512 \\
 \lbrack 3,2,2,2,1 \rbrack & 933120 & 6600960-6600960 \beta  & 19681056 \beta ^2-37778400 \beta +19681056 & -31639248 \beta ^3+87319728 \beta ^2-87319728 \beta +31639248 & 28889568 \beta ^4-101842272 \beta ^3+146304576 \beta ^2-101842272 \beta +28889568 \\
 \lbrack 3,3,1,1,1,1 \rbrack & 907200 & 5279904-5279904 \beta  & 12392352 \beta ^2-23732352 \beta +12392352 & -14642208 \beta ^3+40243392 \beta ^2-40243392 \beta +14642208 & 8690976 \beta ^4-30481920 \beta ^3+43690752 \beta ^2-30481920 \beta +8690976 \\
 \lbrack 3,3,2,1,1 \rbrack & 777600 & 5593536-5593536 \beta  & 16972416 \beta ^2-32451840 \beta +16972416 & -27786240 \beta ^3+76085568 \beta ^2-76085568 \beta +27786240 & 25847424 \beta ^4-90061632 \beta ^3+128863872 \beta ^2-90061632 \beta +25847424 \\
 \lbrack 3,3,2,2 \rbrack & 583200 & 5031360-5031360 \beta  & 18894528 \beta ^2-36078624 \beta +18894528 & -40031568 \beta ^3+109283472 \beta ^2-109283472 \beta +40031568 & 51631920 \beta ^4-178968672 \beta ^3+255695040 \beta ^2-178968672 \beta +51631920 \\
 \lbrack 3,3,3,1 \rbrack & 486000 & 4253472-4253472 \beta  & 16218144 \beta ^2-30878496 \beta +16218144 & -34916184 \beta ^3+94756392 \beta ^2-94756392 \beta +34916184 & 45792216 \beta ^4-157312368 \beta ^3+224074512 \beta ^2-157312368 \beta +45792216 \\
 \lbrack 4,1,1,1,1,1,1 \rbrack & 846720 & 3870720-3870720 \beta  & 6652800 \beta ^2-12700800 \beta +6652800 & -5080320 \beta ^3+13910400 \beta ^2-13910400 \beta +5080320 & 1451520 \beta ^4-5080320 \beta ^3+7257600 \beta ^2-5080320 \beta +1451520 \\
 \lbrack 4,2,1,1,1,1 \rbrack & 846720 & 4983552-4983552 \beta  & 11829888 \beta ^2-22571136 \beta +11829888 & -14128128 \beta ^3+38586240 \beta ^2-38586240 \beta +14128128 & 8467200 \beta ^4-29465856 \beta ^3+42094080 \beta ^2-29465856 \beta +8467200 \\
 \lbrack 4,2,2,1,1 \rbrack & 725760 & 5273856-5273856 \beta  & 16168320 \beta ^2-30821760 \beta +16168320 & -26735616 \beta ^3+72828288 \beta ^2-72828288 \beta +26735616 & 25100928 \beta ^4-86888448 \beta ^3+124002432 \beta ^2-86888448 \beta +25100928 \\
 \lbrack 4,2,2,2 \rbrack & 544320 & 4739712-4739712 \beta  & 17969088 \beta ^2-34225536 \beta +17969088 & -38427840 \beta ^3+104441280 \beta ^2-104441280 \beta +38427840 & 50001792 \beta ^4-172334208 \beta ^3+245686464 \beta ^2-172334208 \beta +50001792 \\
 \lbrack 4,3,1,1,1 \rbrack & 604800 & 4463424-4463424 \beta  & 13910400 \beta ^2-26429760 \beta +13910400 & -23405760 \beta ^3+63310464 \beta ^2-63310464 \beta +23405760 & 22377600 \beta ^4-76616064 \beta ^3+108948672 \beta ^2-76616064 \beta +22377600 \\
 \lbrack 4,3,2,1 \rbrack & 453600 & 4005504-4005504 \beta  & 15413760 \beta ^2-29278800 \beta +15413760 & -33489072 \beta ^3+90503136 \beta ^2-90503136 \beta +33489072 & 44305920 \beta ^4-151371072 \beta ^3+215172720 \beta ^2-151371072 \beta +44305920 \\
 \lbrack 4,3,3 \rbrack & 252000 & 2636676-2636676 \beta  & 12314880 \beta ^2-23335200 \beta +12314880 & -33536232 \beta ^3+90142848 \beta ^2-90142848 \beta +33536232 & 58192272 \beta ^4-197053776 \beta ^3+279341280 \beta ^2-197053776 \beta +58192272 \\
 \lbrack 4,4,1,1 \rbrack & 352800 & 3184128-3184128 \beta  & 12538944 \beta ^2-23709888 \beta +12538944 & -27908352 \beta ^3+74736000 \beta ^2-74736000 \beta +27908352 & 37850400 \beta ^4-127578240 \beta ^3+180516960 \beta ^2-127578240 \beta +37850400 \\
 \lbrack 4,4,2 \rbrack & 235200 & 2480224-2480224 \beta  & 11678976 \beta ^2-22090752 \beta +11678976 & -32068928 \beta ^3+85911296 \beta ^2-85911296 \beta +32068928 & 56101696 \beta ^4-189108096 \beta ^3+267643840 \beta ^2-189108096 \beta +56101696 \\
 \lbrack 5,1,1,1,1,1 \rbrack & 635040 & 3870720-3870720 \beta  & 9525600 \beta ^2-17992800 \beta +9525600 & -11793600 \beta ^3+31600800 \beta ^2-31600800 \beta +11793600 & 7318080 \beta ^4-24796800 \beta ^3+35078400 \beta ^2-24796800 \beta +7318080 \\
 \lbrack 5,2,1,1,1 \rbrack & 544320 & 4082400-4082400 \beta  & 12932640 \beta ^2-24454080 \beta +12932640 & -22105440 \beta ^3+59300640 \beta ^2-59300640 \beta +22105440 & 21440160 \beta ^4-72656640 \beta ^3+102896640 \beta ^2-72656640 \beta +21440160 \\
 \lbrack 5,2,2,1 \rbrack & 408240 & 3659040-3659040 \beta  & 14295600 \beta ^2-27043920 \beta +14295600 & -31521600 \beta ^3+84580560 \beta ^2-84580560 \beta +31521600 & 42279840 \beta ^4-143154720 \beta ^3+202792320 \beta ^2-143154720 \beta +42279840 \\
 \lbrack 5,3,1,1 \rbrack & 340200 & 3088800-3088800 \beta  & 12236400 \beta ^2-23096880 \beta +12236400 & -27388800 \beta ^3+73128960 \beta ^2-73128960 \beta +27388800 & 37332360 \beta ^4-125392320 \beta ^3+177172920 \beta ^2-125392320 \beta +37332360 \\
 \lbrack 5,3,2 \rbrack & 226800 & 2405070-2405070 \beta  & 11388960 \beta ^2-21508680 \beta +11388960 & -31440900 \beta ^3+84007320 \beta ^2-84007320 \beta +31440900 & 55271040 \beta ^4-185713200 \beta ^3+262510920 \beta ^2-185713200 \beta +55271040 \\
 \lbrack 5,4,1 \rbrack & 176400 & 1906020-1906020 \beta  & 9208080 \beta ^2-17332560 \beta +9208080 & -25965720 \beta ^3+68904000 \beta ^2-68904000 \beta +25965720 & 46674720 \beta ^4-155190960 \beta ^3+218602080 \beta ^2-155190960 \beta +46674720 \\
 \lbrack 5,5 \rbrack & 79380 & 1005050-1005050 \beta  & 5798500 \beta ^2-10881600 \beta +5798500 & -20004300 \beta ^3+52753600 \beta ^2-52753600 \beta +20004300 & 45424140 \beta ^4-149581200 \beta ^3+210045500 \beta ^2-149581200 \beta +45424140 \\
 \lbrack 6,1,1,1,1 \rbrack & 399168 & 3096576-3096576 \beta  & 10160640 \beta ^2-19051200 \beta +10160640 & -17998848 \beta ^3+47513088 \beta ^2-47513088 \beta +17998848 & 18083520 \beta ^4-59911488 \beta ^3+84163968 \beta ^2-59911488 \beta +18083520 \\
 \lbrack 6,2,1,1 \rbrack & 299376 & 2768256-2768256 \beta  & 11172384 \beta ^2-20983968 \beta +11172384 & -25463808 \beta ^3+67402368 \beta ^2-67402368 \beta +25463808 & 35296560 \beta ^4-117265536 \beta ^3+164999376 \beta ^2-117265536 \beta +35296560 \\
 \lbrack 6,2,2 \rbrack & 199584 & 2153016-2153016 \beta  & 10375008 \beta ^2-19508160 \beta +10375008 & -29136144 \beta ^3+77252064 \beta ^2-77252064 \beta +29136144 & 52049952 \beta ^4-173181408 \beta ^3+243897888 \beta ^2-173181408 \beta +52049952 \\
 \lbrack 6,3,1 \rbrack & 166320 & 1813914-1813914 \beta  & 8845200 \beta ^2-16606296 \beta +8845200 & -25164108 \beta ^3+66485880 \beta ^2-66485880 \beta +25164108 & 45594576 \beta ^4-150805152 \beta ^3+211986072 \beta ^2-150805152 \beta +45594576 \\
 \lbrack 6,4 \rbrack & 77616 & 986352-986352 \beta  & 5711544 \beta ^2-10707408 \beta +5711544 & -19772256 \beta ^3+52053024 \beta ^2-52053024 \beta +19772256 & 45034176 \beta ^4-147999024 \beta ^3+207658416 \beta ^2-147999024 \beta +45034176 \\
 \lbrack 7,1,1,1 \rbrack & 216216 & 2064384-2064384 \beta  & 8615376 \beta ^2-16066512 \beta +8615376 & -20321280 \beta ^3+53061120 \beta ^2-53061120 \beta +20321280 & 29151864 \beta ^4-95001984 \beta ^3+132769224 \beta ^2-95001984 \beta +29151864 \\
 \lbrack 7,2,1 \rbrack & 144144 & 1602342-1602342 \beta  & 7967232 \beta ^2-14886648 \beta +7967232 & -23104116 \beta ^3+60529896 \beta ^2-60529896 \beta +23104116 & 42623280 \beta ^4-139459824 \beta ^3+195250104 \beta ^2-139459824 \beta +42623280 \\
 \lbrack 7,3 \rbrack & 72072 & 926226-926226 \beta  & 5424090 \beta ^2-10138548 \beta +5424090 & -18980220 \beta ^3+49717920 \beta ^2-49717920 \beta +18980220 & 43654548 \beta ^4-142589076 \beta ^3+199597776 \beta ^2-142589076 \beta +43654548 \\
 \lbrack 8,1,1 \rbrack & 102960 & 1179648-1179648 \beta  & 6054048 \beta ^2-11243232 \beta +6054048 & -18137088 \beta ^3+46964736 \beta ^2-46964736 \beta +18137088 & 34577136 \beta ^4-111272832 \beta ^3+154888272 \beta ^2-111272832 \beta +34577136 \\
 \lbrack 8,2 \rbrack & 61776 & 809312-809312 \beta  & 4833568 \beta ^2-8994464 \beta +4833568 & -17246016 \beta ^3+44816128 \beta ^2-44816128 \beta +17246016 & 40408080 \beta ^4-130622848 \beta ^3+182155568 \beta ^2-130622848 \beta +40408080 \\
 \lbrack 9,1 \rbrack & 43758 & 589824-589824 \beta  & 3629340 \beta ^2-6718140 \beta +3629340 & -13353984 \beta ^3+34352640 \beta ^2-34352640 \beta +13353984 & 32280534 \beta ^4-102857040 \beta ^3+142720578 \beta ^2-102857040 \beta +32280534 \\
 \lbrack 10 \rbrack & 16796 & 262144-262144 \beta  & 1896180 \beta ^2-3500640 \beta +1896180 & -8355840 \beta ^3+21381120 \beta ^2-21381120 \beta +8355840 & 24771604 \beta ^4-78301080 \beta ^3+108368260 \beta ^2-78301080 \beta +24771604 \\
\hline
\end{longtblr}
\end{adjustwidth}
}
{\tiny
\begin{adjustwidth}{-1.4in}{-1.4in}
\begin{longtblr}[
theme = fancy,
caption = {Polynomials $n_{\nu,\lambda}(\phi)$ for the \WLBE{}},
label = {tblr:SchurEpsilonWLphi},
]{
colspec = {X[-1,c]X[-1,c]X[-1,c]X[-1,c]X[-1,c]X[-1,c]X[-1,c]X[-1,c]X[-1,c]X[-1,c]},
cells = {mode=dmath,font=\tiny},
rowsep = {1pt},
colsep = {2pt},
rowhead = 1,
hspan=minimal,
}
\hline
 \diagbox{\nu}{\lambda} & {[]} & {[1]} & {[1,1]} & {[2]} & {[2,1]} & {[3]} & {[2,2]} & {[3,1]} & {[4]} \\
\hline
 \lbrack 1 \rbrack & 1 & \phi  & 0 & 0 & 0 & 0 & 0 & 0 & 0 \\
 \lbrack 1,1 \rbrack & 1 & \phi  & 0 & 0 & 0 & 0 & 0 & 0 & 0 \\
 \lbrack 2 \rbrack & 2 & 3 \phi +1 & \phi ^2+\phi  & \phi ^2+\phi  & 0 & 0 & 0 & 0 & 0 \\
 \lbrack 1,1,1 \rbrack & 2 & 2 \phi  & 0 & 0 & 0 & 0 & 0 & 0 & 0 \\
 \lbrack 2,1 \rbrack & 4 & 6 \phi +2 & 2 \phi ^2+2 \phi  & 2 \phi ^2+2 \phi  & 0 & 0 & 0 & 0 & 0 \\
 \lbrack 3 \rbrack & 5 & 10 \phi +6 & 6 \phi ^2+9 \phi +1 & 6 \phi ^2+9 \phi +2 & 2 \phi ^3+6 \phi ^2+3 \phi  & \phi ^3+3 \phi ^2+2 \phi  & 0 & 0 & 0 \\
 \lbrack 1,1,1,1 \rbrack & 6 & 6 \phi  & 0 & 0 & 0 & 0 & 0 & 0 & 0 \\
 \lbrack 2,1,1 \rbrack & 12 & 18 \phi +6 & 6 \phi ^2+6 \phi  & 6 \phi ^2+6 \phi  & 0 & 0 & 0 & 0 & 0 \\
 \lbrack 2,2 \rbrack & 18 & 36 \phi +20 & 22 \phi ^2+30 \phi +4 & 22 \phi ^2+30 \phi +6 & 8 \phi ^3+20 \phi ^2+10 \phi  & 4 \phi ^3+10 \phi ^2+6 \phi  & 0 & 0 & 0 \\
 \lbrack 3,1 \rbrack & 15 & 30 \phi +18 & 18 \phi ^2+27 \phi +3 & 18 \phi ^2+27 \phi +6 & 6 \phi ^3+18 \phi ^2+9 \phi  & 3 \phi ^3+9 \phi ^2+6 \phi  & 0 & 0 & 0 \\
 \lbrack 4 \rbrack & 14 & 35 \phi +29 & 30 \phi ^2+58 \phi +12 & 30 \phi ^2+58 \phi +22 & 20 \phi ^3+70 \phi ^2+51 \phi +5 & 10 \phi ^3+35 \phi ^2+33 \phi +6 & 2 \phi ^4+12 \phi ^3+17 \phi ^2+5 \phi  & 3 \phi ^4+18 \phi ^3+28 \phi ^2+11 \phi  & \phi ^4+6 \phi ^3+11 \phi ^2+6 \phi  \\
 \lbrack 1,1,1,1,1 \rbrack & 24 & 24 \phi  & 0 & 0 & 0 & 0 & 0 & 0 & 0 \\
 \lbrack 2,1,1,1 \rbrack & 48 & 72 \phi +24 & 24 \phi ^2+24 \phi  & 24 \phi ^2+24 \phi  & 0 & 0 & 0 & 0 & 0 \\
 \lbrack 2,2,1 \rbrack & 72 & 144 \phi +80 & 88 \phi ^2+120 \phi +16 & 88 \phi ^2+120 \phi +24 & 32 \phi ^3+80 \phi ^2+40 \phi  & 16 \phi ^3+40 \phi ^2+24 \phi  & 0 & 0 & 0 \\
 \lbrack 3,1,1 \rbrack & 60 & 120 \phi +72 & 72 \phi ^2+108 \phi +12 & 72 \phi ^2+108 \phi +24 & 24 \phi ^3+72 \phi ^2+36 \phi  & 12 \phi ^3+36 \phi ^2+24 \phi  & 0 & 0 & 0 \\
 \lbrack 3,2 \rbrack & 72 & 180 \phi +138 & 156 \phi ^2+276 \phi +60 & 156 \phi ^2+276 \phi +96 & 108 \phi ^3+336 \phi ^2+234 \phi +24 & 54 \phi ^3+168 \phi ^2+144 \phi +24 & 12 \phi ^4+60 \phi ^3+78 \phi ^2+24 \phi  & 18 \phi ^4+90 \phi ^3+126 \phi ^2+48 \phi  & 6 \phi ^4+30 \phi ^3+48 \phi ^2+24 \phi  \\
 \lbrack 4,1 \rbrack & 56 & 140 \phi +116 & 120 \phi ^2+232 \phi +48 & 120 \phi ^2+232 \phi +88 & 80 \phi ^3+280 \phi ^2+204 \phi +20 & 40 \phi ^3+140 \phi ^2+132 \phi +24 & 8 \phi ^4+48 \phi ^3+68 \phi ^2+20 \phi  & 12 \phi ^4+72 \phi ^3+112 \phi ^2+44 \phi  & 4 \phi ^4+24 \phi ^3+44 \phi ^2+24 \phi  \\
 \lbrack 5 \rbrack & 42 & 126 \phi +130 & 140 \phi ^2+325 \phi +95 & 140 \phi ^2+325 \phi +165 & 140 \phi ^3+560 \phi ^2+520 \phi +90 & 70 \phi ^3+280 \phi ^2+330 \phi +100 & 30 \phi ^4+190 \phi ^3+315 \phi ^2+135 \phi +10 & 45 \phi ^4+285 \phi ^3+515 \phi ^2+285 \phi +26 & 15 \phi ^4+95 \phi ^3+200 \phi ^2+150 \phi +24 \\
\hline
\end{longtblr}
\end{adjustwidth}
}
{\tiny
\begin{longtblr}[
theme = fancy,
caption = {Polynomials $n_{\nu,\lambda}(\phi)$ for the \WLBE{} evaluated at $\phi=1$},
label = {tblr:SchurEpsilonWL},
]{
colspec = {X[-1,c]X[-1,c]X[-1,c]X[-1,c]X[-1,c]X[-1,c]X[-1,c]X[-1,c]X[-1,c]X[-1,c]},
cells = {mode=dmath,font=\tiny},
rowsep = {1pt},
colsep = {2pt},
rowhead = 1,
hspan=minimal,
}
\hline
 \diagbox{\nu}{\lambda} & {[]} & {[1]} & {[1,1]} & {[2]} & {[2,1]} & {[3]} & {[2,2]} & {[3,1]} & {[4]} \\
\hline
 \lbrack 1 \rbrack & 1 & 1 & 0 & 0 & 0 & 0 & 0 & 0 & 0 \\
 \lbrack 1,1 \rbrack & 1 & 1 & 0 & 0 & 0 & 0 & 0 & 0 & 0 \\
 \lbrack 2 \rbrack & 2 & 4 & 2 & 2 & 0 & 0 & 0 & 0 & 0 \\
 \lbrack 1,1,1 \rbrack & 2 & 2 & 0 & 0 & 0 & 0 & 0 & 0 & 0 \\
 \lbrack 2,1 \rbrack & 4 & 8 & 4 & 4 & 0 & 0 & 0 & 0 & 0 \\
 \lbrack 3 \rbrack & 5 & 16 & 16 & 17 & 11 & 6 & 0 & 0 & 0 \\
 \lbrack 1,1,1,1 \rbrack & 6 & 6 & 0 & 0 & 0 & 0 & 0 & 0 & 0 \\
 \lbrack 2,1,1 \rbrack & 12 & 24 & 12 & 12 & 0 & 0 & 0 & 0 & 0 \\
 \lbrack 2,2 \rbrack & 18 & 56 & 56 & 58 & 38 & 20 & 0 & 0 & 0 \\
 \lbrack 3,1 \rbrack & 15 & 48 & 48 & 51 & 33 & 18 & 0 & 0 & 0 \\
 \lbrack 4 \rbrack & 14 & 64 & 100 & 110 & 146 & 84 & 36 & 60 & 24 \\
 \lbrack 1,1,1,1,1 \rbrack & 24 & 24 & 0 & 0 & 0 & 0 & 0 & 0 & 0 \\
 \lbrack 2,1,1,1 \rbrack & 48 & 96 & 48 & 48 & 0 & 0 & 0 & 0 & 0 \\
 \lbrack 2,2,1 \rbrack & 72 & 224 & 224 & 232 & 152 & 80 & 0 & 0 & 0 \\
 \lbrack 3,1,1 \rbrack & 60 & 192 & 192 & 204 & 132 & 72 & 0 & 0 & 0 \\
 \lbrack 3,2 \rbrack & 72 & 318 & 492 & 528 & 702 & 390 & 174 & 282 & 108 \\
 \lbrack 4,1 \rbrack & 56 & 256 & 400 & 440 & 584 & 336 & 144 & 240 & 96 \\
 \lbrack 5 \rbrack & 42 & 256 & 560 & 630 & 1310 & 780 & 680 & 1156 & 484 \\
 \lbrack 1,1,1,1,1,1 \rbrack & 120 & 120 & 0 & 0 & 0 & 0 & 0 & 0 & 0 \\
 \lbrack 2,1,1,1,1 \rbrack & 240 & 480 & 240 & 240 & 0 & 0 & 0 & 0 & 0 \\
 \lbrack 2,2,1,1 \rbrack & 360 & 1120 & 1120 & 1160 & 760 & 400 & 0 & 0 & 0 \\
 \lbrack 2,2,2 \rbrack & 432 & 1864 & 2864 & 3024 & 4024 & 2184 & 1000 & 1592 & 592 \\
 \lbrack 3,1,1,1 \rbrack & 300 & 960 & 960 & 1020 & 660 & 360 & 0 & 0 & 0 \\
 \lbrack 3,2,1 \rbrack & 360 & 1590 & 2460 & 2640 & 3510 & 1950 & 870 & 1410 & 540 \\
 \lbrack 3,3 \rbrack & 300 & 1746 & 3750 & 4098 & 8466 & 4842 & 4368 & 7206 & 2874 \\
 \lbrack 4,1,1 \rbrack & 280 & 1280 & 2000 & 2200 & 2920 & 1680 & 720 & 1200 & 480 \\
 \lbrack 4,2 \rbrack & 280 & 1648 & 3552 & 3912 & 8088 & 4672 & 4176 & 6944 & 2800 \\
 \lbrack 5,1 \rbrack & 210 & 1280 & 2800 & 3150 & 6550 & 3900 & 3400 & 5780 & 2420 \\
 \lbrack 6 \rbrack & 132 & 1024 & 2940 & 3360 & 9760 & 5952 & 8020 & 13832 & 5980 \\
 \lbrack 1,1,1,1,1,1,1 \rbrack & 720 & 720 & 0 & 0 & 0 & 0 & 0 & 0 & 0 \\
 \lbrack 2,1,1,1,1,1 \rbrack & 1440 & 2880 & 1440 & 1440 & 0 & 0 & 0 & 0 & 0 \\
 \lbrack 2,2,1,1,1 \rbrack & 2160 & 6720 & 6720 & 6960 & 4560 & 2400 & 0 & 0 & 0 \\
 \lbrack 2,2,2,1 \rbrack & 2592 & 11184 & 17184 & 18144 & 24144 & 13104 & 6000 & 9552 & 3552 \\
 \lbrack 3,1,1,1,1 \rbrack & 1800 & 5760 & 5760 & 6120 & 3960 & 2160 & 0 & 0 & 0 \\
 \lbrack 3,2,1,1 \rbrack & 2160 & 9540 & 14760 & 15840 & 21060 & 11700 & 5220 & 8460 & 3240 \\
 \lbrack 3,2,2 \rbrack & 2160 & 12336 & 26256 & 28392 & 58416 & 32880 & 30024 & 48984 & 19128 \\
 \lbrack 3,3,1 \rbrack & 1800 & 10476 & 22500 & 24588 & 50796 & 29052 & 26208 & 43236 & 17244 \\
 \lbrack 4,1,1,1 \rbrack & 1680 & 7680 & 12000 & 13200 & 17520 & 10080 & 4320 & 7200 & 2880 \\
 \lbrack 4,2,1 \rbrack & 1680 & 9888 & 21312 & 23472 & 48528 & 28032 & 25056 & 41664 & 16800 \\
 \lbrack 4,3 \rbrack & 1200 & 8856 & 24840 & 27600 & 79176 & 46440 & 64152 & 107616 & 44400 \\
 \lbrack 5,1,1 \rbrack & 1260 & 7680 & 16800 & 18900 & 39300 & 23400 & 20400 & 34680 & 14520 \\
 \lbrack 5,2 \rbrack & 1080 & 8100 & 22860 & 25660 & 73800 & 43860 & 60000 & 101620 & 42540 \\
 \lbrack 6,1 \rbrack & 792 & 6144 & 17640 & 20160 & 58560 & 35712 & 48120 & 82992 & 35880 \\
 \lbrack 7 \rbrack & 429 & 4096 & 14784 & 17094 & 64960 & 40320 & 74935 & 130655 & 57841 \\
 \lbrack 1,1,1,1,1,1,1,1 \rbrack & 5040 & 5040 & 0 & 0 & 0 & 0 & 0 & 0 & 0 \\
 \lbrack 2,1,1,1,1,1,1 \rbrack & 10080 & 20160 & 10080 & 10080 & 0 & 0 & 0 & 0 & 0 \\
 \lbrack 2,2,1,1,1,1 \rbrack & 15120 & 47040 & 47040 & 48720 & 31920 & 16800 & 0 & 0 & 0 \\
 \lbrack 2,2,2,1,1 \rbrack & 18144 & 78288 & 120288 & 127008 & 169008 & 91728 & 42000 & 66864 & 24864 \\
 \lbrack 2,2,2,2 \rbrack & 18144 & 101568 & 214272 & 229008 & 469440 & 259776 & 240432 & 387504 & 148080 \\
 \lbrack 3,1,1,1,1,1 \rbrack & 12600 & 40320 & 40320 & 42840 & 27720 & 15120 & 0 & 0 & 0 \\
 \lbrack 3,2,1,1,1 \rbrack & 15120 & 66780 & 103320 & 110880 & 147420 & 81900 & 36540 & 59220 & 22680 \\
 \lbrack 3,2,2,1 \rbrack & 15120 & 86352 & 183792 & 198744 & 408912 & 230160 & 210168 & 342888 & 133896 \\
 \lbrack 3,3,1,1 \rbrack & 12600 & 73332 & 157500 & 172116 & 355572 & 203364 & 183456 & 302652 & 120708 \\
 \lbrack 3,3,2 \rbrack & 10800 & 77688 & 214992 & 235728 & 670824 & 385920 & 538920 & 891864 & 358992 \\
 \lbrack 4,1,1,1,1 \rbrack & 11760 & 53760 & 84000 & 92400 & 122640 & 70560 & 30240 & 50400 & 20160 \\
 \lbrack 4,2,1,1 \rbrack & 11760 & 69216 & 149184 & 164304 & 339696 & 196224 & 175392 & 291648 & 117600 \\
 \lbrack 4,2,2 \rbrack & 10080 & 73248 & 203520 & 224560 & 640176 & 371328 & 515472 & 858160 & 348624 \\
 \lbrack 4,3,1 \rbrack & 8400 & 61992 & 173880 & 193200 & 554232 & 325080 & 449064 & 753312 & 310800 \\
 \lbrack 4,4 \rbrack & 4900 & 44224 & 155152 & 174152 & 650384 & 387616 & 735260 & 1246220 & 525700 \\
 \lbrack 5,1,1,1 \rbrack & 8820 & 53760 & 117600 & 132300 & 275100 & 163800 & 142800 & 242760 & 101640 \\
 \lbrack 5,2,1 \rbrack & 7560 & 56700 & 160020 & 179620 & 516600 & 307020 & 420000 & 711340 & 297780 \\
 \lbrack 5,3 \rbrack & 4725 & 42900 & 150840 & 169950 & 635280 & 380400 & 719175 & 1223055 & 518505 \\
 \lbrack 6,1,1 \rbrack & 5544 & 43008 & 123480 & 141120 & 409920 & 249984 & 336840 & 580944 & 251160 \\
 \lbrack 6,2 \rbrack & 4158 & 38448 & 136272 & 155172 & 582480 & 353664 & 662970 & 1138458 & 490230 \\
 \lbrack 7,1 \rbrack & 3003 & 28672 & 103488 & 119658 & 454720 & 282240 & 524545 & 914585 & 404887 \\
 \lbrack 8 \rbrack & 1430 & 16384 & 72072 & 84084 & 400384 & 251904 & 605770 & 1065218 & 480238 \\
 \lbrack 1,1,1,1,1,1,1,1,1 \rbrack & 40320 & 40320 & 0 & 0 & 0 & 0 & 0 & 0 & 0 \\
 \lbrack 2,1,1,1,1,1,1,1 \rbrack & 80640 & 161280 & 80640 & 80640 & 0 & 0 & 0 & 0 & 0 \\
 \lbrack 2,2,1,1,1,1,1 \rbrack & 120960 & 376320 & 376320 & 389760 & 255360 & 134400 & 0 & 0 & 0 \\
 \lbrack 2,2,2,1,1,1 \rbrack & 145152 & 626304 & 962304 & 1016064 & 1352064 & 733824 & 336000 & 534912 & 198912 \\
 \lbrack 2,2,2,2,1 \rbrack & 145152 & 812544 & 1714176 & 1832064 & 3755520 & 2078208 & 1923456 & 3100032 & 1184640 \\
 \lbrack 3,1,1,1,1,1,1 \rbrack & 100800 & 322560 & 322560 & 342720 & 221760 & 120960 & 0 & 0 & 0 \\
 \lbrack 3,2,1,1,1,1 \rbrack & 120960 & 534240 & 826560 & 887040 & 1179360 & 655200 & 292320 & 473760 & 181440 \\
 \lbrack 3,2,2,1,1 \rbrack & 120960 & 690816 & 1470336 & 1589952 & 3271296 & 1841280 & 1681344 & 2743104 & 1071168 \\
 \lbrack 3,2,2,2 \rbrack & 103680 & 733440 & 2010816 & 2186784 & 6186720 & 3515472 & 4940256 & 8105856 & 3209952 \\
 \lbrack 3,3,1,1,1 \rbrack & 100800 & 586656 & 1260000 & 1376928 & 2844576 & 1626912 & 1467648 & 2421216 & 965664 \\
 \lbrack 3,3,2,1 \rbrack & 86400 & 621504 & 1719936 & 1885824 & 5366592 & 3087360 & 4311360 & 7134912 & 2871936 \\
 \lbrack 3,3,3 \rbrack & 54000 & 472608 & 1628928 & 1802016 & 6648912 & 3879576 & 7418016 & 12391128 & 5088024 \\
 \lbrack 4,1,1,1,1,1 \rbrack & 94080 & 430080 & 672000 & 739200 & 981120 & 564480 & 241920 & 403200 & 161280 \\
 \lbrack 4,2,1,1,1 \rbrack & 94080 & 553728 & 1193472 & 1314432 & 2717568 & 1569792 & 1403136 & 2333184 & 940800 \\
 \lbrack 4,2,2,1 \rbrack & 80640 & 585984 & 1628160 & 1796480 & 5121408 & 2970624 & 4123776 & 6865280 & 2788992 \\
 \lbrack 4,3,1,1 \rbrack & 67200 & 495936 & 1391040 & 1545600 & 4433856 & 2600640 & 3592512 & 6026496 & 2486400 \\
 \lbrack 4,3,2 \rbrack & 50400 & 445056 & 1540560 & 1712640 & 6334896 & 3721008 & 7089072 & 11896128 & 4922880 \\
 \lbrack 4,4,1 \rbrack & 39200 & 353792 & 1241216 & 1393216 & 5203072 & 3100928 & 5882080 & 9969760 & 4205600 \\
 \lbrack 5,1,1,1,1 \rbrack & 70560 & 430080 & 940800 & 1058400 & 2200800 & 1310400 & 1142400 & 1942080 & 813120 \\
 \lbrack 5,2,1,1 \rbrack & 60480 & 453600 & 1280160 & 1436960 & 4132800 & 2456160 & 3360000 & 5690720 & 2382240 \\
 \lbrack 5,2,2 \rbrack & 45360 & 406560 & 1416480 & 1588400 & 5895440 & 3502400 & 6626400 & 11208320 & 4697760 \\
 \lbrack 5,3,1 \rbrack & 37800 & 343200 & 1206720 & 1359600 & 5082240 & 3043200 & 5753400 & 9784440 & 4148040 \\
 \lbrack 5,4 \rbrack & 19600 & 211780 & 902720 & 1023120 & 4770920 & 2885080 & 7045680 & 12057360 & 5186080 \\
 \lbrack 6,1,1,1 \rbrack & 44352 & 344064 & 987840 & 1128960 & 3279360 & 1999872 & 2694720 & 4647552 & 2009280 \\
 \lbrack 6,2,1 \rbrack & 33264 & 307584 & 1090176 & 1241376 & 4659840 & 2829312 & 5303760 & 9107664 & 3921840 \\
 \lbrack 6,3 \rbrack & 18480 & 201546 & 862344 & 982800 & 4591308 & 2796012 & 6797880 & 11690064 & 5066064 \\
 \lbrack 7,1,1 \rbrack & 24024 & 229376 & 827904 & 957264 & 3637760 & 2257920 & 4196360 & 7316680 & 3239096 \\
 \lbrack 7,2 \rbrack & 16016 & 178038 & 768824 & 885248 & 4158420 & 2567124 & 6198920 & 10759616 & 4735920 \\
 \lbrack 8,1 \rbrack & 11440 & 131072 & 576576 & 672672 & 3203072 & 2015232 & 4846160 & 8521744 & 3841904 \\
 \lbrack 9 \rbrack & 4862 & 65536 & 343200 & 403260 & 2333184 & 1483776 & 4429282 & 7841834 & 3586726 \\
 \lbrack 1,1,1,1,1,1,1,1,1,1 \rbrack & 362880 & 362880 & 0 & 0 & 0 & 0 & 0 & 0 & 0 \\
 \lbrack 2,1,1,1,1,1,1,1,1 \rbrack & 725760 & 1451520 & 725760 & 725760 & 0 & 0 & 0 & 0 & 0 \\
 \lbrack 2,2,1,1,1,1,1,1 \rbrack & 1088640 & 3386880 & 3386880 & 3507840 & 2298240 & 1209600 & 0 & 0 & 0 \\
 \lbrack 2,2,2,1,1,1,1 \rbrack & 1306368 & 5636736 & 8660736 & 9144576 & 12168576 & 6604416 & 3024000 & 4814208 & 1790208 \\
 \lbrack 2,2,2,2,1,1 \rbrack & 1306368 & 7312896 & 15427584 & 16488576 & 33799680 & 18703872 & 17311104 & 27900288 & 10661760 \\
 \lbrack 2,2,2,2,2 \rbrack & 1119744 & 7782912 & 21139968 & 22783488 & 64097280 & 35944320 & 50880000 & 82714368 & 32205312 \\
 \lbrack 3,1,1,1,1,1,1,1 \rbrack & 907200 & 2903040 & 2903040 & 3084480 & 1995840 & 1088640 & 0 & 0 & 0 \\
 \lbrack 3,2,1,1,1,1,1 \rbrack & 1088640 & 4808160 & 7439040 & 7983360 & 10614240 & 5896800 & 2630880 & 4263840 & 1632960 \\
 \lbrack 3,2,2,1,1,1 \rbrack & 1088640 & 6217344 & 13233024 & 14309568 & 29441664 & 16571520 & 15132096 & 24687936 & 9640512 \\
 \lbrack 3,2,2,2,1 \rbrack & 933120 & 6600960 & 18097344 & 19681056 & 55680480 & 31639248 & 44462304 & 72952704 & 28889568 \\
 \lbrack 3,3,1,1,1,1 \rbrack & 907200 & 5279904 & 11340000 & 12392352 & 25601184 & 14642208 & 13208832 & 21790944 & 8690976 \\
 \lbrack 3,3,2,1,1 \rbrack & 777600 & 5593536 & 15479424 & 16972416 & 48299328 & 27786240 & 38802240 & 64214208 & 25847424 \\
 \lbrack 3,3,2,2 \rbrack & 583200 & 5031360 & 17184096 & 18894528 & 69251904 & 40031568 & 76726368 & 127336752 & 51631920 \\
 \lbrack 3,3,3,1 \rbrack & 486000 & 4253472 & 14660352 & 16218144 & 59840208 & 34916184 & 66762144 & 111520152 & 45792216 \\
 \lbrack 4,1,1,1,1,1,1 \rbrack & 846720 & 3870720 & 6048000 & 6652800 & 8830080 & 5080320 & 2177280 & 3628800 & 1451520 \\
 \lbrack 4,2,1,1,1,1 \rbrack & 846720 & 4983552 & 10741248 & 11829888 & 24458112 & 14128128 & 12628224 & 20998656 & 8467200 \\
 \lbrack 4,2,2,1,1 \rbrack & 725760 & 5273856 & 14653440 & 16168320 & 46092672 & 26735616 & 37113984 & 61787520 & 25100928 \\
 \lbrack 4,2,2,2 \rbrack & 544320 & 4739712 & 16256448 & 17969088 & 66013440 & 38427840 & 73352256 & 122332416 & 50001792 \\
 \lbrack 4,3,1,1,1 \rbrack & 604800 & 4463424 & 12519360 & 13910400 & 39904704 & 23405760 & 32332608 & 54238464 & 22377600 \\
 \lbrack 4,3,2,1 \rbrack & 453600 & 4005504 & 13865040 & 15413760 & 57014064 & 33489072 & 63801648 & 107065152 & 44305920 \\
 \lbrack 4,3,3 \rbrack & 252000 & 2636676 & 11020320 & 12314880 & 56606616 & 33536232 & 82287504 & 138861504 & 58192272 \\
 \lbrack 4,4,1,1 \rbrack & 352800 & 3184128 & 11170944 & 12538944 & 46827648 & 27908352 & 52938720 & 89727840 & 37850400 \\
 \lbrack 4,4,2 \rbrack & 235200 & 2480224 & 10411776 & 11678976 & 53842368 & 32068928 & 78535744 & 133006400 & 56101696 \\
 \lbrack 5,1,1,1,1,1 \rbrack & 635040 & 3870720 & 8467200 & 9525600 & 19807200 & 11793600 & 10281600 & 17478720 & 7318080 \\
 \lbrack 5,2,1,1,1 \rbrack & 544320 & 4082400 & 11521440 & 12932640 & 37195200 & 22105440 & 30240000 & 51216480 & 21440160 \\
 \lbrack 5,2,2,1 \rbrack & 408240 & 3659040 & 12748320 & 14295600 & 53058960 & 31521600 & 59637600 & 100874880 & 42279840 \\
 \lbrack 5,3,1,1 \rbrack & 340200 & 3088800 & 10860480 & 12236400 & 45740160 & 27388800 & 51780600 & 88059960 & 37332360 \\
 \lbrack 5,3,2 \rbrack & 226800 & 2405070 & 10119720 & 11388960 & 52566420 & 31440900 & 76797720 & 130442160 & 55271040 \\
 \lbrack 5,4,1 \rbrack & 176400 & 1906020 & 8124480 & 9208080 & 42938280 & 25965720 & 63411120 & 108516240 & 46674720 \\
 \lbrack 5,5 \rbrack & 79380 & 1005050 & 5083100 & 5798500 & 32749300 & 20004300 & 60464300 & 104157060 & 45424140 \\
 \lbrack 6,1,1,1,1 \rbrack & 399168 & 3096576 & 8890560 & 10160640 & 29514240 & 17998848 & 24252480 & 41827968 & 18083520 \\
 \lbrack 6,2,1,1 \rbrack & 299376 & 2768256 & 9811584 & 11172384 & 41938560 & 25463808 & 47733840 & 81968976 & 35296560 \\
 \lbrack 6,2,2 \rbrack & 199584 & 2153016 & 9133152 & 10375008 & 48115920 & 29136144 & 70716480 & 121131456 & 52049952 \\
 \lbrack 6,3,1 \rbrack & 166320 & 1813914 & 7761096 & 8845200 & 41321772 & 25164108 & 61180920 & 105210576 & 45594576 \\
 \lbrack 6,4 \rbrack & 77616 & 986352 & 4995864 & 5711544 & 32280768 & 19772256 & 59659392 & 102964848 & 45034176 \\
 \lbrack 7,1,1,1 \rbrack & 216216 & 2064384 & 7451136 & 8615376 & 32739840 & 20321280 & 37767240 & 65850120 & 29151864 \\
 \lbrack 7,2,1 \rbrack & 144144 & 1602342 & 6919416 & 7967232 & 37425780 & 23104116 & 55790280 & 96836544 & 42623280 \\
 \lbrack 7,3 \rbrack & 72072 & 926226 & 4714458 & 5424090 & 30737700 & 18980220 & 57008700 & 98934528 & 43654548 \\
 \lbrack 8,1,1 \rbrack & 102960 & 1179648 & 5189184 & 6054048 & 28827648 & 18137088 & 43615440 & 76695696 & 34577136 \\
 \lbrack 8,2 \rbrack & 61776 & 809312 & 4160896 & 4833568 & 27570112 & 17246016 & 51532720 & 90214768 & 40408080 \\
 \lbrack 9,1 \rbrack & 43758 & 589824 & 3088800 & 3629340 & 20998656 & 13353984 & 39863538 & 70576506 & 32280534 \\
 \lbrack 10 \rbrack & 16796 & 262144 & 1604460 & 1896180 & 13025280 & 8355840 & 30067180 & 53529476 & 24771604 \\
\hline
\end{longtblr}
}


\bibliographystyle{utphys}
\bibliography{sample.bib}{}

\end{document}